\documentclass[11pt,journal,draftcls,onecolumn,peerreviewca]{IEEEtran}
\usepackage{amsfonts}
\usepackage{amsmath}
\usepackage{amssymb}
\usepackage{bm}
\usepackage{xcolor,graphicx,float}
\usepackage{color, soul}
\usepackage{stfloats}
\usepackage[numbers,sort&compress]{natbib}
\usepackage[amsmath,thmmarks]{ntheorem}
\usepackage{theorem}
\usepackage{algorithm}
\usepackage{algorithmic}
\usepackage{graphicx}
\usepackage{subfigure}
\usepackage{pict2e}
\usepackage{comment}
\usepackage{makecell}

\newtheorem{remark}{Remark}

\newtheorem{proposition}{Proposition}

\theoremheaderfont{\sc}\theorembodyfont{\upshape}
\theoremstyle{nonumberplain}
\theoremseparator{}
\theoremsymbol{\rule{1ex}{1ex}}

\newtheorem{proof}{Proof}

\usepackage[utf8]{inputenc}
\begin{document}
\title{CFAR based NOMP for Line Spectral Estimation and Detection}
\author{Menghuai Xu, Jiang~Zhu, Jun Fang, Ning Zhang and Zhiwei Xu
\thanks{Menghuai Xu, Jiang Zhu and Zhiwei Xu are with the engineering research center of oceanic sensing technology and equipment, Ministry of Education, Ocean College, Zhejiang University, No.1 Zheda Road, Zhoushan, 316021, China (email: \{menghuaixu, jiangzhu16, xuzw\}@zju.edu.cn). Jun Fang is with the National Key Laboratory on Communications, University of Electronic Science and Technology of China, 2006 Xiyuan Avenue, Chengdu, Sichuan, 611731, China (email: JunFang@uestc.edu.cn). Ning Zhang is with the Nanjing Marine Radar Institute, Nanjing, China (email: zhangn\_ee@163.com). The corresponding author is Jiang Zhu (email: jiangzhu16@zju.edu.cn). }}

\maketitle
\begin{abstract}
The line spectrum estimation problem is considered in this paper. We propose a CFAR-based Newtonized OMP (NOMP-CFAR) method which can maintain a desired false alarm rate without the knowledge of the noise variance. The NOMP-CFAR consists of two steps, namely, an initialization step and a detection step.
In the initialization step, NOMP is employed to obtain candidate sinusoidal components. In the detection step, CFAR detector is applied to detect each candidate frequency, and remove the most unlikely frequency component.
Then, the Newton refinements are used to refine the remaining parameters. The relationship between the false alarm rate and the required threshold is established.
By comparing with the NOMP, NOMP-CFAR has only $1$ dB performance loss in additive white Gaussian noise scenario with false alarm probability $10^{-2}$ and detection probability $0.8$ without knowledge of noise variance.
For varied noise variance scenario, NOMP-CFAR still preserves its CFAR property, while NOMP violates the CFAR.
Besides, real experiments are also conducted to demonstrate the detection performance of NOMP-CFAR, compared to CFAR and NOMP.
\end{abstract}
{\bf keywords}: Newtonized orthogonal matching pursuit (NOMP), CFAR, gridless compressed sening, line spectral estimation.
\section{Introduction}

Line spectrum estimation and detection are fundamental problems in signal processing fields, as they arise in various applications such as radar and sonar target estimation and detection \cite{Kayest, Kaydet, Zhang2019}, direction of arrival estimation \cite{yangzaibook}, situational awareness in automotive millimeter wave radar \cite{High-Performance Automotive Radar}, vital signs identification \cite{VDS} and so on.

The classical fast Fourier transform (FFT) based approach is often adopted to obtain the spectrum \cite{DFT-based Estimation}. Then CFAR detection is implemented to perform target detection. However, FFT is susceptible to "off-grid" effects \cite{Modelmismatch}: The signal from the target leaks into several points in the discrete Fourier transform (DFT) grid, unless it lies exactly on the DFT grid. Besides, this technique also suffers intertarget interference, i.e., the weak target is hard to detect given that the frequency of the weak target is close to that of a strong target, which limits the resolution capability. Later, subspace methods such as MUSIC \cite{MUSIC} and ESPRIT \cite{ESPRIT} which exploit the low-rank structure of the auto-correlation matrix are proposed. The capability of resolving multiple closely-spaced frequencies is enhanced especially at high signal to noise ratio (SNR). However, their performance degrades at medium and low SNRs.

More recent techniques using compressed sensing (CS) based methods exploit the sparse structure of the line spectrum in the frequency domain. By discretezing the frequency onto a finite set of grids, on-grid methods such as orthogonal matching pursuit (OMP) \cite{OMP},  least absolute shrinkage and selection operator (Lasso) \cite{Lasso1, Lasso2} are used to solve the line spectrum estimation problems. However, the frequencies are continuous parameters and they can not lie on the grids exactly. As a consequence, on-grid methods suffer from mismatch issues \cite{Modelmismatch}. To overcome the model mismatch, off-grid and grid-less methods are proposed such as RELAX \cite{RELAXAngle, RELAXIMP}, iterative reweighted approach (IRA) \cite{Fang}, variational line spectra estimation (VALSE) \cite{VALSE}, atomic norm soft thresholding (AST) \cite{AST}, Newtonized OMP (NOMP) \cite{MadhowNOMP} and so on. The frequency estimation accuracy of these off-grid and grid-less methods are better than the classical methods such as MUSIC and on-grid methods. However, the computation complexity of IRA and AST are high. For VALSE, it automatically estimates the noise variance, model order and other nuisance parameters. However, it is numerically found that when the number of targets is large, VALSE tend to output spurious components leading to large false alarms. For NOMP, it provides the constant false alarm rate (CFAR) based termination criterion by calculating the threshold with the knowledge of noise variance and a specified probability of false alarm, which is similar to the standard threshold detection. In practice, the interference levels often vary over time. Besides, the standard threshold is sensitive to the changes of the noise variance, and small errors in setting the threshold will have major impacts on radar performance \cite{FundamentalsRadarSP}. Therefore, it is of vital importance to incorporate the CFAR detector into the CS based algorithm such that both the advantages of the CS and CFAR are preserved.

The CFAR detector is incorporated into the NOMP algorithm named as NOMP-CFAR, so that the integrated detector threshold can be adjusted to maintain the desired false alarm rate. Compared to \cite{MadhowNOMP} using CFAR with the knowledge of noise variance similar to standard threshold detection, NOMP-CFAR estimates the interference level from data in real time to main a CFAR even in varied environments. To preserve the super resolution performance of NOMP-CFAR and avoid the target masking effects \cite{FundamentalsRadarSP}, NOMP-CFAR is divided into two steps: Initialization and model order estimation (MOE) or detection step. For the initialization, NOMP is used to obtain a maximum possible number of spectral as a candidate set, which is beneficial for closely spaced weak target detection. For the MOE, the CFAR detector is incorporated and output a soft quantity $\Delta_k$ (\ref{deltadef}), representing the level in which the amplitude of detected frequency exceeds its corresponding threshold. This is different from the conventional CFAR detector as it outputs a binary decision. The candidate frequency is preserved or removed in the candidate set based on the sign of $\Delta_k$. In each iteration, the most unlikely frequency corresponding to the maximum negative $\Delta_k$ is removed. Then the Newton refinements are conducted to improve the estimation accuracy of the parameters of the candidate frequencies.

In this paper, motivated by the high estimation accuracy and low computation of NOMP and its excellent performance in the application fields \cite{mnomp, Guptammw, Marzice, HanTwc, Mamandipoorimaging}, the CFAR is incorporated into the NOMP and NOMP-CFAR is proposed. The main contribution of this work can be summarized in the following three aspects:
\begin{enumerate}
  \item The NOMP-CFAR algorithm is designed, which inherits both the CFAR and super resolution advantages of the CFAR and NOMP. Incorporating the CFAR approach into the NOMP is nontrival as CFAR outputs a soft quantity instead of a simple binary decision. Besides, an initialization step via the NOMP approach is also provided to improve the detection probability, as shown in Subsection \ref{Benefits of Initialization}. Several implement details are also taken into consideration to enhance the performance of NOMP-CFAR. For further details, please refer to Section \ref{CFAR based MDNOMP}.
  \item The performances of NOMP-CFAR in terms of false alarm probability and detection probability are analyzed for the cell averaging (CA) CFAR. Specifically, the relationship between the false alarm rate and the required threshold multiplier is established, and insights into the relationship between CFAR and NOMP are also revealed. The detection probability by ignoring inter-sinusoid interference is derived.
  \item The NOMP-CFAR is also extended to deal with the compressive measurement scenario and the multiple measurement vector setting (MMV).
  \item The performance of NOMP-CFAR algorithm is compared with the NOMP and CFAR detector in numerical simulations and real data experiments. Numerically, it is shown that both NOMP-CFAR and NOMP provide a CFAR for additive white Gaussian noise (AWGN) scenario. For the detection probability $0.76$, NOMP-CFAR yields $1$ dB performance loss, compared to NOMP. For varied noise variance scenario, NOMP-CFAR still preserves its CFAR property, while NOMP violates the CFAR. In real experiments, the detection probability of NOMP-CFAR is higher than that of CFAR and the false alarm is smaller than CFAR.
\end{enumerate}

The outline of the paper is summarized as follows: Section \ref{Problem Setup} introduces the LSE model.
In Section \ref{CFAR based MDNOMP}, NOMP-CFAR algorithm is developed. In addition, NOMP-CFAR is also extended to deal with the compressive measurement model and MMV model. In Section \ref{Numerical Simulation}, numerical experiments are conducted to evaluate the performance of NOMP-CFAR algorithm in terms of estimation accuracy, false alarm probability and detection probability. The real data experiments are shown in Section \ref{Real Experiment}.
And Section \ref{Conclusion} concludes the paper.

Notation:
Let $\mathbf{a}$, $\mathbf{A}$, $\mathcal{A}$ denote the vector, matrix and tensor,  respectively.
For a D-dimension tensor $\mathcal{Y}$, let ${\mathcal{Y}}_{\mathbf{k}}$ denote the ${\mathbf k}$th element of ${\mathcal{Y}}$, where ${\mathbf{k}} \in \mathbb{N}^{D}$ and $\mathbb{N}$ is the set of all natural numbers, and $\tilde{\mathcal{Y}}$ denotes its spectrum. For a set $\mathcal N$, let $|\mathcal N|$ denote its cardinality. For any two frequencies $\omega_l$ and $\omega_k$, the wrap-around distance is defined as ${\rm dist}(\omega_l, \omega_k) \triangleq {\rm min}_{a \in \mathbb{Z}}\left|\omega_k - \omega_l + 2 \pi a\right|$ and $\mathbb{Z}$ is the set of all integers.

\section{Problem Setup} \label{Problem Setup}
Consider a multidimensional line spectral ${\mathcal Z}\in{\mathbb R}^{N_1\times N_2\cdots\times N_d\cdots\times N_D}$ described as
\begin{align}\label{outprodmodel}
{\mathcal Z}=\sum_{k = 1}^{K} x_k \mathbf{a}_{N_1}({\omega}_{1, k}) \circ \mathbf{a}_{N_2}({\omega}_{2, k})\cdots \mathbf{a}_{N_d}({\omega}_{d, k})\cdots\circ \mathbf{a}_{N_D}({\omega}_{D, k}),
\end{align}
where $D$ denotes the dimension (usually $D=1,2,3$), $K$ denotes the number of spectral, $\circ$ denotes the outer product of two matrices, $\mathbf{a}_P(\omega)$ is an array steering vector defined as
\begin{align}\label{the definition of line spectrum}
\mathbf{a}_P(\omega) = \begin{bmatrix}
1, & {\text e}^{\text{j} \omega}, & \cdots, & {\text e}^{\text{j} (P - 1) \omega}
\end{bmatrix}^{\text T}.
\end{align}
The line spectral ${\mathcal Z}$ corrupted by additive noise ${\mathcal \epsilon}$ is described as
\begin{align}
\mathcal{Y}={\mathcal Z}+\mathcal{\varepsilon},
\end{align}
where $\mathcal{Y}$ is the noisy measurement, $\mathcal{\varepsilon}$ denotes the AWGN and its ${\mathbf n}\triangleq[n_1,\cdots,n_d,\cdots,n_D]$th element follows  $\mathcal{\varepsilon}_{{\mathbf n}}\sim {\mathcal {CN}}(0,\sigma^2)$.

By defining ${\boldsymbol \omega}_k = [{\omega}_{1, k}, {\omega}_{2, k}, \cdots,{\omega}_{D, k}]$ and $\mathcal{A}({\boldsymbol \omega}_k) = \mathbf{a}_{N_1}({\omega}_{1, k}) \circ \mathbf{a}_{N_2}({\omega}_{2, k})\cdots \mathbf{a}_{N_d}({\omega}_{d, k})\cdots\circ \mathbf{a}_{N_D}({\omega}_{D, k})$, model (\ref{outprodmodel}) can be simplified as
\begin{align}\label{simple multiple tensor model}
\mathcal{Y} = \sum_{k = 1}^{K} x_k \mathcal{A}({\boldsymbol \omega}_k) + \mathcal{\varepsilon},
\end{align}
Vectorizing $\mathcal{Y}$ yields
\begin{align}\label{multiple vector model}
\mathbf{y} = \sum_{k = 1}^{K} x_k {\mathbf a}_N({\boldsymbol\omega}_k)  + {\boldsymbol\varepsilon},
\end{align}
where $\mathbf{y}={\rm vec}(\mathcal{Y})$,
\begin{align}\label{defineaNomega}
{\mathbf a}_N({\boldsymbol\omega}_k)\triangleq \mathbf{a}_{N_D}({\omega}_{D, k})\otimes\cdots\otimes \mathbf{a}_{N_d}({\omega}_{d, k})\otimes\cdots\otimes\mathbf{a}_{N_2}({\omega}_{2, k})\otimes\mathbf{a}_{N_1}({\omega}_{1, k}),
\end{align}
$\otimes$ denotes the Kronecker product, ${\boldsymbol\varepsilon}={\rm vec}(\mathcal{\varepsilon})$ and $\boldsymbol \varepsilon \sim \mathcal{CN}({\mathbf 0}, \sigma^2 \mathbf{I}_{N})$. By defining ${\mathbf A}$ as
\begin{align}\label{defineA}
\mathbf{A}&=
\begin{bmatrix}
\mathbf{a}_N({\boldsymbol\omega}_1), &\cdots,& \mathbf{a}_N({\boldsymbol\omega}_k), & \cdots, & \mathbf{a}_N({\boldsymbol\omega}_K) \\
\end{bmatrix},
\end{align}
model (\ref{multiple vector model}) is reformulated as a matrix form
\begin{align}\label{receive model}
\mathbf{y} = \mathbf{A} \mathbf{x} + \boldsymbol \varepsilon.
\end{align}

Consequently, the line spectra estimation and detection problem is to obtain $\hat{K}$ and the frequencies $\{(\hat{\omega}_{1,k}, \hat{\omega}_{2,k},\cdots, \hat{\omega}_{D,k})\}_{k=1}^{\hat{K}}$ without the knowledge of noise variance $\sigma^2$.

The SNR in dB of the $k$th target is defined as \cite{MadhowNOMP}
\begin{align}\label{defofSNR}
{\rm SNR}_k=10\log\frac{N |x_k|^2}{\sigma^2},
\end{align}
where $N=\prod\limits_{d=1}^{D}N_d$, which is the nominal SNR value in our simulations and is called ``integrated SNR''. An alternative definition of SNR is the ``per-sample'' SNR given by ${\rm SNR}_{{\rm sample},k}=10\log\frac{|x_k|^2}{\sigma^2}=10\log {\rm SNR}_k-10\log(N)$. It can be seen that the ``integrated SNR'' is the ``per-sample'' SNR plus the coherent integrated gain $10\log(N)$.

\section{NOMP-CFAR Algorithm}\label{CFAR based MDNOMP}
This section develops the NOMP-CFAR algorithm for target detection and estimation. The false alarm probability is often defined clearly for a binary hypothesis testing problem. For the LSE, an intuitively explanation of the false alarm probability is that for the false alarm probability $\bar{\rm P}_{\rm FA}$, the NOMP-CFAR algorithm generates about $N_{\rm MC}\bar{\rm P}_{\rm FA}$ false targets whose frequencies are not near the frequencies of the true targets in $N_{\rm MC}$ Monte Carlo trials. Firstly, the CFAR criterion is proposed and the threshold with the false alarm probability $\bar{\rm P}_{\rm FA}$ is provided. Secondly, the details of NOMP-CFAR which novelly combines the CFAR and NOMP is presented. Note that we have carefully design NOMP-CFAR to ensure that the measured false alarm probability is close to the nominal false alarm probability. Finally, NOMP-CFAR is extended to deal with both the compressive measurement model and an MMV model.
\subsection{CFAR Detector Design}\label{detdesign}
The stopping criterion proposed in \cite{MadhowNOMP} assumes that the noise variance is known and constant. This stopping criterion sets the threshold accurately to guarantee a specified probability of false alarm. In practice, the noise variance is unknown and can often be variable. To provide predictable detection and false alarm behaviour in realistic scenarios, CFAR detection or ``adaptive threshold detection '' is developed.

The idea of CFAR detector is to estimate the noise variance from the data in real time, so that the detection threshold can be adjusted to maintain the desired $\bar{\rm P}_{\rm FA}$. Below we present the details of the CFAR design.

We consider a binary hypothesis testing problem
\begin{subequations}\label{modelbht}
\begin{align}
&{\mathcal H}_0:{\mathcal Y}=\mathcal{\varepsilon},\\
&{\mathcal H}_1:{\mathcal Y}={x}\mathbf{a}_{N_1}({\omega}_{1, k}) \circ \mathbf{a}_{N_2}({\omega}_{2, k})\cdots \mathbf{a}_{N_d}({\omega}_{d, k})\cdots\circ \mathbf{a}_{N_D}({\omega}_{D, k})+\mathcal{\varepsilon},
\end{align}
\end{subequations}
and we choose between the null hypothesis ${\mathcal H}_0$ and the alternative hypothesis ${\mathcal H}_1$, where $\mathcal{\varepsilon}$ is the AWGN with its elements $\mathcal{\varepsilon}_{\mathbf n}$ being $\mathcal{\varepsilon}_{\mathbf n}\sim {\mathcal {CN}}({0},\sigma^2)$. With the knowledge of noise variance $\sigma^2$, a typical detector deciding ${\mathcal H}_1$ is
\begin{align}\label{GLRTnoiseaware}
T({{\mathcal Y}},\sigma^2)=\frac{|\tilde{\mathcal Y}_{\tilde{\mathbf n}_{\rm peak}}|^2}{\sigma^2}>\alpha',
\end{align}
where $\tilde{\mathcal Y}$ denotes the normalized $D$ dimensional DFT of ${\mathcal Y}$ defined as
\begin{align}\label{DdimDFT}
\tilde{\mathcal Y}_{\tilde{\mathbf n}}=\frac{1}{\sqrt{N}}\sum\limits_{n_1=0}^{N_1-1}{\rm e}^{-{\rm j}\frac{2\pi}{N_1}\tilde{n}_1n_1}\sum\limits_{n_2=0}^{N_2-1}{\rm e}^{-{\rm j}\frac{2\pi}{N_2}\tilde{n}_2n_2}\cdots\sum\limits_{n_D=0}^{N_D-1}{\rm e}^{-{\rm j}\frac{2\pi}{N_D}\tilde{n}_Dn_D}{\mathcal Y}_{\mathbf n},
\end{align}
${\tilde{\mathbf n}}_{\rm peak}$ is the peak localization of $|\tilde{\mathcal Y}|^2$ given by
\begin{align}\label{definekpeak}
{\tilde{\mathbf n}}_{\rm peak}=\underset{{\tilde{\mathbf n}}\in \tilde{\mathcal N}}{\operatorname{argmax}} ~|\tilde{\mathcal Y}_{\tilde{\mathbf n}}|^2
\end{align}
where $\tilde{\mathcal N}=\{\tilde{\mathbf n}\in{\mathbb N}^D,0\leq k_d\leq N_d-1,d=1,2,\cdots,D\}$.
In other words, the detector (\ref{GLRTnoiseaware}) decides that the signal is present if the peak value of the spectrum exceeds a threshold.
It is also worth noting that (\ref{GLRTnoiseaware}) can be evaluated efficiently through $D$ dimensional DFT.
Given the false alarm probability $\bar{\rm P}_{\rm FA}$, the threshold multiplier $\alpha'$ can be obtained by \cite{MadhowNOMP}
\begin{align}\label{the threshold of NOMP}
\alpha'_{\rm NOMP} = - \ln \left(1 - (1 - \bar{\rm P}_{\rm FA}) ^ {\frac{1}{N}}\right).
\end{align}

Motivated by the conventional CFAR approach, the noise variance $\sigma^2$ is estimated with the data samples as
\begin{align}\label{sigmaest}
\hat{\sigma}^2 = \frac{1}{N_r}\sum\limits_{{\tilde{\mathbf n}}\in {\mathcal{T}}_{{\tilde{\mathbf n}}_{\rm peak}}} \left|\tilde{\mathcal Y}_{\tilde{\mathbf n}}\right|^2,
\end{align}
where ${\tilde{\mathbf n}}_{\rm peak}$ (\ref{definekpeak}) is the index of the cell under test (CUT) $\tilde{\mathcal Y}_{{\tilde{\mathbf n}}_{\rm peak}}$, ${\mathcal{T}}_{{\tilde{\mathbf n}}_{\rm peak}}$ is the index set of the reference cells, $N_r$ denotes the number of reference cells and $N_r=|{\mathcal{T}}_{{\tilde{\mathbf n}}_{\rm peak}}|$. The reference cells are averaged to estimate the noise variance. In practice, we set guard cells between the reference cells and the CUT. The reason is that a frequency not on the DFT grid exactly might straddle frequency cells. In this case, the energy in the cell adjacent to $\tilde{\mathcal Y}_{{\tilde{\mathbf n}}_{\rm peak}}$ would contain both noise and signal energy. The extra energy from the signal would tend to raise the estimate of the noise variance, resulting in a higher threshold and a lower $\bar{\rm P}_{\rm FA}$ and $\bar{\rm P}_{\rm D}$ than intended. The details of selecting reference cells and guard cells can be referred to \cite{FundamentalsRadarSP}. While in multiple target scenario, the construction of reference cells from \cite{FundamentalsRadarSP} can not be directly borrowed. The reason is that there may some other signals lie on the reference cells. Once NOMP-CFAR estimates these signals, eliminates the effects of these signals, and detects a new signal, directly constructing the reference cells via the CFAR approach tends to make the estimate of the noise variance lower, resulting in a lower threshold and a higher $\bar{\rm P}_{\rm FA}$ than intended, see Subection \ref{AdCFARMDOMP} for further details. Consequently, the cells nearby the other targets should be excluded from the reference cells.

With $\hat{\sigma}^2$ (\ref{sigmaest}), the detector in the case of unknown noise variance is
\begin{align}\label{GLRT_noiseunknown}
T({{\mathcal Y}})=\frac{|\tilde{\mathcal Y}_{{\tilde{\mathbf n}}_{\rm peak}}|^2}{\hat{\sigma}^2}>\alpha.
\end{align}
Given this detector $T({{\mathcal Y}})$ (\ref{GLRT_noiseunknown}), its false alarm probability $\bar{\rm P}_{\rm FA}$ and detection probability $\bar{\rm P}_{\rm D}$ are analyzed. Firstly, we need to calculate the probability density function (PDF) of $\hat{\sigma}^2$ (\ref{sigmaest}). Note that the normalized DFT is a unitary matrix, and the DFT of ${\mathcal Y}$ under the null hypothesis ${\mathcal H}_0$ is still the AWGN, i.e.,
\begin{align}
p(\tilde{\mathcal Y}|{\mathcal H}_0)=\prod\limits_{\tilde{\mathbf n}}{\mathcal{CN}}(\tilde{\mathcal Y}_{\tilde{\mathbf n}};0,\sigma^2)
\end{align}
In general, the exact PDF of $\hat{\sigma}^2$ (\ref{sigmaest}) is hard to obtain as the reference cells depend on the peak localization ${\tilde{\mathbf n}}_{\rm peak}$ of the spectrum. Provided that the reference cells are not too near the peak, the PDF of $\tilde{\mathcal Y}_{\tilde{\mathbf n}}$, ${\tilde{\mathbf n}}\in {\mathcal{T}}_{{\tilde{\mathbf n}}_{\rm peak}}$ is approximated as a Gaussian distribution, i.e.,
\begin{align}\label{assum1}
p(\tilde{\mathcal Y}_{{\mathcal{T}}_{{\tilde{\mathbf n}}_{\rm peak}}}|{\mathcal H}_0)\approx\prod\limits_{\tilde{\mathbf n}\in {\mathcal{T}}_{{\tilde{\mathbf n}}_{\rm peak}}}{\mathcal{CN}}(\tilde{\mathcal Y}_{\tilde{\mathbf n}};0,\sigma^2).
\end{align}
Based on (\ref{assum1}), we proceed to calculate the average false alarm $\bar{\rm P}_{\rm FA}$ and the result is summarized as Proposition \ref{Propfalse}.
\begin{proposition}\label{Propfalse}
The average false alarm probability $\bar{\rm P}_{\rm FA}$ and the required threshold multiplier $\alpha$ for the detector (\ref{GLRT_noiseunknown}) is
\begin{align}\label{simplePoe}
\bar{\rm P}_{\rm FA} = 1 - \frac{1}{(N_r - 1)!} \int_{0}^{+ \infty} \left(1 - {\rm e}^{- \frac{\alpha}{N_r} x} \right)^{N} x^{N_r - 1} {\rm e}^{-x} {\rm d}x,
\end{align}
or
\begin{align}\label{resultPoe}
\bar{\rm P}_{\rm FA} &= 1 - \sum_{n = 0}^{N} (-1)^n \binom{N}{n} \frac{1}{(\frac{n \alpha}{N_r} + 1)^{N_r}} .
\end{align}
\end{proposition}
\begin{proof}
The basic idea is to obtain the false alarm probability ${\rm P}_{\rm FA}$ with $\hat{\sigma}^2$ fixed, and then averaging that with respect to the PDF of $\hat{\sigma}^2$ yields the desired false alarm probability $\bar{\rm P}_{\rm FA}$. Define the threshold $\hat{T}$ as
\begin{align}\label{distribution of That}
\hat{T} = \alpha\left(\frac{1}{N_r}\sum\limits_{{\tilde{\mathbf n}}\in {\mathcal{T}}_{{\tilde{\mathbf n}}_{\rm peak}}} \left|\tilde{\mathcal Y}_{\tilde{\mathbf n}}\right|^2\right).
\end{align}
Under the assumption (\ref{assum1}), the distribution of $\hat{T}$ is chi-square distribution with $2 N_r$ degrees of freedom and variance $\alpha \sigma^2 / 2$, i.e.,
\begin{align}\label{PDF of That}
p_{\hat{T}}(\hat{T}) = \frac{1}{(N_r - 1)!} \left(\frac{N_r}{\alpha \sigma^2}\right)^{N_r}\hat{T}^{N_r - 1} {\rm e}^{- \frac{N_r \hat{T}}{\alpha \sigma^2}}.
\end{align}
The false alarm event is declared if the peak of the spectrum $|\tilde{\mathcal Y}_{{\tilde{\mathbf n}}_{\rm peak}}|^2$ exceeds the threshold $\hat{T}$, i.e.,
\begin{align}\label{barPfacal}
\bar{\rm P}_{\rm FA}={\rm E}_{\hat{T}}\left[{\rm P}\left(|\tilde{\mathcal Y}_{{\tilde{\mathbf n}}_{\rm peak}}|^2\geq \hat{T}\right)\right],
\end{align}
where ${\rm E}_{\hat{T}}[\cdot]$ is taken with respect to the PDF $p_{\hat{T}}(T)$ (\ref{PDF of That}). With $\hat{T}$ being fixed, ${\rm P}\left(|\tilde{\mathcal Y}_{{\tilde{\mathbf n}}_{\rm peak}}|^2\geq \hat{T}\right)$ is
\begin{align}\label{CDFpeak}
{\rm P}\left(|\tilde{\mathcal Y}_{{\tilde{\mathbf n}}_{\rm peak}}|^2\geq \hat{T}\right)&={\rm P}\left(\underset{{\tilde{\mathbf n}}\in \tilde{\mathcal N}}{\operatorname{max}} ~|\tilde{\mathcal Y}_{\tilde{\mathbf n}}|^2\geq \hat{T}\right)\notag\\
&=1-{\rm P}\left(\underset{{\tilde{\mathbf n}}\in\tilde{\mathcal N}}{\operatorname{max}} ~|\tilde{\mathcal Y}_{\tilde{\mathbf n}}|^2\leq \hat{T}\right)\notag\\
&=1-{\rm P}\left(|\tilde{\mathcal Y}_{\tilde{\mathbf n}}|^2\leq \hat{T},{{\tilde{\mathbf n}}\in \tilde{\mathcal N}}\right)\notag\\
&\stackrel{a}=1-\prod\limits_{{{\tilde{\mathbf n}}\in \tilde{\mathcal N}}}{\rm P}\left(|\tilde{\mathcal Y}_{\tilde{\mathbf n}}|^2\leq \hat{T}\right),
\end{align}
where $\stackrel{a}=$ is due to the independence of $\tilde{\mathcal Y}_{\tilde{\mathbf n}}$, ${{\tilde{\mathbf n}}\in \tilde{\mathcal N}}$.
The PDF of $|\tilde{\mathcal Y}_{\tilde{\mathbf n}}|^2$ is chi-square distribution with 2 degrees of freedom and variance $\sigma^2 / 2$ with CDF
\begin{align}\label{CDF of Yk2}
{\rm P}\left(|\tilde{\mathcal Y}_{\tilde{\mathbf n}}|^2\leq \hat{T}\right) = 1 - {\rm e}^{- \hat{T} / \sigma^2}.
\end{align}
By substituting (\ref{CDFpeak}) in (\ref{barPfacal}), the average false alarm probability $\bar{\rm P}_{\rm FA}$ is
\begin{align}\label{detail of Poe}
\bar{\rm P}_{\rm FA} = 1 - \int_{0}^{+ \infty} \prod\limits_{{{\tilde{\mathbf n}}\in \tilde{\mathcal N}}} {\rm P} \left( \left| \tilde{\mathcal{Y}}_{\tilde{\mathbf{n}}} \right|^2 < \hat{T}\right) p_{\hat{T}}(\hat{T}) {\rm d}\hat{T}.
\end{align}
Substituting (\ref{CDF of Yk2}) and (\ref{PDF of That}) into (\ref{detail of Poe}) yields
\begin{align}
\bar{\rm P}_{\rm FA} = 1 - \int_{0}^{+ \infty} \left(1 - {\rm e}^{- \hat{T} / \sigma^2}\right)^{N} \frac{1}{(N_r - 1)!} \left(\frac{N_r}{\alpha \sigma^2}\right)^{N_r} \hat{T}^{N_r - 1} {\rm e}^{- \frac{N_r \hat{T}}{\alpha \sigma^2}} {\rm d}\hat{T},
\end{align}
which can be simplified as (\ref{simplePoe}) by doing a change of variable $x = \frac{N_r \hat{T}}{\alpha \sigma^2}$. Furthermore, by using the binomial expansion
\begin{align}\label{binomial expansion }
\left(1 - {\rm e}^{- \frac{\alpha}{N_r} x} \right)^{N} = \sum_{n = 0}^{N} \binom{N}{n} (-1)^n {\rm e}^{- \frac{n \alpha}{N_r} x},
\end{align}
(\ref{simplePoe}) can be simplified as
\begin{align}\label{binomial Poe}
\bar{\rm P}_{\rm FA} = 1 - \sum_{n = 0}^{N} \binom{N}{n} \frac{(-1)^n}{(N_r - 1)!} \int_{0}^{+ \infty} {\rm e}^{- (\frac{n \alpha}{N_r} + 1) x} x^{N_r - 1} {\rm d}x.
\end{align}
Utilizing $\int_{0}^{+ \infty} {\rm e}^{- (\beta + 1) x} x^{N_r - 1} {\rm d}x = \frac{(N_r - 1)!}{(\beta + 1)^{N_r}}$, (\ref{resultPoe}) is obtained.
\end{proof}

Note that we have provided two formulas (\ref{simplePoe}) and (\ref{resultPoe}) for $\bar{\rm P}_{\rm FA}$ with the required threshold multiplier $\alpha$.
Eq. (\ref{simplePoe}) is simple for calculating via numerical integration. In contrast, eq. (\ref{resultPoe}) is hard to calculate as the range of the terms in the summation can be very large. However, eq. (\ref{resultPoe}) indeed provides some insight into the relationship between the establishing result and the previous results, as shown in the following Remarks.
\begin{remark}
As $N_r\to\infty$, one has
\begin{align}\label{limitation of Nr}
 \lim_{N_r \to + \infty}  (1 + \frac{n \alpha}{N_r})^{-N_r} = {\rm e}^{- n \alpha}.
\end{align}
Substituting (\ref{limitation of Nr}) in (\ref{resultPoe}), $\bar{\rm P}_{\rm FA}$ can be simplified as
\begin{align}\label{resultPoe Nr is big}
\bar{\rm P}_{\rm FA} = 1 - \sum_{n = 0}^{N} (-1)^n \binom{N}{n} {\rm e}^{- n \alpha} = 1 - (1 - {\rm e}^{- \alpha})^N,
\end{align}
and $\alpha$ can be calculated as
\begin{align}\label{alpha of big Nr}
\alpha = -\ln \left(1 - (1 - \bar{\rm P}_{\rm FA}) ^ {\frac{1}{N}}\right).
\end{align}
Similar results, consistent with (\ref{alpha of big Nr}), are obtained in \cite{MadhowNOMP}, which
focuses on the noise variance aware case. This makes sense as $N_r\to\infty$, the noise variance estimate is exact and $\alpha=\alpha'$ with $\alpha'$ given in (\ref{the threshold of NOMP}) and $\bar{\rm P}_{\rm FA}={\rm P}_{\rm FA}$.
\end{remark}

\begin{remark}
When the required threshold multiplier $\alpha$ is large such that the false alarm probability $\bar{\rm P}_{\rm FA}$ is small, the first and second terms corresponding to $n=0$ and $n=1$ in the summation is much larger than the remaining terms, i.e., \begin{align}\label{constraintofDom}
\frac{(N - n)! n!}{(N - 1)!} \left(\frac{n \alpha + N_r}{\alpha + N_r} \right)^{N_r} \gg 1, n = 2, \cdots, N.
\end{align}
Consequently, (\ref{resultPoe}) can be approximated as
\begin{align}\label{resultPoeapp}
\bar{\rm P}_{\rm FA} \approx 1 - \sum_{n = 0}^{1} (-1)^n \binom{N}{n} \frac{1}{(\frac{n \alpha}{N_r} + 1)^{N_r}}={N}{(\frac{\alpha}{N_r} + 1)^{-N_r}}=N\bar{\rm P}_{\rm FA}({\rm bin}),
\end{align}
where $\bar{\rm P}_{\rm FA}({\rm bin})={({\alpha}/{N_r} + 1)^{-N_r}}$ is the average false alarm probability of traditional CA-CFAR method for each cell if we examine one bin.
Hence, $\bar{\rm P}_{\rm FA}$ increases approximately linearly with the number of bins examined.
Eq. (\ref{resultPoeapp}) demonstrates that when $\bar{\rm P}_{\rm FA}$ is very small, it is equal to the $N$ times of the average false alarm probability of traditional CA-CFAR method. The average number of false alarms among all the $N$ test cells is $N\bar{\rm P}_{\rm FA}({\rm bin})\ll 1$. Intuitively, when $N\bar{\rm P}_{\rm FA}({\rm bin})\ll 1$, the false alarm is dominated by the CUT corresponding to the peak of the spectrum with false alarm probability $\bar{\rm P}_{\rm FA}$. This implies the approximation (\ref{resultPoeapp}).
\end{remark}
\begin{remark}
The two formulas (\ref{simplePoe}) and (\ref{resultPoe}) are obtained with the CA-CFAR approach. Following the similar steps, one could establish the specified $\bar{\rm P}_{\rm FA}$ with the required threshold multiplier $\alpha$ for other CFAR approaches, such as the smallest-of cell-averaging CFAR (SOCA CFAR), greater-of cell-averaging CFAR (GOCA CFAR), order statistic CFAR (OS CFAR), etc \cite{FundamentalsRadarSP, RadarCFARRohling}. The OS-CFAR rank orders the reference window data samples $\tilde{\mathcal{Y}}_{\tilde{\mathbf n}}$ to form a new sequence in ascending numerical order, ${\tilde{\mathbf n}} \in {\mathcal{T}}_{{\tilde{\mathbf n}}_{\rm peak}}$.
The $r$th element of the ordered list is called the $r$th order statistic, denoted by $\tilde{\mathcal{Y}}_{(r)}$.
Thus the new sequence of the reference window data samples $\tilde{\mathcal{Y}}_{\tilde{\mathbf n}}$ can be denoted as $\left\{ \tilde{\mathcal{Y}}_{(1)}, \cdots,  \tilde{\mathcal{Y}}_{(N_r)} \right\}$.
In OS-CFAR, the $r$th order statistic is selected as representative of the noise variance and a threshold is set as a multiple of this value
\begin{align}\label{ThatofOSCFAR}
\hat{T} = \alpha_{\rm OS} \tilde{\mathcal{Y}}_{(r)}.
\end{align}
As shown in \cite{FundamentalsRadarSP}, the PDF of $\hat{T}$ is
\begin{align}\label{PDF of ThatinOS}
p_{\hat{T}}(\hat{T}) = \frac{r}{\alpha_{\rm OS}} \binom{N}{r} \left[{\rm e}^{- \hat{T} / \alpha_{\rm OS}} \right] ^ {N - r + 1} \left[ 1- {\rm e}^{- \hat{T} / \alpha_{\rm OS}} \right] ^ {r - 1}.
\end{align}
Substituting (\ref{PDF of ThatinOS}) in (\ref{detail of Poe}), the average false alarm rate $\bar{\rm P}_{\rm FA}$ of OS-CFAR can be obtained.
\end{remark}

Define
\begin{align}\label{deltadef}
\Delta = 10 \log\left(\frac{|\tilde{\mathcal{Y}}_{{\tilde{\mathbf n}}_{\rm peak}}|^2}{\hat{T}}\right).
\end{align}
where $\hat{T}$ is (\ref{distribution of That}). Note that $\Delta = 0$ implies that the integrated amplitude of the detected signal just cross the corresponding threshold. For $\Delta > 0$, $\Delta$ is the excess in which the amplitude of the detected signal above the corresponding threshold. While for $\Delta < 0$, $-\Delta$ is the excess in which the amplitude of the detected signal below the corresponding threshold. For the traditional CFAR detector, it makes a binary decision based on the sign of $\Delta$. In contrast, we provide a ``soft'' CFAR detector which outputs $\Delta$. The ``soft'' CFAR detector is summarized in Algorithm \ref{CFARdetector}.

\begin{algorithm}[ht]
\caption{CFAR detector}\label{CFARdetector}
\begin{algorithmic}[1]
\STATE For a given false alarm probability $\bar{\rm P}_{\rm FA}$, calculate the required threshold multiplier $\alpha$ for CA-CFAR via (\ref{simplePoe}). \\
\STATE Procedure: CFAR $({\mathcal Y}, \alpha)$ : \\
\STATE Compute $\tilde{\mathcal{Y}}$ via (\ref{DdimDFT}), \\
\STATE Find the peak location ${\tilde{\mathbf n}}_{\rm peak}$ via (\ref{definekpeak}), \\
\STATE Estimate the noise variance as $\hat{\sigma}^2$ (\ref{sigmaest}), \\
\STATE Compute $\Delta$ (\ref{deltadef}), \\
\STATE Return $\Delta$.
\end{algorithmic}
\end{algorithm}

It is also meaningful to obtain the detection probability $\bar{\rm P}_{\rm D}$ for the required threshold multiplier $\alpha$. Under the alternative hypothesis ${\mathcal H}_1$, we declare that a frequency $\boldsymbol \omega$ exists if the spectrum $\tilde{\mathcal Y}$ at the DFT frequency ${\boldsymbol \omega}_{\hat{\mathbf n}}$ exceeds its corresponding threshold, where the $d$th element of $\hat{\mathbf n}$ is
\begin{align}
\hat{n}_d = \underset{n_d}{\operatorname{argmin}}|\omega_d - \bar{\omega}_{n_d}|,
\end{align}
where $\bar{\omega}_{n_d} \triangleq (n_d - 1) 2 \pi / N_d$ denotes the $d$th dimensional the $n_d$th DFT frequency grid.
Straightforward calculation shows that under the alternative hypothesis ${\mathcal H}_1$,
\begin{align}
\tilde{\mathcal Y}_{\hat{\mathbf n}} &= \sqrt{N} x\prod\limits_{d=1}^D\left({\rm e}^{{\rm j} \frac{(N_d-1)(\omega_d - \bar{\omega}_{\hat{n}_d})}{2}} \frac{ \sin\left(\frac{N_d(\omega_d - \bar{\omega}_{\hat{n}_d})}{2}\right) }{N_d\sin\left(\frac{(\omega_d - \bar{\omega}_{\hat{n}_d})}{2}\right)}\right) + \tilde{\varepsilon} \nonumber \\
&\triangleq \sqrt{N} x \beta_{\hat{\mathbf n}} + \tilde{\varepsilon},
\end{align}
$|\omega_d - \bar{\omega}_{\hat{n}_d}|\leq \pi / N_d$ and $|\beta_{\hat{\mathbf n}}| \le 1$.
It can be known that $\tilde{\mathcal Y}_{ \hat{\mathbf n} } \sim \mathcal{CN}(\sqrt{N} x \beta_{\hat{\mathbf n}}, \sigma^2)$. Conditioned on $\hat{T}$, the detection probability ${\rm P}_{\rm D}(\hat{T})$ is
\begin{align}\label{PD of T}
{\rm P}_{\rm D}(\hat{T}) &= {\rm P}\left(\left|\tilde{\mathcal Y}_{\hat{\mathbf n}}\right|^2 > \hat{T}\right) \nonumber\\
&= {\rm P}\left( \frac{\left|\tilde{\mathcal Y}_{\hat{\mathbf n}}\right|^2 }{\sigma^2}> \frac{2 \hat{T}}{\sigma^2}\right) \nonumber \\
&= Q_1 \left(\sqrt{\frac{2 N x^2 |\beta_{\hat{\mathbf n}}|^2 }{\sigma^2}}, \sqrt{\frac{2 \hat{T} }{\sigma^2}} \right)
\end{align}
where $Q_1(\cdot, \cdot)$ is the Marcum Q-function.
The average detection probability can be calculated as
\begin{align}\label{the average detection probability}
\bar{\rm P}_{\rm D} = \int_{0}^{+\infty} {\rm P}_{\rm D}(\hat{T}) p_{\hat{T}}(\hat{T}) {\rm d}\hat{T}.
\end{align}
Substituting (\ref{PDF of That}) and (\ref{PD of T}) into (\ref{the average detection probability}), one obtains
\begin{align}\label{the barPD integral}
\bar{\rm P}_{\rm D} &= \int_{0}^{+\infty} Q_1\left(\sqrt{\frac{2 N |x|^2 |\beta_{\hat{\mathbf n}}|^2 }{\sigma^2}}, \sqrt{\frac{2 \hat{T} }{\sigma^2}}\right) \frac{1}{(N_r - 1)!} \left(\frac{N_r}{\alpha \sigma^2}\right)^{N_r} \hat{T}^{N_r - 1} {\rm e}^{- \frac{N_r \hat{T}}{\alpha \sigma^2}} {\rm d}\hat{T} \nonumber \\
&\stackrel{a}= \frac{1}{(N_r - 1)!} \left(\frac{N_r}{\alpha}\right)^{N_r} \int_{0}^{+\infty} Q_1\left(\sqrt{\frac{2 N |x|^2 |\beta_{\hat{\mathbf n}}|^2 }{\sigma^2}}, \sqrt{2 T}\right) T^{N_r - 1} {\rm e}^{-\frac{N_r}{\alpha} T } {\rm d}{T}\nonumber \\
&= \frac{1}{(N_r - 1)!} \left(\frac{N_r}{\alpha}\right)^{N_r} \int_{0}^{+\infty} Q_1\left(|\beta_{\hat{\mathbf n}}|\sqrt{2{\rm SNR}}, \sqrt{2 T}\right) T^{N_r - 1} {\rm e}^{-\frac{N_r}{\alpha} T } {\rm d}{T}\nonumber \\
&\stackrel{b}\approx \frac{1}{(N_r - 1)!} \left(\frac{N_r}{\alpha}\right)^{N_r} \int_{0}^{+\infty} Q_1\left(0.88^D\sqrt{2{\rm SNR}}, \sqrt{2 T}\right) T^{N_r - 1} {\rm e}^{-\frac{N_r}{\alpha} T } {\rm d}{T}\triangleq \tilde{\rm P}_{{\rm D}}({\rm SNR}),
\end{align}
where the $\stackrel{a}=$ is due to the substitution $T = \frac{\hat{T}}{\sigma^2}$, $\stackrel{b}\approx$ is due to
${\rm E}[\beta_{\hat{\mathbf n}}]=0.88^D$ by assuming a uniform distribution for a frequency within a DFT grid interval, see \cite{MadhowNOMP} for further details.

The average detection probability $\bar{\rm P}_{\rm D}$ in the single target scene is (\ref{the barPD integral}). In the following, we discuss the detection probability $\bar{\rm P}_{\rm all}$ for multiple targets scene. Let ${\rm E}_k$ denote the event of the $k$th target being detected, and let ${\rm P}\left({\rm E}_k\right)$ denote the probability of the $k$th target being detected. As a result, the probability of all targets being detected can be calculated as
\begin{align}\label{AllTargetsDetected}
{\rm P}\left({\rm E}_1{\rm E}_2\cdots {\rm E}_K\right) = {\rm P}\left({\rm E}_1\right) {\rm P} \left({\rm E}_2|{\rm E}_1\right) \cdots {\rm P}\left({\rm E}_K|{\rm E}_1\cdots{\rm E}_{K-1}\right),
\end{align}
where ${\rm P} \left({\rm E}_i|{\rm E}_1\cdots {\rm E}_{i-1}\right)$ denote the probability of the $i$th target being detected conditioned on all the previous $i-1$ detected targets. Obviously, ${\rm P} \left({\rm E}_i|{\rm E}_1\cdots {\rm E}_{i-1}\right)\leq \tilde{\rm P}_{{\rm D}}({\rm SNR}_i)$. Consequently, ${\rm P}\left({\rm E}_1{\rm E}_2\cdots {\rm E}_K\right)$ can be upper bounded as
\begin{align}\label{AllTargetsDetectedupperbound}
{\rm P}\left({\rm E}_1{\rm E}_2\cdots {\rm E}_K\right) \leq \prod\limits_{k=1}^K \tilde{\rm P}_{{\rm D}}({\rm SNR}_k)\triangleq \bar{\rm P}_{\rm all}^{\rm upper}.
\end{align}
In Subsection \ref{Validation of CFAR Property}, $\bar{\rm P}_{\rm all}^{\rm upper}$ is evaluated to provide an upper bound of the detection probability of all the targets.

\subsection{NOMP-CFAR}\label{AdCFARMDOMP}
Subsection \ref{detdesign} has proposed the CFAR detector for a single target scenario and provided (\ref{simplePoe}) for calculating the threshold multiplier $\alpha$ for a specified $\bar{\rm P}_{\rm FA}$. Now we focus on the multiple target detection problem and design NOMP-CFAR algorithm.

For the target detection and estimation problem, one should try to find a minimal representation of the line spectral such that the residue can be well approximated as the AWGN. Note that the conventional CFAR based approach consists of two steps: Firstly, performing the fast Fourier transform (FFT) on $\mathcal{Y}$ to obtain the spectrum $\tilde{\mathcal{Y}}$ \footnote{For range estimation, range Doppler estimation, the spectrum corresponds to range spectrum and range-Doppler spectrum, respectively.}. Secondly, implementing the CFAR detection such as CA or OS methods for all the cells by specifying the guard length, training length and the average false alarm rate $\bar{\rm P}_{\rm FA}$, the target detection results are obtained. Obviously, the FFT based approach is suboptimal as it suffers from grid mismatch and inter-target intergerence. Besides, the resolution of the FFT-based approach is limited by the Rayleigh criterion. Therefore, we propose NOMP-CFAR to overcome the drawbacks of the traditional CFAR based approach, while inherit the advantages of the CFAR behaviour.

The main steps of the NOMP-CFAR can be divided into two steps: Initialization and MOE or Detection step. The initialization step provides a good initial point of the NOMP-CFAR. While for the MOE step, it is iterative and it includes an activation step and deactivation step, i.e., by activating the target or deactivating the target. For the initialization step: It is assumed that the number of targets $K$ is upper bounded by a known constant $K_{\rm max}$, i.e., $K\leq K_{\rm max}$. This is reasonable as one could choose $K_{\rm max} = N$ as the number of targets must be less than the number of measurements $N$. Then we run NOMP consisting of coarse detection, single refinement, cyclic refinement, least squares (LS) step to obtain the candidate amplitudes and frequencies set ${\mathcal K}_{K_{\rm max}}=\{(\hat{x}_k,\hat{{\boldsymbol \omega}}_k),k=1,2,\cdots,K_{\rm max}\}$. Note that this step is necessary as it alleviates the intertarget interference and provides fine results for the MOE step. For further details please see \cite{MadhowNOMP}. For the MOE step, the CFAR approach is implemented to perform target detection, i.e., activate and deactivate the components of the frequencies. In addition, all the steps such as coarse detection, single refinement, cyclic refinement, LS in NOMP are also incorporated to enhance the target estimation and detection accuracy. Below we present the details of the MOE step.

Given the candidate set ${\mathcal K}=\{(\hat{x}_k,\hat{{\boldsymbol \omega}}_k),k=1,2,\cdots,\hat{K}\}$ and $\hat{K}=|{\mathcal K}|\leq K_{\rm max}$, one could use the CFAR criterion to detect the $k$th frequency. To implement the CFAR detector, we proceed as follows: The residue of the observation $\mathcal{Y}_{\rm r}({\mathcal K})$ after eliminating all the targets ${\mathcal K}$ is
\begin{align}\label{residue observation}
\mathcal{Y}_{\rm r}({\mathcal K}) = \mathcal{Y} - \sum_{k = 1}^{\hat{K}} \hat{x}_k \mathcal{A}(\hat{\boldsymbol \omega}_k).
\end{align}
Provided that all the amplitudes and the frequencies are detected and estimated in high accuracy, $\mathcal{Y}_{\rm r}({\mathcal K})$ is approximately AWGN, i.e.,
\begin{align}\label{appAWGN}
\mathcal{Y}_{\rm r}({\mathcal K}) \approx \varepsilon_{\rm r}
\end{align}
and ${\rm vec }(\varepsilon_{\rm r})\sim {\mathcal {CN}}({\mathbf 0},\sigma_{\rm r}^2{\mathbf I}_{N})$. To detect the $k$ target, we obtain the pseudo measurement $\mathcal{Y}_{{\rm r}}({\mathcal K}\backslash (\hat{x}_k,\hat{{\boldsymbol \omega}}_k)) $ by adding the contribution of the $k$th frequency to the residue $\mathcal{Y}_{\rm r}({\mathcal K})$, i.e.,
\begin{align}\label{under test model}
\mathcal{Y}_{{\rm r}}({\mathcal K}\backslash (\hat{x}_k,\hat{{\boldsymbol \omega}}_k)) = \mathcal{Y}_{\rm r}({\mathcal K}) + \hat{x}_k \mathcal{A}(\hat{\boldsymbol \omega}_k) \stackrel{a}\approx \hat{x}_k \mathcal{A}(\hat{\boldsymbol \omega}_k) + \varepsilon_{\rm r},
\end{align}
where $\stackrel{a}\approx$ is due to (\ref{appAWGN}). We now input $\mathcal{Y}_{{\rm r}}({\mathcal K}\backslash (\hat{x}_k,\hat{{\boldsymbol \omega}}_k))$ and the required threshold multiplier $\alpha$ to the CFAR detector to provide a soft detection of the $k$th frequency, $k=1,2,\cdots,\hat{K}$, i.e., calculate $\Delta_k$ (\ref{deltadef}). Note that $\Delta_k\geq 0$ means that the power of the CUT exceeds the estimated threshold, and this target should be detected if the CFAR is adopted. If we radically implement the CFAR detector to detect all the targets and then stop, miss event may happen in certain scenarios. For example, in the Initialization step, a strong target may be detected as two closely-spaced frequencies with distinguishable amplitudes. This splitting phenomenon degrades the SNR of the frequency and the CFAR detector may miss the two targets at the same time, which misses the true frequency. To overcome this problem, we only activate or deactivate one frequency in each iteration. In this way, even after splitting, the frequency corresponding to the weaker amplitude may be eliminated and the frequency one with stronger amplitude can be refined, and the true frequency may be detected. To activate or deactivate one frequency in each iteration, the greedy approach is adopted by calculating
\begin{align}\label{minimum characteristic}
\hat{k} = \underset{k=1,2,\cdots,\hat{K}}{\operatorname{argmin}}{\Delta}_{k}
\end{align}
and ${\Delta}_{\hat{k}}$. Depending on the sign of ${\Delta}_{\hat{k}}$, we choose two different actions:
\begin{itemize}
  \item [A1] ${\Delta}_{\hat{k}}<0$: In this setting, deactive the $\hat{k}$th frequency by removing the $\hat{k}$th frequency, i.e., updating $\hat{\mathcal K}' = {\mathcal K} \backslash (\hat{x}_{\hat{k}},\hat{{\boldsymbol \omega}}_{\hat{k}})$. Perform the single refinement, cyclic refinement and LS for the remaining targets and obtain $\hat{\mathcal K}''$. This refinement step is necessary to ensure the high estimation accuracy. Then for the updated parameter set $\hat{\mathcal K}''$, perform CFAR detection as done before.
  \item [A2] ${\Delta}_{\hat{k}} \geq 0$: In this setting, the frequency set is kept. Then, perform CFAR detection on the residual $\mathcal{Y}_{\rm r}({\mathcal K})$ to determine whether a new amplitude frequency pair should be added on the list ${\mathcal K}$ or not. If the CFAR does not detect any target or the number of targets is greater than $K_{\rm max}$, then the algorithm stops. Otherwise, add the new amplitude frequency pair into the list and iterate.
\end{itemize}
Note that compared to the traditional CFAR, our procedure can be viewed as a soft CFAR as we proceed very gently.
The NOMP-CFAR is summarized as Algorithm \ref{AdapCFARMDNOMPalgorithm}.

\begin{algorithm}
\caption{NOMP-CFAR} \label{AdapCFARMDNOMPalgorithm}
\begin{algorithmic}[1]
\STATE Procedure: NOMP-CFAR $({\mathcal Y}, K_{\rm max}, \alpha)$ : \\
Initialization step:\\
\STATE Run NOMP to obtain the candidate amplitude frequency pair set ${\mathcal K}_{K_{\rm max}}=\{(\hat{x}_k,\hat{{\boldsymbol \omega}}_k),k=1,2,\cdots,K_{\rm max}\}$.\\
MOE step: \\
\STATE $\hat{\mathcal K} \leftarrow {\mathcal K}_{K_{\rm max}}$
\STATE \textbf{while} $\hat{K} \geq 0$  and $\hat{K} \leq K_{\rm max}$, \\
\STATE \label{step1} \quad Compute $\Delta_k \leftarrow {\rm CFAR}(\mathcal{Y}_{\rm r}({\mathcal K}\backslash (\hat{x}_k,\hat{{\boldsymbol \omega}}_k)), \alpha)$ for $k = 1, \cdots, \hat{K}$. \\
\STATE \quad Compute $\hat{k}$ and $\Delta_{\hat{k}}$ (\ref{minimum characteristic}). \\
\STATE \quad \textbf{if} $\Delta_{\hat{k}} < 0$\\
\STATE \quad \quad ${\mathcal K} \leftarrow {\mathcal K} \backslash (\hat{x}_{\hat{k}},\hat{{\boldsymbol \omega}}_{\hat{k}})$
\STATE \quad \quad Perform cyclic refinement and LS on the remaining parameter set ${\mathcal K}$ and go to Step \ref{step1}. \\
\STATE \quad \textbf{else} \\
\STATE \quad \quad Compute $\Delta_{\rm new} \leftarrow {\rm CFAR}(\mathcal{Y}_{\rm r}({\mathcal K}), \alpha)$
\STATE \quad \quad \textbf{if} $\Delta_{\rm new} \geq 0$ \\
\STATE \quad \quad \quad Run single refinement to obtain the new amplitude frequency pair $(\hat{x}_{\rm new},\hat{{\boldsymbol \omega}}_{\rm new})$ , ${\mathcal K} \leftarrow {\mathcal K} \cup (\hat{x}_{\rm new},\hat{{\boldsymbol \omega}}_{\rm new})$ and go to Step \ref{step1}. \\
\STATE \quad \quad \textbf{else}
\STATE \quad \quad \quad  \textbf{break} \\
\STATE \quad \quad \textbf{end if} \\
\STATE \quad \textbf{end if} \\
\STATE \textbf{end while}
\STATE \textbf{return} ${\mathcal K}$
\end{algorithmic}
\end{algorithm}

\subsection{Extension of NOMP-CFAR}
The above NOMP-CFAR can be easily extended to deal with the MMV setting and the compressive setting.
\subsubsection{MMV setting}
The MMV model can be described as
\begin{align}\label{sth snapshot Y}
\mathcal{Y}(s) = \sum_{k = 1}^{K} x_k(s) \mathcal{A}({\boldsymbol \omega}_k) + \mathcal{\varepsilon}(s), \quad s = 1,2,\cdots,S,
\end{align}
where $S$ denotes the number of snapshots. Details about single refinement, cyclic refinement and LS step can be referred to \cite{mnomp}. Here we just focus on the difference and provide the details of the CFAR detection.

Define $\tilde{\mathcal{Y}}(s)$ as the normalized DFT of $\mathcal{Y}(s)$. The CFAR detector deciding the alternative hypothesis ${\mathcal H}_1$ is
\begin{align}\label{detector_noiseunknownmmv}
T({{\mathcal Y}})=\frac{\frac{1}{S}\sum\limits_{s=1}^S|\tilde{\mathcal Y}_{{\tilde{\mathbf n}}_{\rm peak}}(s)|^2}{\hat{\sigma}_{\rm mmv}^2}\geq \alpha_{\rm mmv},
\end{align}
where ${\tilde{\mathbf n}}_{\rm peak}$ denotes the peak localization of $\frac{1}{S}\sum\limits_{s=1}^S|\tilde{\mathcal Y}(s)|^2$ which is the average of the spectrum over the snapshots, $\hat{\sigma}_{\rm mmv}^2$ is
\begin{align}
\hat{\sigma}_{\rm mmv}^2 = \frac{1}{SN_r}\sum\limits_{s=1}^S\sum\limits_{{\tilde{\mathbf n}}\in {\mathcal{T}}_{{\tilde{\mathbf n}}_{\rm peak}}} \left|\tilde{\mathcal Y}_{\tilde{\mathbf n}}(s)\right|^2
\end{align}
$\alpha_{\rm mmv}$ is the required threshold multiplier. Then, given $\alpha_{\rm mmv}$, the false alarm probability $\bar{\rm P}_{\rm FA}$
\begin{align}\label{simple type of Poe}
\bar{\rm P}_{\rm FA} = 1 - \frac{N_r}{\alpha (SN_r - 1)!} \int_{0}^{+ \infty} \left(1-\mathrm{e}^{- T} \sum_{s=0}^{S - 1} \frac{T^{s}}{s!}\right)^N \mathrm{e}^{- \frac{N_r T}{ \alpha}} \left( \frac{N_r T}{\alpha} \right)^{SN_r - 1} dT.
\end{align}
is established, see Appendix \ref{MMVscenario} for the details. Note that for SMV scenario, i.e., $S=1$, (\ref{simple type of Poe}) reduces to (\ref{simplePoe}).
\subsubsection{Compressive Setting}
For simplicity, we focus on the $1$-dimensional compressive line spectrum estimation problem formulated as
\begin{align}\label{compressive scene model}
{\mathbf y} = {\boldsymbol \Phi}\left(\sum_{k = 1}^{K}  {\mathbf a}({\boldsymbol \omega}_k)x_k\right) + \varepsilon,
\end{align}
where ${\boldsymbol \Phi}\in{\mathbb C}^{M\times N}$ is the compression matrix and $M/N$ is the compression rate. Define a ``generalized'' atom as
\begin{align}
\breve{{\mathbf a}}(\omega)=\frac{{\boldsymbol\Phi}{\mathbf a}(\omega)}{\|{\boldsymbol\Phi}{\mathbf a}(\omega)\|_2}.
\end{align}
The CS model becomes
\begin{align}\label{CompressiveSceneModel}
\mathbf{y} = \sum\limits_{k=1}^K\breve{x}_k\breve{\mathbf a}(\omega) + \varepsilon.
\end{align}
The steps of the coarse detection, single refinement, cyclic refinement and LS are the same as that of NOMP. Here we present the CFAR details. Given $\mathbf{y}({\mathcal K})=\mathbf{y}-\sum\limits_{k=1}^{|{\mathcal K}|}\hat{\breve{x}}_k\breve{\mathbf a}(\hat{\omega})$, where ${\mathcal K}=\{(\hat{\breve{x}}_k,\hat{\omega}_k),k=1,2,\cdots,|\mathcal K|\}$ and $(\hat{\breve{x}}_k,\hat{\omega}_k)$ denotes the estimate of the amplitude frequency pair, calculate the ``generalized'' spectrum as
\begin{align}
\tilde{\mathbf{y}}({\mathcal K})=\bar{\mathbf{a}}^{\rm H}(\omega)\mathbf{y}({\mathcal K})
\end{align}
by restricting $\omega$ to the DFT grid. Then the detector (\ref{GLRT_noiseunknown}) is used to perform detection, and  the noise variance estimate is given by (\ref{sigmaest}). Besides, the required threshold multiplier for a given $\bar{\rm P}_{\rm FA}$ is (\ref{simplePoe}).

\section{Numerical Simulation}\label{Numerical Simulation}
In this section, the derived theoretical result is verified and the performance of the proposed NOMP-CFAR is investigated. Firstly, the false alarm probability versus the required threshold multiplier is evaluated, which provides a way to calculate the required threshold multiplier for a specified false alarm probability. Secondly, the CFAR property is validated. Third, the benefits of initialization is demonstrated. Fourthly, the performance of the proposed NOMP-CFAR is evaluated, including the SMV, MMV and compressive settings. Finally, we consider a time-varying noise variance scenario and show the CFAR property of the proposed algorithm. Since the performance of FFT based CFAR detection is poor in multiple target scenarios (here we set the number of frequencies as $16$), we do not compare with it. The NOMP achieves the state-of-art performance and we mainly focus on comparing NOMP-CFAR with it.

The parameters are set as follows: $N=256$, the number of frequencies is $K = 16$, the frequencies are randomly drawn and the wrap around distances of any two frequencies are greater than $2.5\Delta_{\rm DFT}$, where $\Delta_{\rm DFT}=2 \pi / N$.
The noise variance is $\sigma^2 = 1$.
The default algorithm parameters are set as follows: The average false alarm probability is $\bar{\rm P}_{\rm FA} = 10 ^ {-2}$, the number of training cells $N_r = 50$ unless stated otherwise.
Therefore, the required multiple factor $\alpha$ can be calculated by (\ref{binomial Poe}), which is $\alpha = 11.22$.
And the threshold of NOMP is $\tau_{\rm NOMP} = \alpha'_{\rm NOMP} \sigma^2 = 10.15$.
The over sampling rate is $\gamma = 4$ and the upper bound of the number of targets is $K_{\rm max} = 2K$.
All the results are averaged over $3000$ Monte Carlo (MC) trials.

A frequency $\omega$ is detected provided that $\underset{k=1,\cdots,\hat{K}}{\operatorname{min}}{\rm dist}(\hat{\omega}_k, \omega)\leq 0.5\Delta_{\rm DFT}$. We compute the detection probability ${\rm P}_{\rm D}$ defined as the detection probability of all the true frequencies. An estimated frequency $\hat{\omega}$ is a false alarm provided that $\underset{k=1,\cdots,{K}}{\operatorname{min}}{\rm dist}(\hat{\omega}, \omega_k)\geq  0.5\Delta_{\rm DFT}$, where $\hat{K}$ denotes the number of estimated frequencies. We calculate the false alarm probability $\bar{\rm P}_{\rm FA}$ defined as the false alarm rate of all the estimated frequencies.

\subsection{$\bar{\rm P}_{\rm FA}$ versus the $\alpha$}
In this subsection, the false alarm probability $\bar{\rm P}_{\rm FA}$ for a required threshold multiplier $\alpha$ of NOMP-CFAR and NOMP are calculated and results are shown in Fig. \ref{PoevsalphabyS}. In addition, the required threshold multiplier $\alpha$ of NOMP-CFAR is also approximately calculated through the CFAR approach (\ref{resultPoeapp}) termed as CFAR approx. for the SMV scenario. For the SMV scenario $S=1$ and $\bar{\rm P}_{\rm FA} \leq 0.5$ which is often of interest as $\bar{\rm P}_{\rm FA}$ is often set as a very small value such as $10^{-2}$ or $10^{-1}$ and $\bar{\rm P}_{\rm FA}$ is fixed, the required threshold multiplier $\alpha$ of NOMP-CFAR is higher than $\alpha^{'}_{\rm NOMP}$ of NOMP, i.e., $\alpha^{'}_{\rm NOMP}< \alpha$. For small $\bar{\rm P}_{\rm FA}$, i.e., $\bar{\rm P}_{\rm FA} \leq 0.01$, the required multiplier calculated by CFAR approach approximates that of NOMP-CFAR well.
In detail, for $\bar{\rm P}_{\rm FA}=10^{-2}$, the required threshold multipliers of NOMP-CFAR, CFAR approx. and NOMP are $11.22$, $11.25$ and $10.15$, respectively, corresponding to $10.50$ dB, $10.51$ dB and $10.06$ dB. For the MMV scenario, the required threshold multiplier of NOMP-CFAR is close to that of NOMP. In addition, for a specified $\bar{\rm P}_{\rm FA}$, the required threshold multiplier decreases as the number of snapshots $S$ increases. For example, for $\bar{\rm P}_{\rm FA}=10^{-2}$, the required threshold multipliers of NOMP-CFAR for $S=1$, $S=10$ and $S=50$ are $11.22$, $2.81$ and $1.67$, respectively, corresponding to $10.50$ dB, $4.49$ dB and $2.23$ dB.
\begin{figure*}
  \centering
  \subfigure[]{
  \includegraphics[width=65mm]{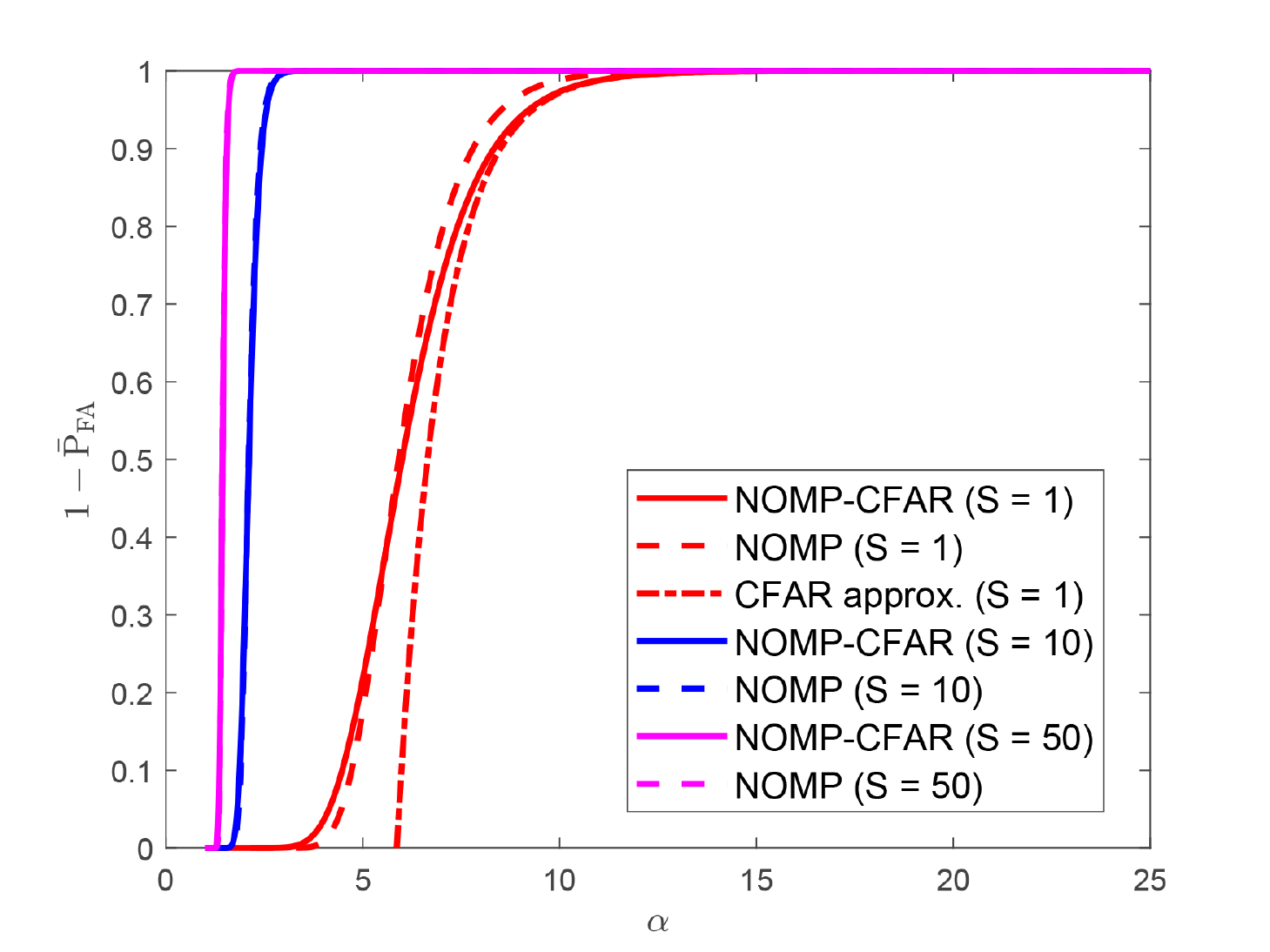}}
  \subfigure[]{
  \includegraphics[width=65mm]{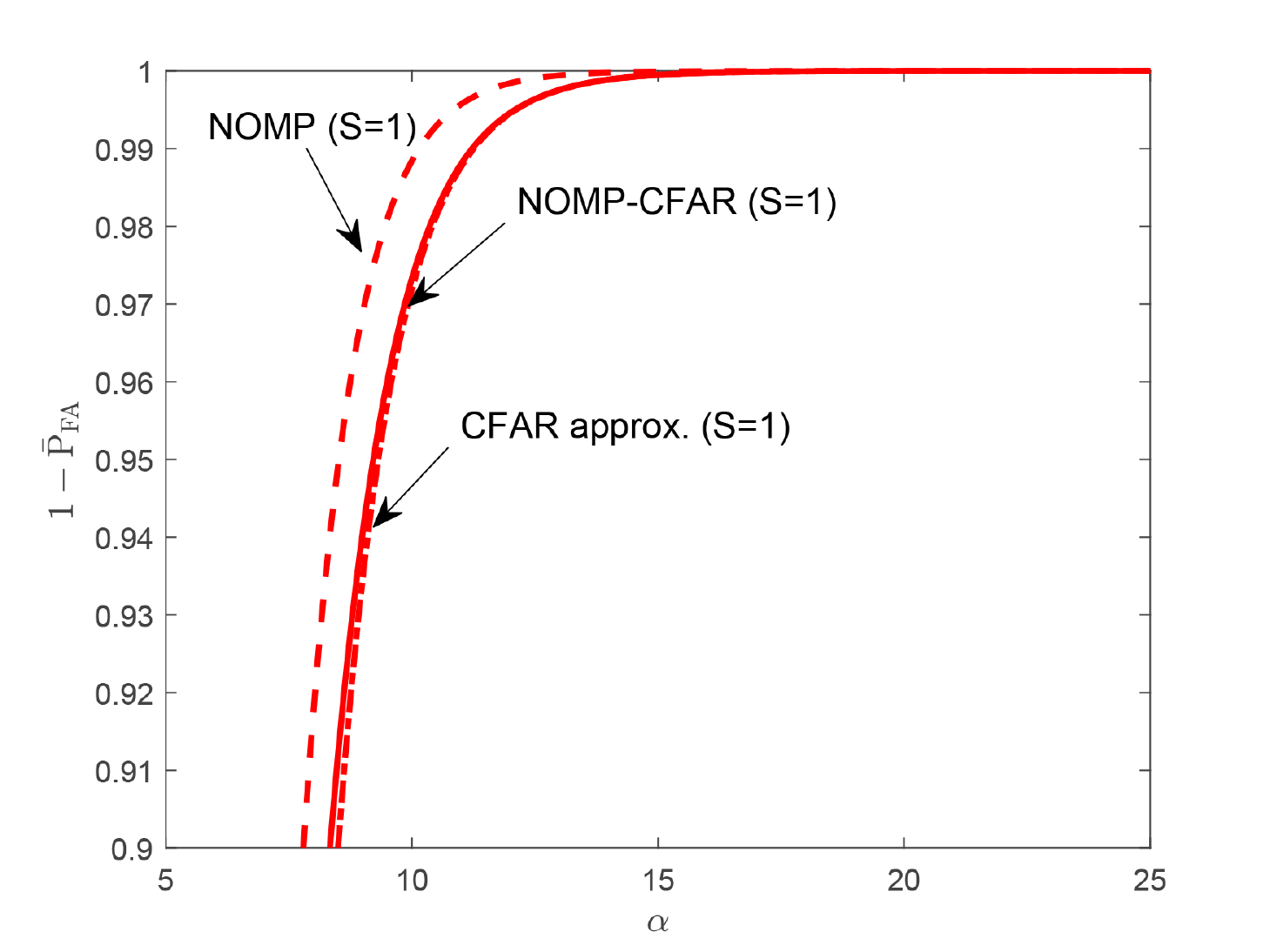}}
  \caption{The $\bar{\rm P}_{\rm FA}$ versus $\alpha$ for the snapshots $S=1$, $S=10$, and $S=50$ (a), (b) Zoomed in. respectively.}\label{PoevsalphabyS}
\end{figure*}
\subsection{Validation of CFAR Property}\label{Validation of CFAR Property}
The CFAR Property is verified by conducting a simulation in a SMV scenario. The nominal $\bar{\rm P}_{\rm FA}$ is varied from $0.01$ to $0.12$. The results are shown in Fig. \ref{performance versus PFA by SNR}. It can be seen that the measured $\bar{\rm P}_{\rm FA}$ is close to that of the nominal $\bar{\rm P}_{\rm FA}$ at three different SNRs ${\rm SNR}=14$ dB, ${\rm SNR}=15$ dB and ${\rm SNR}=18$ dB, demonstrating the CFAR behaviour of both the NOMP and  NOMP-CFAR algorithm in constant noise variance scenario. In addition, the measured detection probability $\bar{\rm P}_{\rm D}$ and the upper bound of the theoretical detection probability $\bar{\rm P}_{\rm all}^{\rm upper}$ (\ref{AllTargetsDetectedupperbound}) are also plotted and shown in Fig. \ref{PDvsPFAbySNR}. It can be seen that $\bar{\rm P}_{\rm all}^{\rm upper}$ is an upper bound of the measured detection probability $\bar{\rm P}_{\rm D}$ of the NOMP-CFAR algorithm. Besides, the detection probability of NOMP is higher than that of NOMP-CFAR, and the performance gap lowers as SNR increases.
\begin{figure}
  \centering
  \subfigure[]{
  \label{PFAvs PFAbySNR}
  \includegraphics[width=65mm]{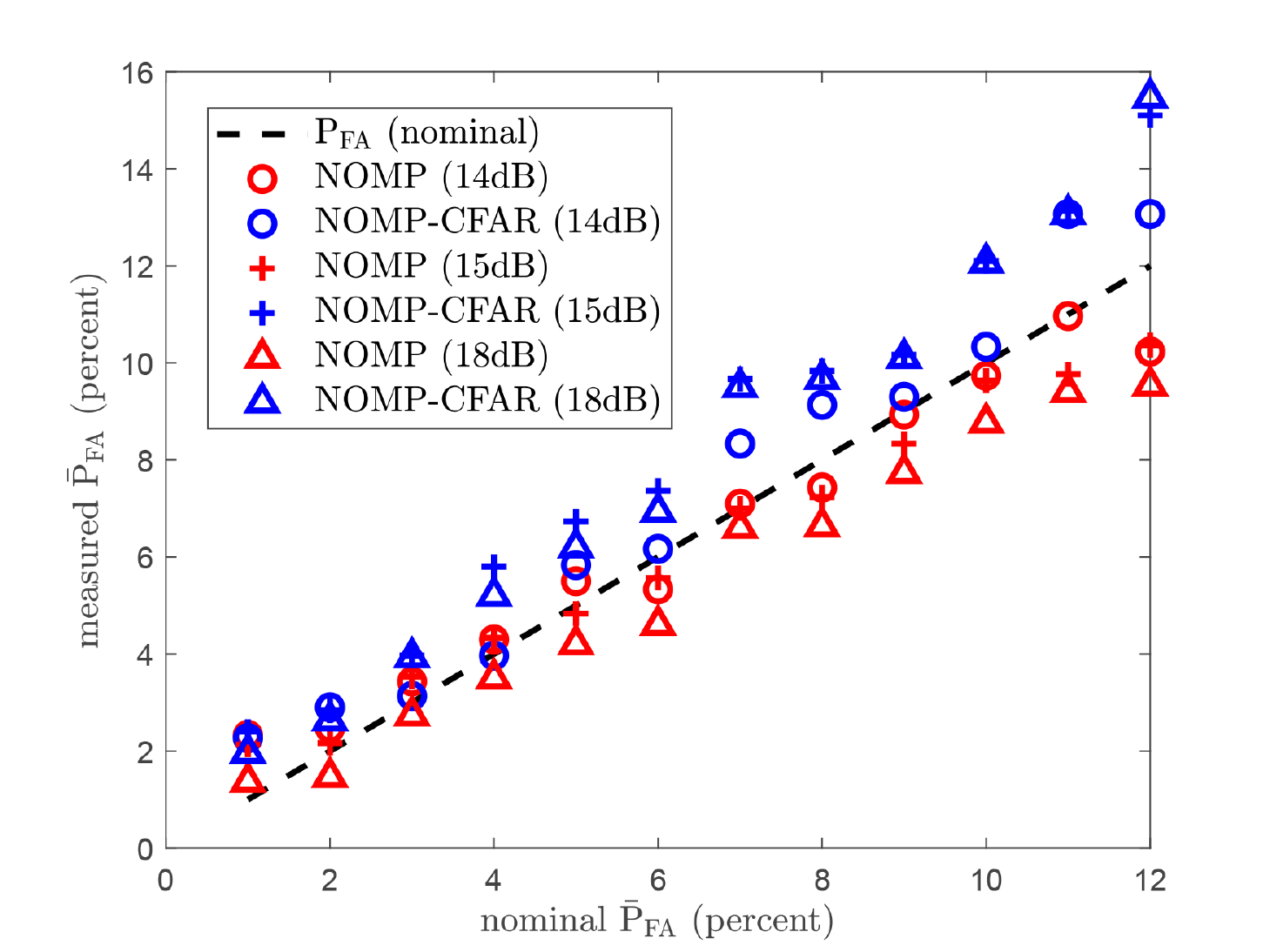}}
  \subfigure[]{
  \label{PDvsPFAbySNR}
  \includegraphics[width=65mm]{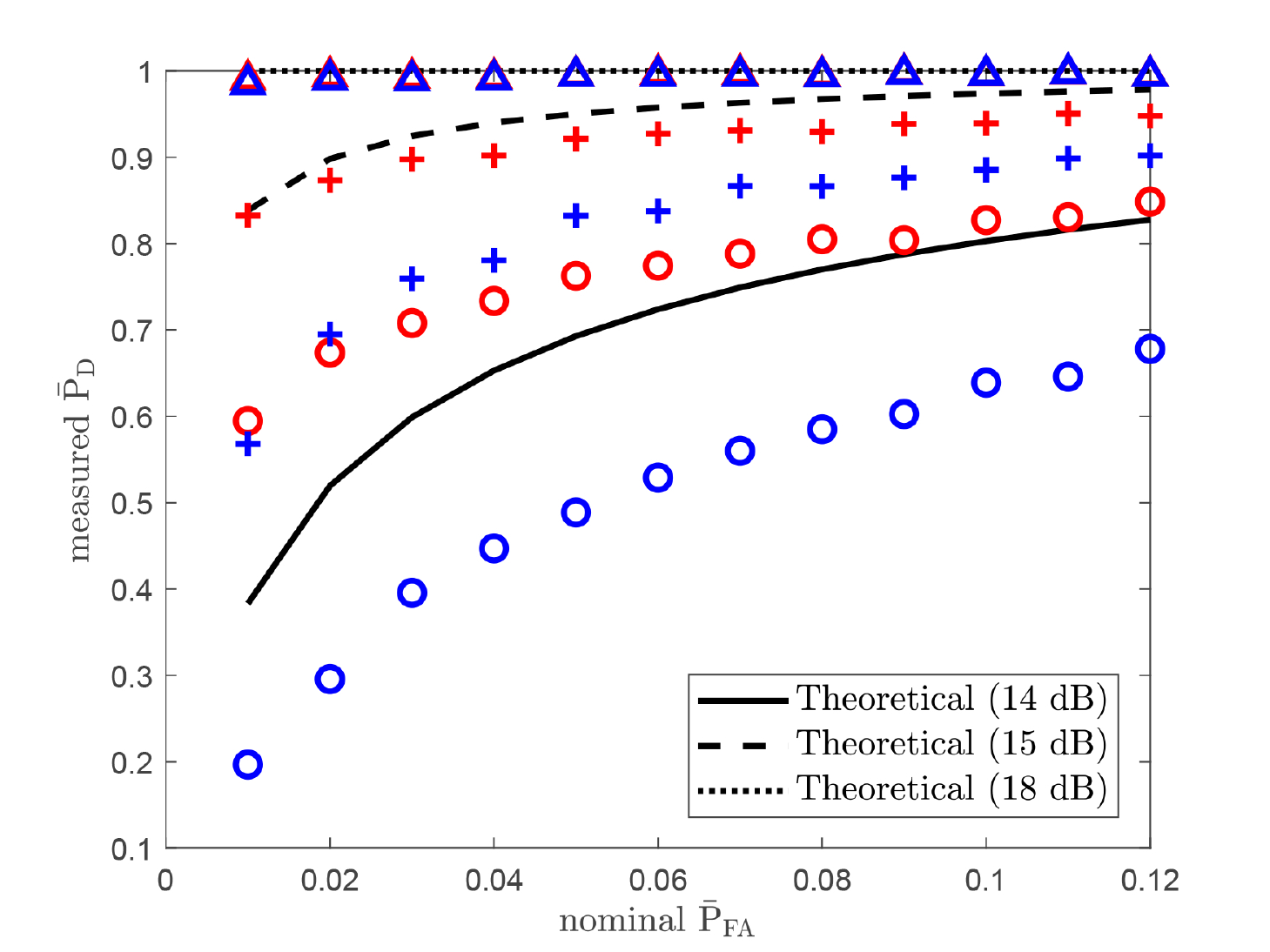}}
  \caption{The performances of measured $\bar{\rm P}_{\rm FA}$ (a) and $\bar{\rm P}_{\rm D}$ (b) versus nominal $\bar{\rm P}_{\rm FA}$ for three different SNRs ${\rm SNR}=14$ dB, ${\rm SNR}=15$ dB and ${\rm SNR}=18$ dB.}\label{performance versus PFA by SNR}
\end{figure}

\subsection{Benefits of Initialization} \label{Benefits of Initialization}
This subsection conducts a simulation to illustrate the benefits of initialization.
We design the NOMP-CFAR (forward) algorithm which just incrementally detects a new frequency by the CFAR detection algorithm and stops until no frequency is detected.
The measured false alarm probability $\bar{\rm P}_{\rm FA}$ and the detection probability $\bar{\rm P}_{\rm D}$ versus the ${\rm SNR}$ are shown in Fig. \ref{PoeandPDforvsback}. From Fig. \ref{Poeforvsback}, the measured false alarm rates $\bar{\rm P}_{\rm FA}$ of NOMP, NOMP-CFAR, NOMP-CFAR (forward) are close to the nominal $\bar{\rm P}_{\rm FA}$ for ${\rm SNR}\geq 17$ dB. From Fig. \ref{PDforvsback}, the detection probability $\bar{\rm P}_{\rm D}$ of NOMP is highest, followed by NOMP-CFAR, NOMP-CFAR (forward). For ${\rm SNR}\geq 17$ dB,  the detection probabilities $\bar{\rm P}_{\rm D}$ of NOMP and NOMP-CFAR are close to $1$, while the detection probability $\bar{\rm P}_{\rm D}$ of NOMP-CFAR (forward) is near $0.8$. Besides, for $\bar{\rm P}_{\rm D}=0.6$, the SNRs of NOMP, NOMP-CFAR, NOMP-CFAR (forward) are ${\rm SNR}=14$ dB, ${\rm SNR}=15$ dB and ${\rm SNR}=17$ dB, respectively, demonstrating that initialization yield $1$ dB performance gain.
\begin{figure}
  \centering
  \subfigure[]{
  \label{Poeforvsback} 
  \includegraphics[width=65mm]{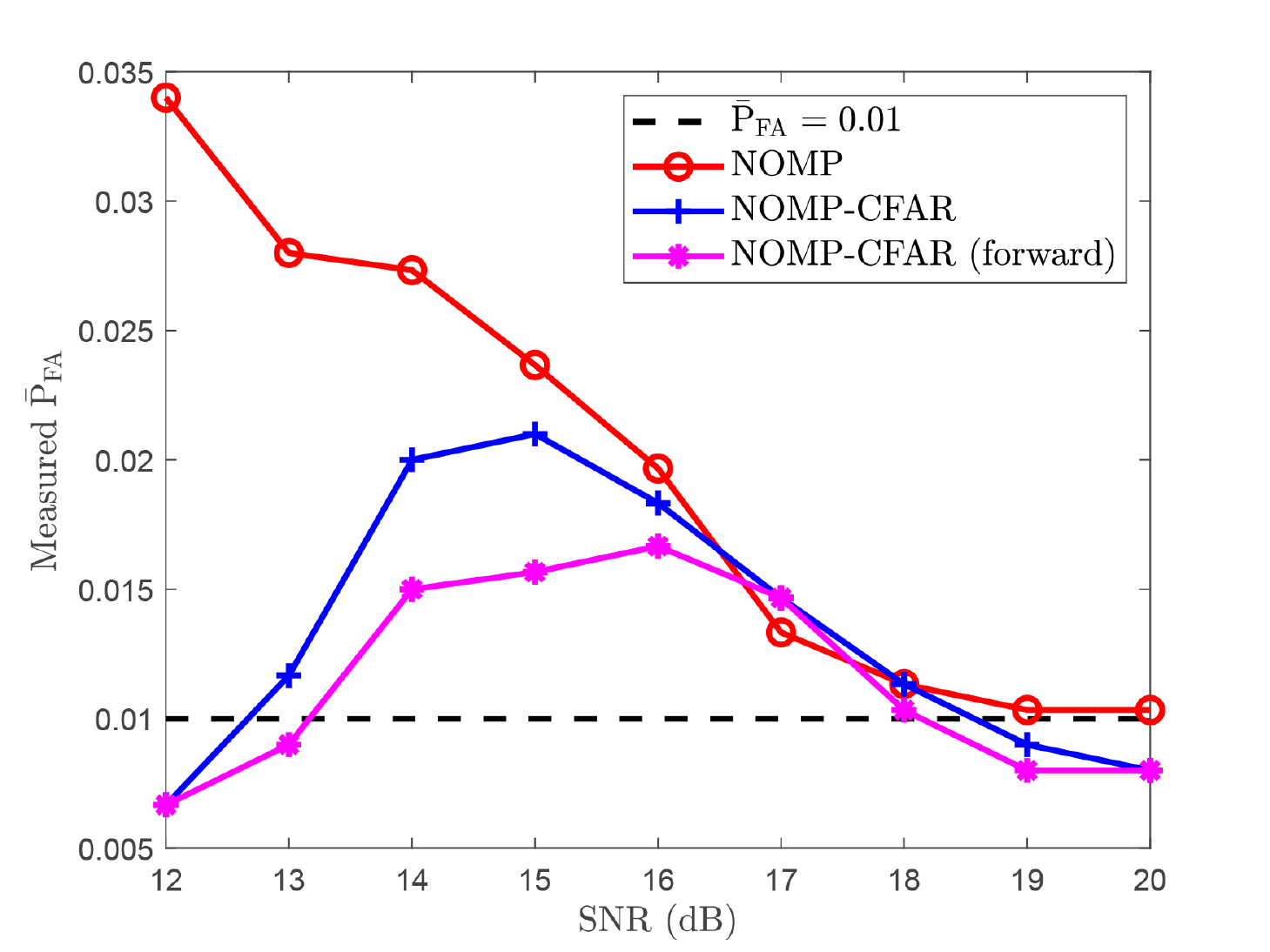}}
  \subfigure[]{
  \label{PDforvsback} 
  \includegraphics[width=65mm]{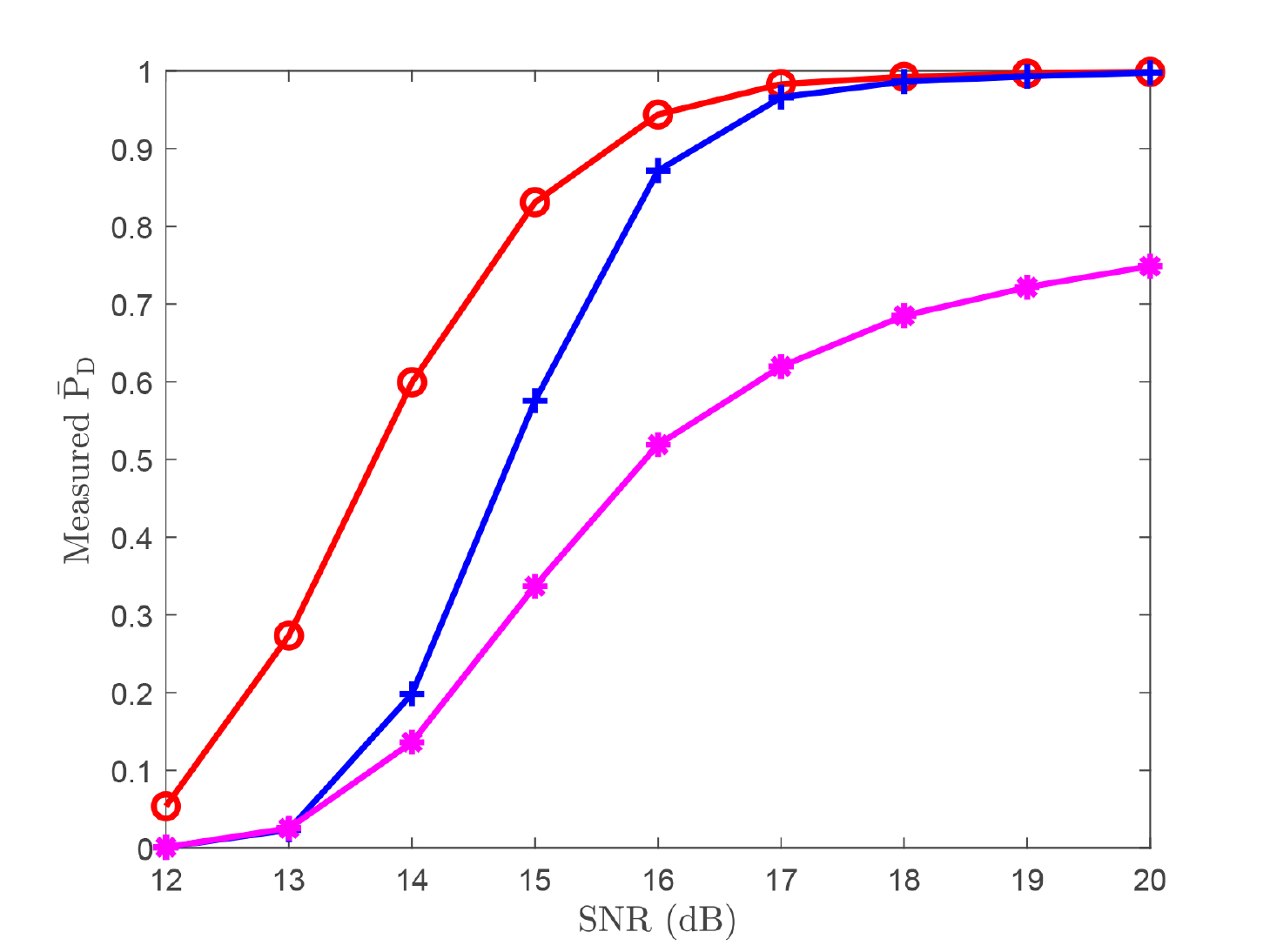}}
  \caption{The $\bar{\rm P}_{\rm FA}$ (a) and $\bar{\rm P}_{\rm D}$ (b) versus ${\rm SNR}$ of NOMP, Adap-NOMP-CFAR algorithm with and without initialization. The dashed line denotes the nominal $\bar{\rm P}_{\rm FA}$.}\label{PoeandPDforvsback} 
\end{figure}

\subsection{Performance versus the SNR in the SMV Setting} \label{Performance versus the SNR SMV}
The performances of the proposed algorithm is investigated in the SMV scenario. The SNRs of all the targets are identical and are varied from $12$ dB to $20$ dB.
The false alarm rate $\bar{\rm P}_{\rm FA}$, the detection probability $\bar{\rm P}_{\rm D}$ of all the targets, the correct model order estimation rate ${\rm P}(\hat{K} = K)$, the MSE of the frequency estimation error $\|\hat{\omega}-\omega\|_2^2$ averaged over the trials in which $\hat{K} = K$, the normalized reconstruction error ${\rm NMSE}(\hat{\mathbf z})=\|\hat{\mathbf z}-{\mathbf z}\|^2/\|{\mathbf z}\|_2^2$ are used as the performance metrics. The results are shown in Fig. \ref{vsSNR_all_fig}.

\begin{figure*}
  \centering
  \subfigure[]{
  \label{POE_fig} 
  \includegraphics[width=50mm]{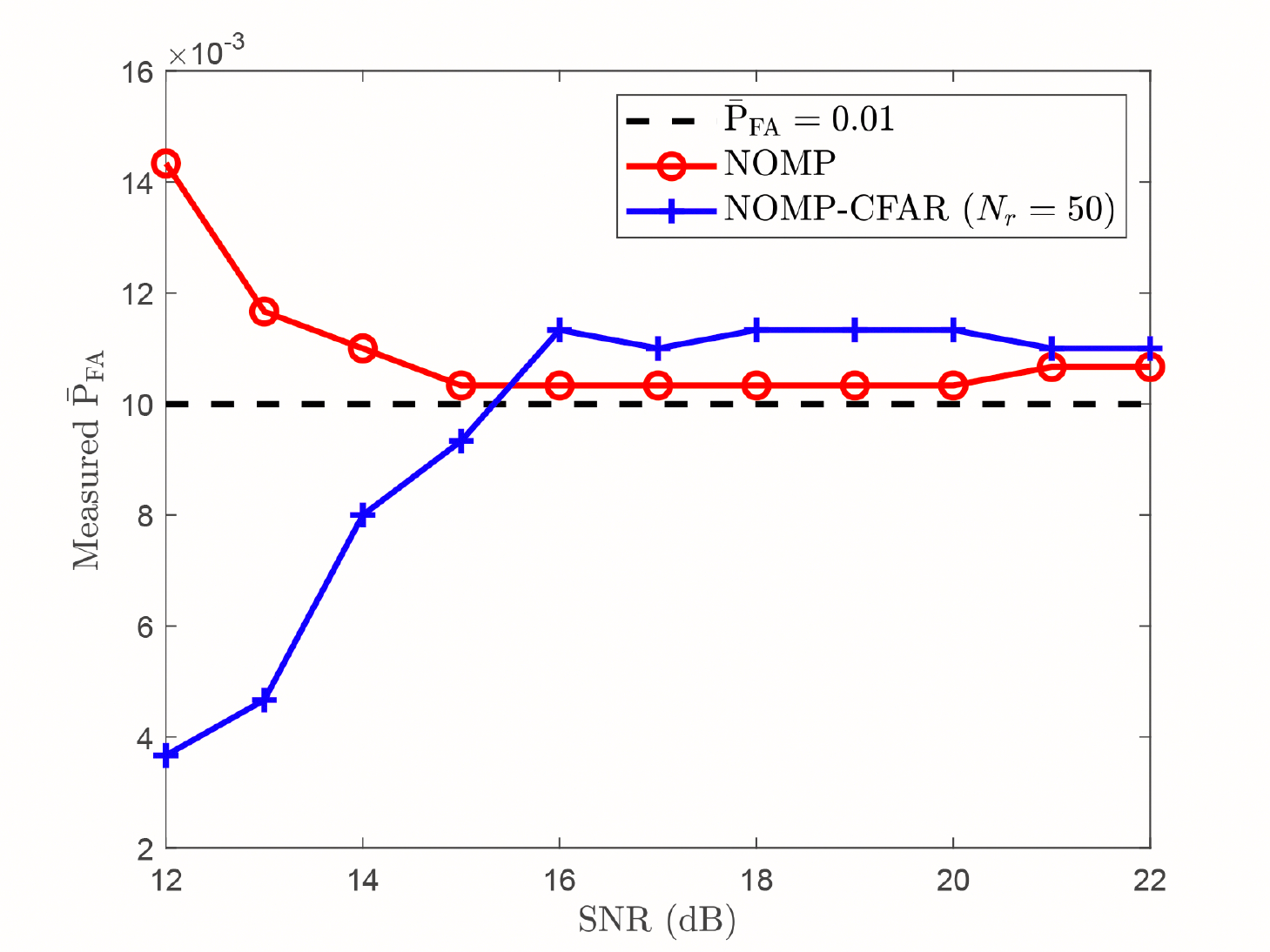}}
  \subfigure[]{
  \label{PD_fig} 
  \includegraphics[width=50mm]{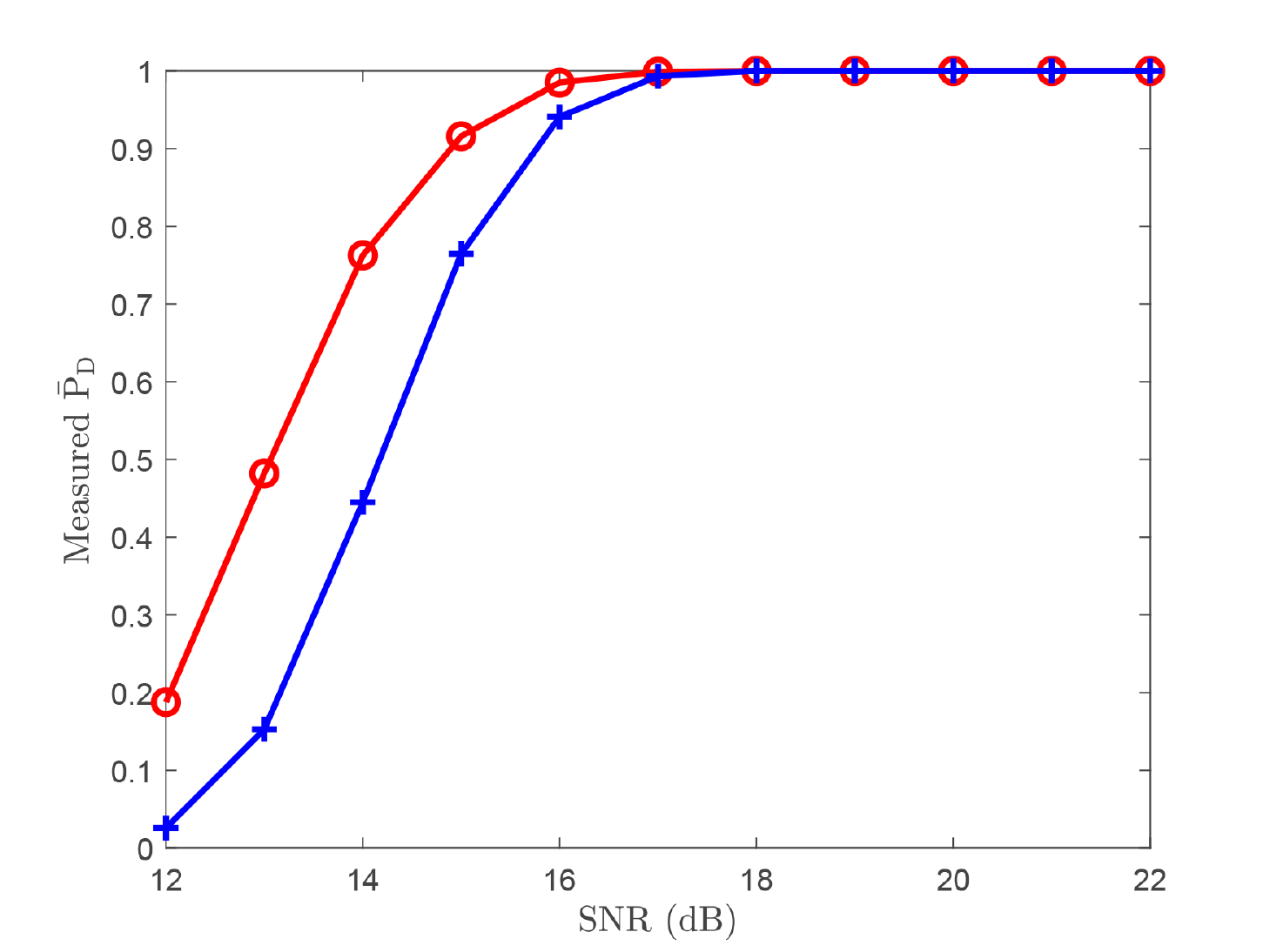}}
  \subfigure[]{
  \label{Pequal_fig}
  \includegraphics[width=50mm]{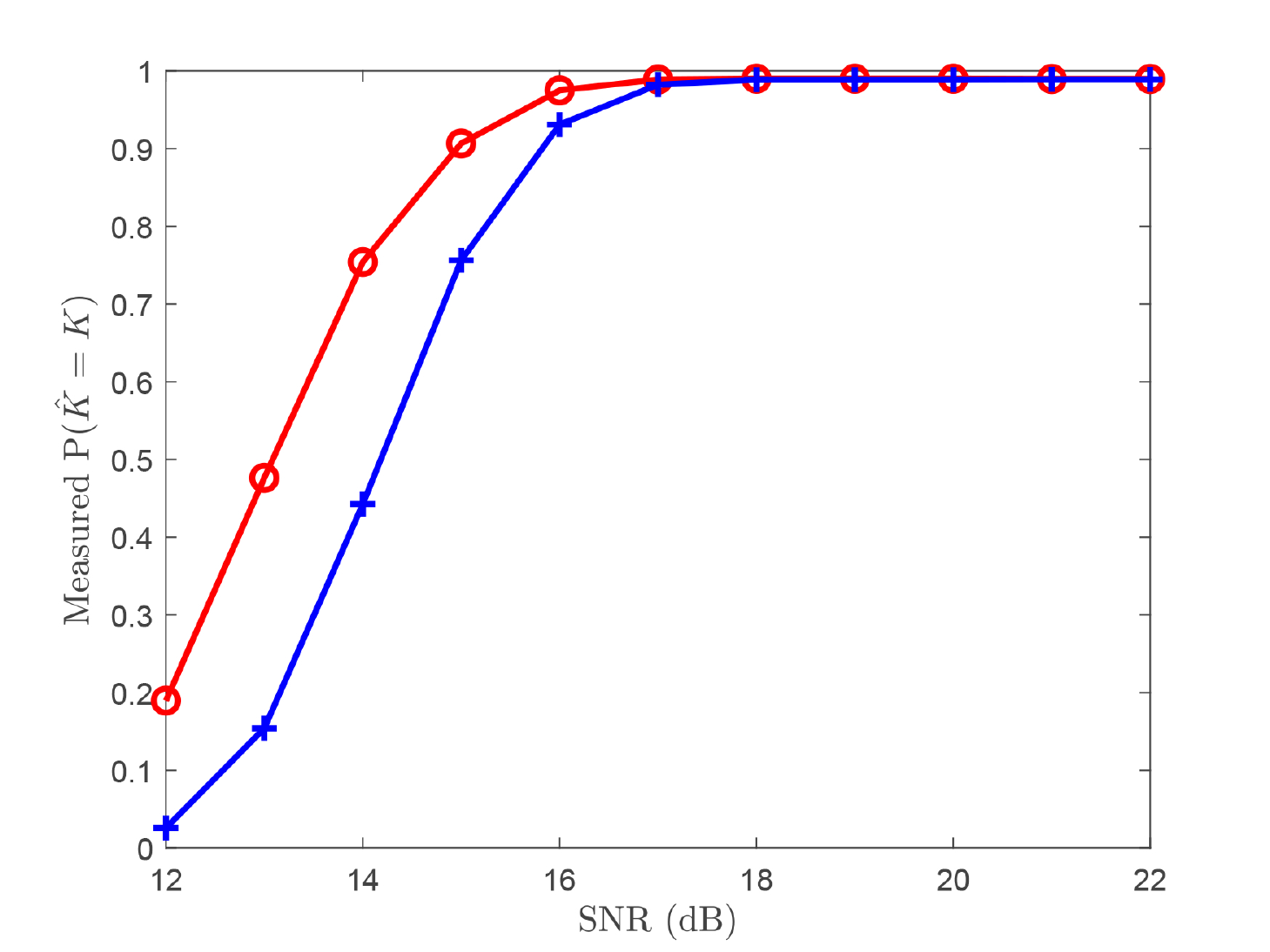}}
  \subfigure[]{
  \label{MSEvsSNR_fig}
  \includegraphics[width=50mm]{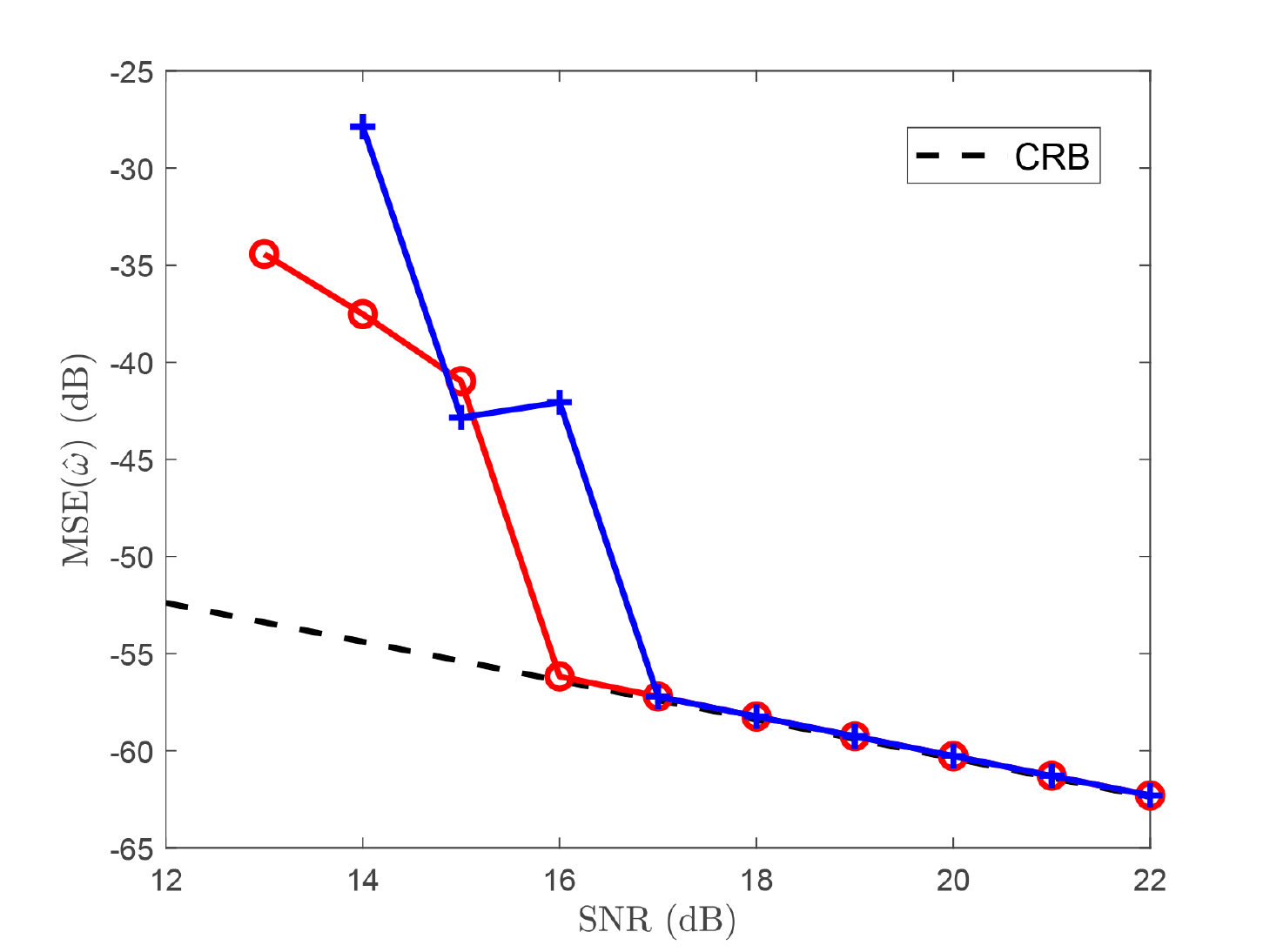}}
  \subfigure[]{
  \label{RMSEvsSNR_fig}
  \includegraphics[width=50mm]{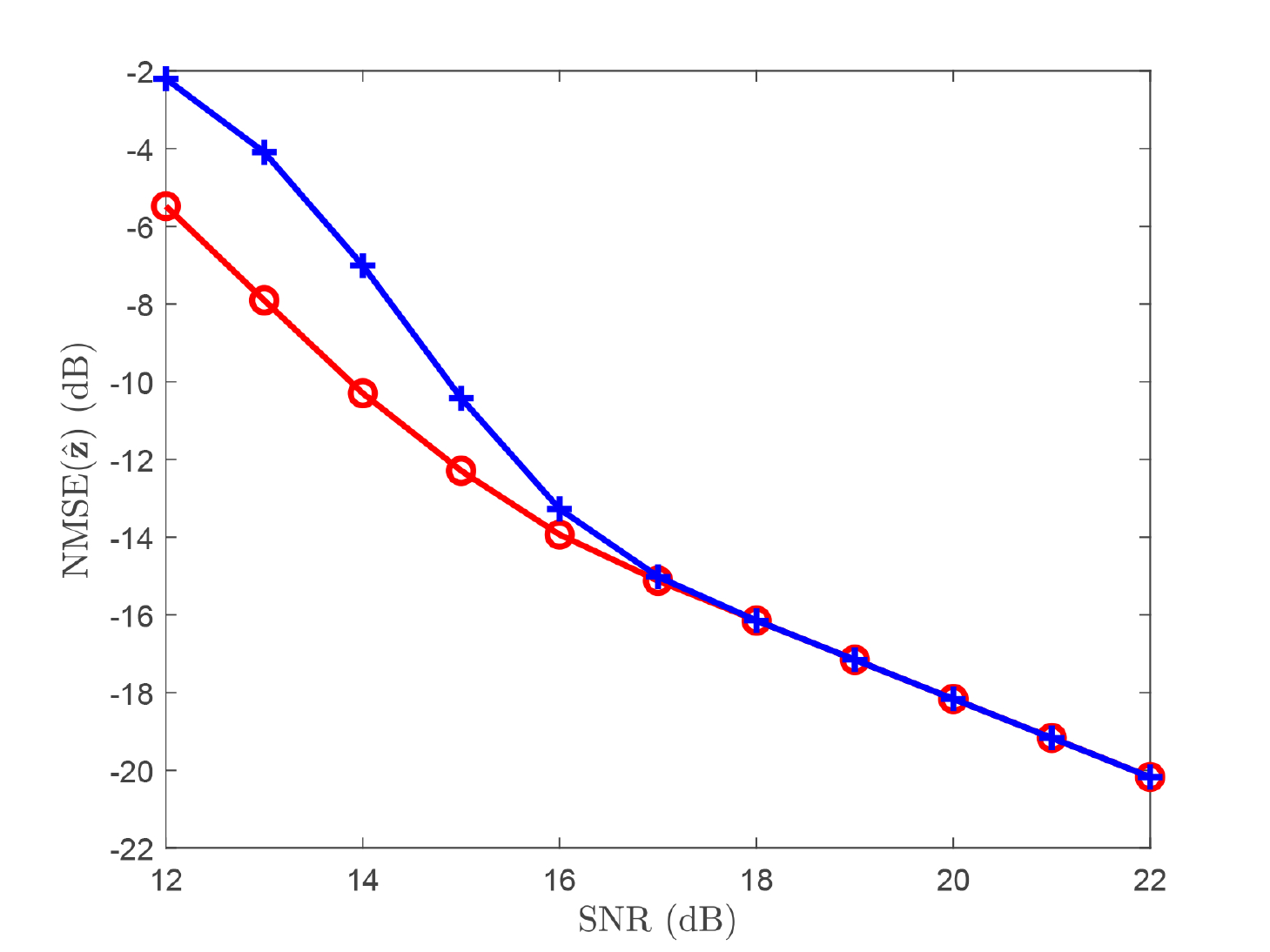}}
  \caption{The performances of NOMP, NOMP-CFAR versus SNR. (a) The false alarm rate $\bar{\rm P}_{\rm FA}$, (b) the detection probability $\bar{\rm P}_{\rm D}$ of all the targets, (c) the correct model order estimation rate ${\rm P}(\hat{K} = K)$, (d) the MSE of the frequency estimation error $\|\hat{\omega}-\omega\|_2^2 / K$, (e) the normalized reconstruction error ${\rm NMSE}(\hat{\mathbf z})=\|\hat{\mathbf z}-{\mathbf z}\|^2/\|{\mathbf z}\|_2^2$.}
  \label{vsSNR_all_fig} 
\end{figure*}

Fig. \ref{POE_fig} shows the false alarm rate $\bar{\rm P}_{\rm FA}$ versus ${\rm SNR}$. At low SNR and ${\rm SNR}\leq 16$ dB, the $\bar{\rm P}_{\rm FA}$ of the NOMP tends to decrease with ${\rm SNR}$ and NOMP-CFAR increase. For ${\rm SNR} \geq 16$ dB, the $\bar{\rm P}_{\rm FA}$ of the two algorithms are close to the nominal $\bar{\rm P}_{\rm FA}$, demonstrating the CFAR property of the two algorithms. For $\bar{\rm P}_{\rm D} = 0.75$, the corresponding SNRs of NOMP and NOMP-CFAR are $14$ dB and $15$ dB, respectively, which demonstrates that NOMP-CFAR has $1$ dB performance loss due to the lacking of knowledge of noise variance.

Fig. \ref{Pequal_fig} shows the model order estimation probability ${\rm P}({\hat{K} = K})$ of both algorithms increase as SNR increases.
The model order estimation probability ${\rm P}({\hat{K} = K})$ of NOMP is higher than that of NOMP-CFAR.
Fig. \ref{MSEvsSNR_fig} and Fig. \ref{RMSEvsSNR_fig} show the estimation errors. For the frequency estimation error, the Cram\'{e}r Rao bound (CRB) is also evaluated for the one target scene. Fig. \ref{MSEvsSNR_fig} shows that the MSEs of NOMP and NOMP-CFAR closely approach to the CRB computed for the true model order at ${\rm SNR} \geq 16$ dB and ${\rm SNR}\geq 17$ dB, respectively. The reason that NOMP-CFAR deviates away from CRB at ${\rm SNR}=16$ dB is that for a frequency lying off the DFT frequency grid, the frequency is not detected. Meanwhile, a wrong frequency is detected which yields a false alarm, causing a large frequency estimation error. For the reconstruction error of the line spectral, NOMP performs better than NOMP-CFAR for ${\rm SNR}\leq16$ dB and performs similarly as SNR increases.

\subsection{Performance versus the Number of Snapshots in the MMV setting}
The false alarm rate and the detection probability of the NOMP-CFAR in the MMV setting is investigated. The SNRs of all the targets are set as $10$ dB. The false alarm rate and the detection probability are shown in Fig. \ref{simulationinMMV}. From Fig. \ref{MeaPoevs S}, the false alarm probability of NOMP is a little lower than that of NOMP-CFAR, and the false alarm rates of both algorithms are close to the nominal false alarm rate for $S\geq 3$. Fig. \ref{PDvs S} shows that the detection probabilities of both algorithms are close to each other and increase as the number of snapshots increases.

\begin{figure}
  \centering
  \subfigure[]{
  \label{MeaPoevs S}
  \includegraphics[width=65mm]{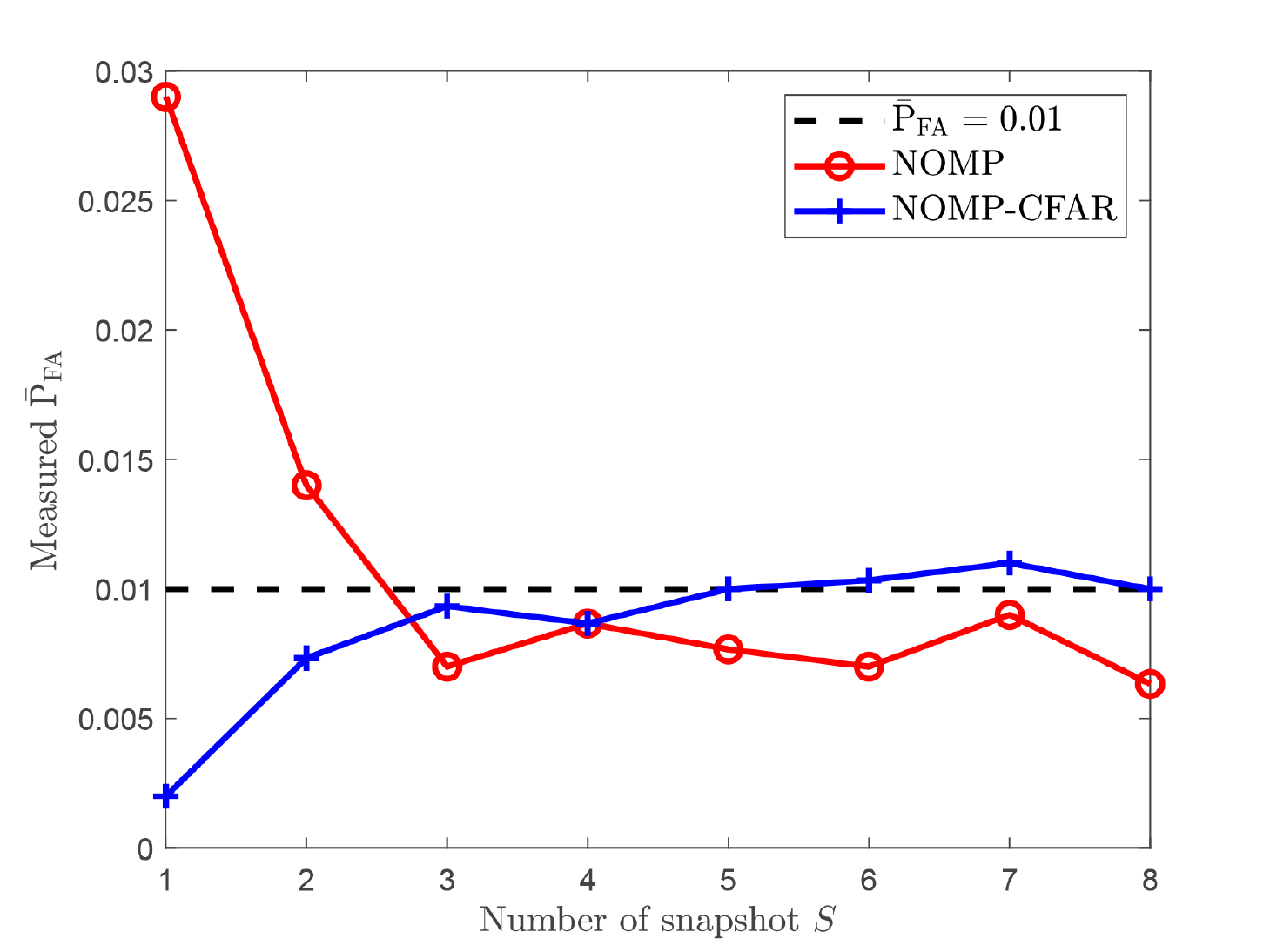}}
  \subfigure[]{
  \label{PDvs S}
  \includegraphics[width=65mm]{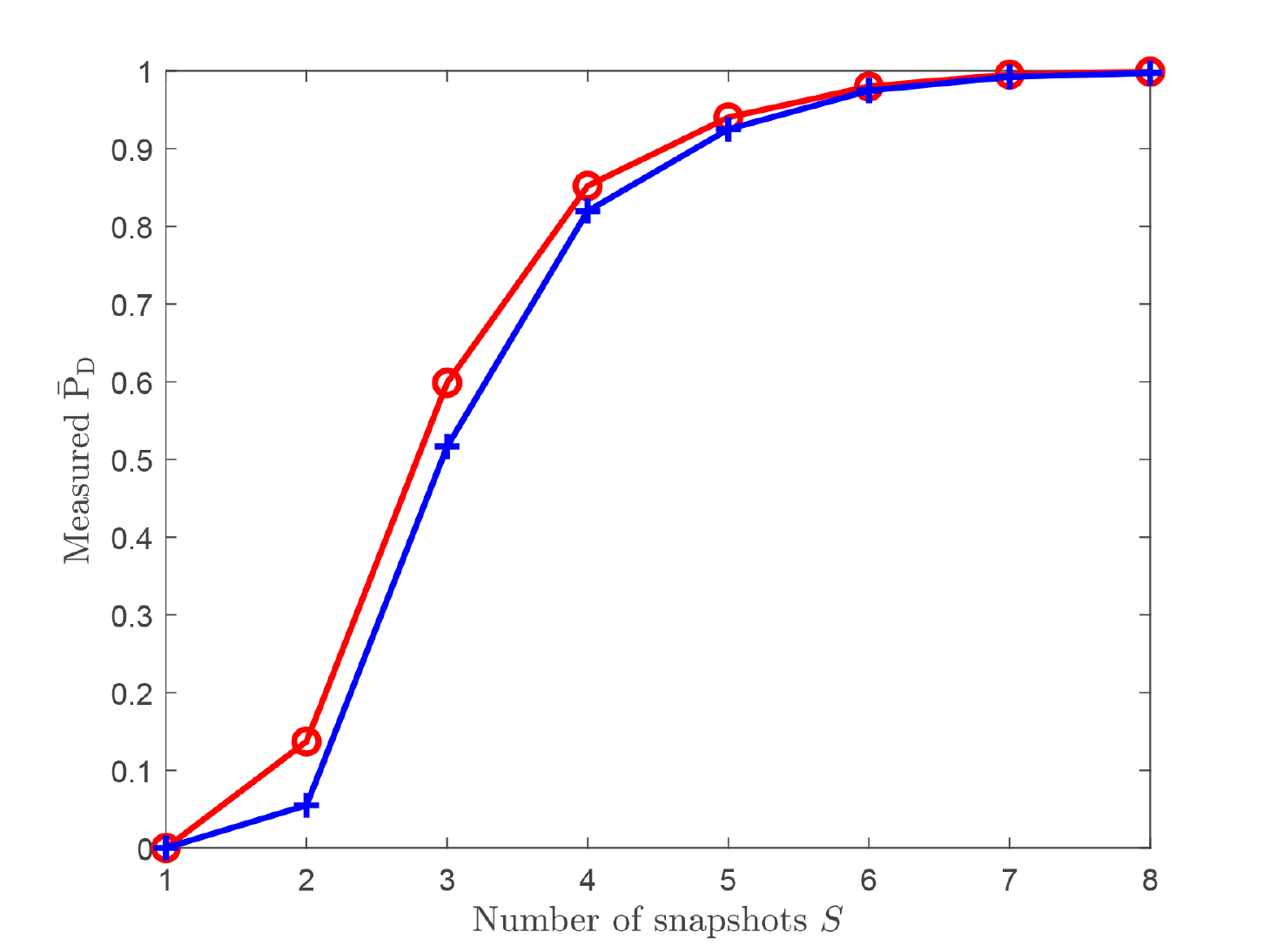}}
  \caption{The measured $\bar{\rm P}_{\rm FA}$ (a) and $\bar{\rm P}_{\rm D}$ (b) of the two algorithm in different numbers of snapshots. }\label{simulationinMMV}
\end{figure}

\subsection{Performance versus the Compression Rate in the Compressive Setting}
The performance of NOMP-CFAR versus the compression rate $M/N$ is investigated. The elements of the compression matrix is drawn i.i.d. from the complex Bernoulli distribution. The number of frequencies is $K=8$, the SNRs of all the targets are $22$ dB. Results are shown in Fig. \ref{simulationvsMinCom}.

\begin{figure}
  \centering
  \subfigure[]{
  \label{MeaPoevs M}
  \includegraphics[width=65mm]{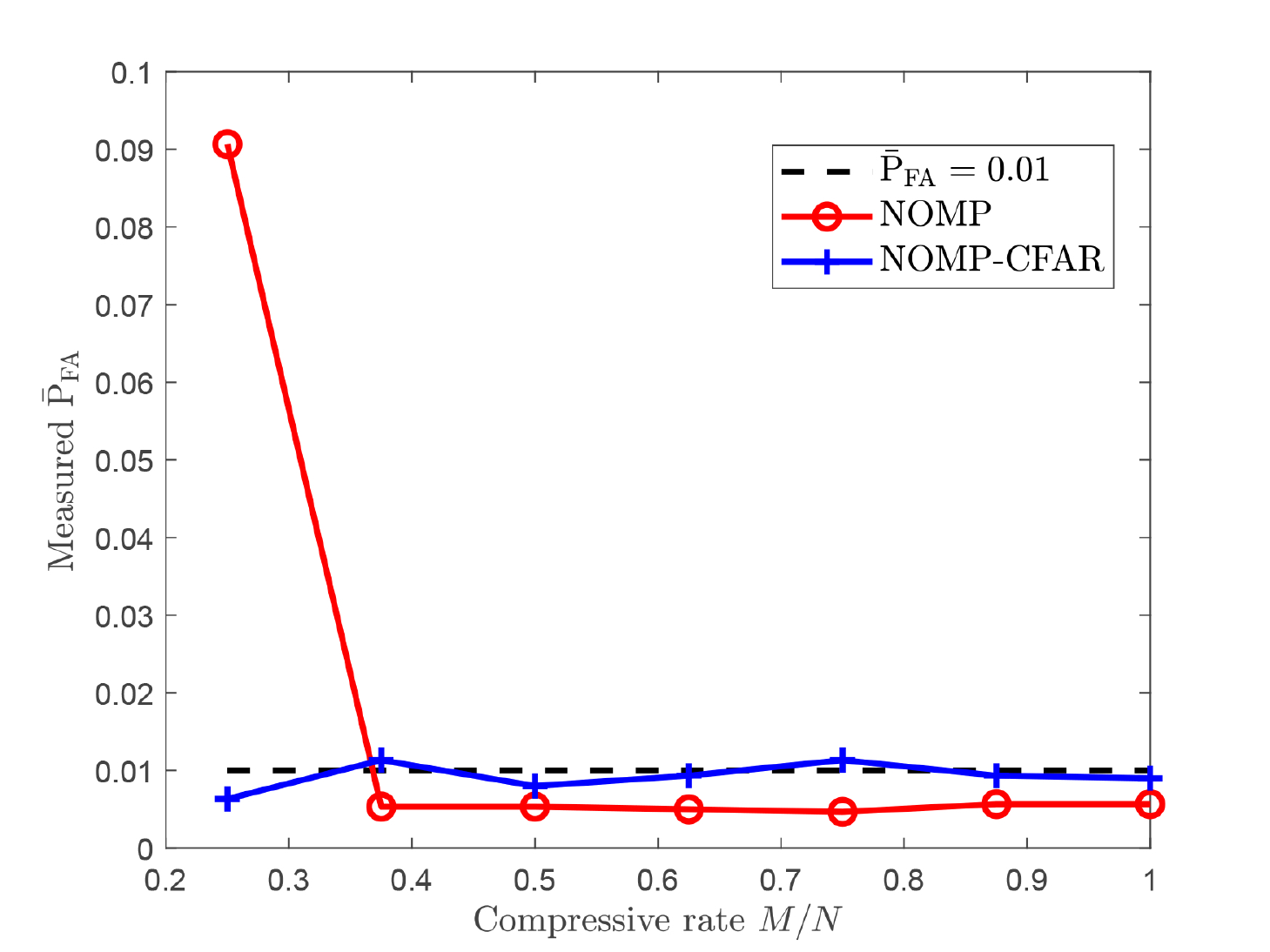}}
  \subfigure[]{
  \label{PDvs M}
  \includegraphics[width=65mm]{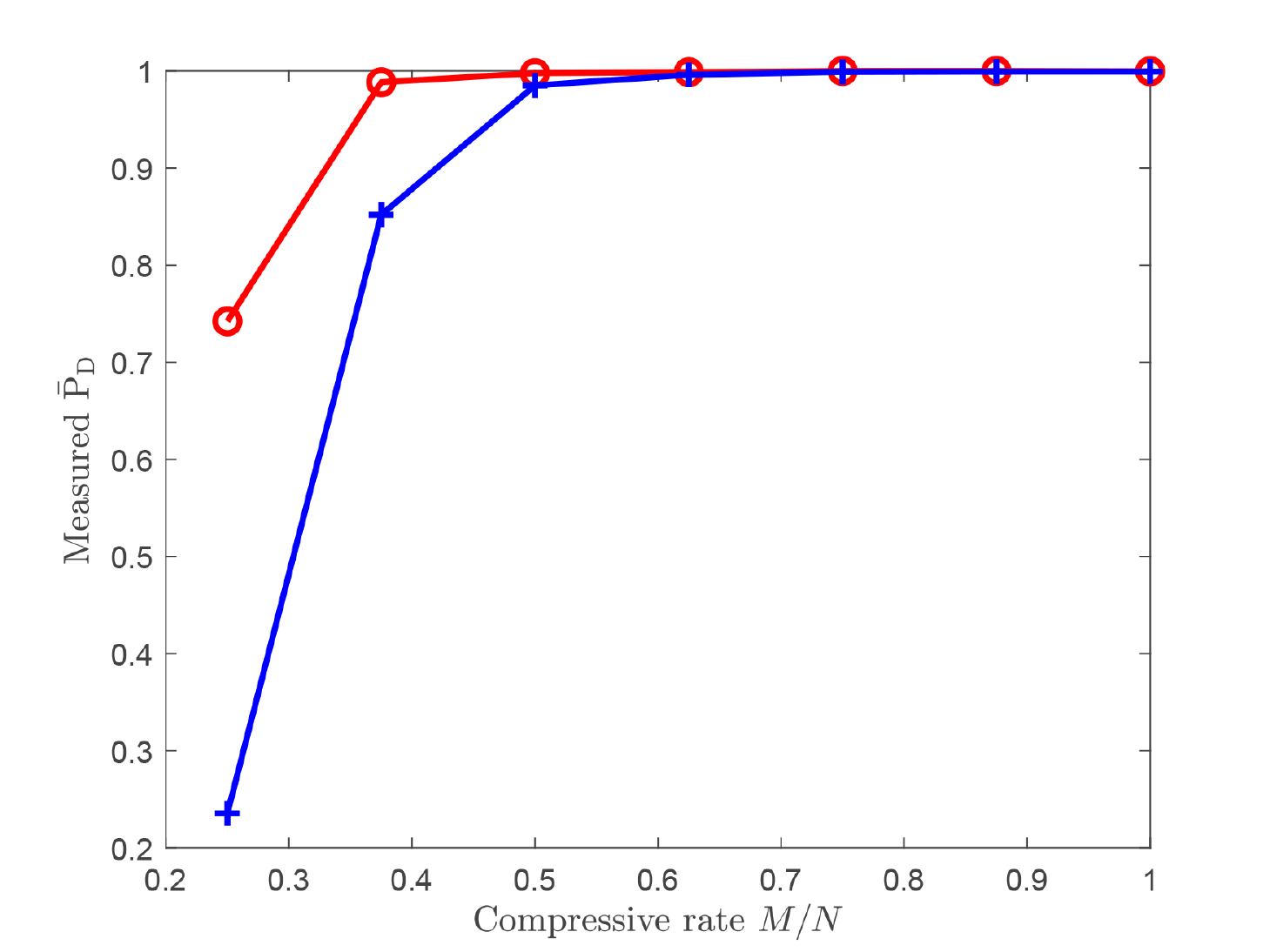}}
  \caption{The measured $\bar{\rm P}_{\rm FA}$ (a) and $\bar{\rm P}_{\rm D}$ (b) of the two algorithm in different compressive rate. }\label{simulationvsMinCom}
\end{figure}
From Fig. \ref{MeaPoevs M}, the false alarm rates of NOMP is a little lower than that of NOMP-CFAR, and the false alarm rates of both algorithms are close to the nominal false alarm rate for $M/N\geq 0.36$. For the detection probability $\bar{\rm P}_{\rm D}$, Fig. \ref{simulationvsMinCom} shows that NOMP performs better than NOMP-CFAR for $M/N\leq 0.375$, and the $\bar{\rm P}_{\rm D}$ of both algorithms approach to $1$ as $M/N\geq 0.5$. For $\bar{\rm P}_{\rm D} = 0.74$, the compression rates required for NOMP and NOMP-CFAR are $0.24$ and $0.32$, respectively, meaning that to achieve $\bar{\rm P}_{\rm D}= 0.74$, NOMP-CFAR needs $4/3$ times the number of measurements of NOMP.

\subsection{The False Alarm Probability versus Strength of Noise Fluctuation}
The nominal SNR of the $k$th target is defined as ${\rm SNR}_k\triangleq 10\log\left(N|x_k|^2/\sigma_0^2\right)$, where $\sigma_0^2$ is the nominal noise variance and is set as $\sigma_0^2=1$, i.e., $0$ dB. The SNR of each frequency is $28$ dB.
For each MC trial, the noise variance is drawn uniformly from $[-u, u]$ in dB, where $u$ characterizes the strength of noise fluctuation. The results are shown in Fig. \ref{PFAvsSigmavar}. Note that the detection probability of all the targets are very close to $1$ ($\geq 0.9993$) and are omitted here. It can be seen that NOMP violates the CFAR property in varied noise variance scenario, while NOMP-CFAR still preserves the CFAR property.
\begin{figure}
  \centering
  \includegraphics[width=65mm]{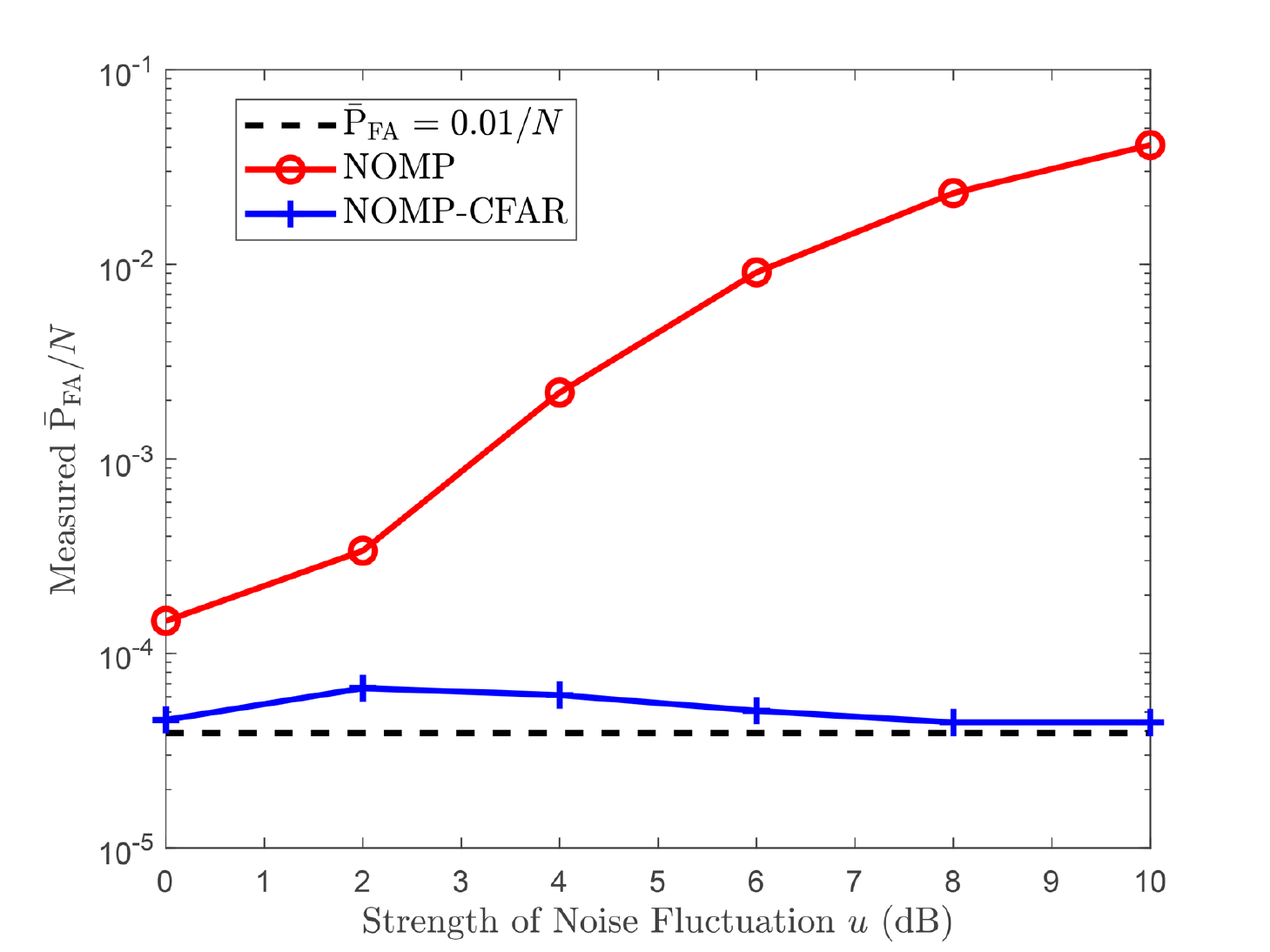}
  \caption{The measured $\bar{\rm P}_{\rm FA}$ versus the strength of the noise fluctuation.}\label{PFAvsSigmavar}
\end{figure}

\section{Real Experiment}\label{Real Experiment}
In this section, the performance of proposed NOMP-CFAR is demonstrated using an IWR1642 radar \cite{TImmwave}. The IWR1642 radar is an FMCW MIMO radar consisting of two transmitters and four receivers. The radar parameters and wave form specifications are listed in Table \ref{tab:Radar_parameters}. For the single input multiple output (SIMO) mode, the base band signal model can be formulated as
\begin{align}\label{MultipleTargetModel}
y(n,m,l) = \sum_{k = 1}^{K} x_k {\text{e}}^{\text{j} [(n - 1) {\omega}_{x, k} + (m - 1) {\omega}_{y, k} + (l - 1) {\omega}_{z, k} ]} + \mathcal{\epsilon}_{n,m,l},
\end{align}
$n=0,1,\cdots,N-1$, $m=0,1,\cdots,M-1$, $l=0,1,\cdots,L-1$, where $n$, $m$, $l$ denote the index of the fast time domain, slow time domain and spatial domain, respectively, $\mathcal{\epsilon}_{n,m,l}$ denotes the noise, $\omega_{x,k}$, $\omega_{y,k}$, $\omega_{z,k}$ are
\begin{subequations}\label{FreqToState}
\begin{align}
{\omega}_{x, k} &= \frac{4 \pi}{c} (f_c v_k + \mu r_k) T_{\text s} \approx \frac{4 \pi}{c} \mu r_k T_{\text s}, \label{fast time omega} \\
{\omega}_{y, k} &= \frac{4 \pi}{c} f_c v_k {T}_{\text{r}}, \label{slow time omega} \\
{\omega}_{z, k} &= \frac{2 \pi}{c} f_c d \sin \theta_k, \label{spatial omega}
\end{align}
\end{subequations}
where $c$ denotes the speed of electromagnetic wave. Note that (\ref{MultipleTargetModel}) reduces to the range estimation, range azimuth and range-Doppler imaging problem, with $M=L=1$, $M=1$ and $L=1$, respectively. Thus, one could apply the NOMP-CFAR algorithm to solve the multidimensional line spectrum estimation problem (\ref{MultipleTargetModel}), and obtain the estimates of ranges, velocities and azimuths via (\ref{FreqToState}).

The proposed NOMP-CFAR is compared with the traditional CFAR and NOMP. The false alarm rates of NOMP-CFAR and NOMP are set as $\bar{\rm P}_{\rm FA} = 10^{-2}$, and the false alarm rate of CFAR is $\bar{\rm P}_{\rm FA, tr} = 10^{-2} / N_{\rm tot}$ to ensure that false alarms produced by the three algorithms are comparable when $\bar{\rm P}_{\rm FA}$ is small, where $N_{\rm tot}$ denotes the total number of cells under test.
The upper bound of target number for NOMP-CFAR algorithm is $K_{\rm max} = 32$ unless stated otherwise.

\newcommand{\tabincell}[2]{\begin{tabular}{@{}#1@{}}#2\end{tabular}}
\begin{table}[!t]
  \centering
  \scriptsize
  \caption{Parameters Setting of the Experiment}
  \label{tab:Radar_parameters}
  \begin{tabular}{ll}
    \\[-2mm]
    \hline
    \hline\\[-2mm]
    { \small Parameters}&\qquad {\small Value}\\
    \hline
    \vspace{1mm}\\[-3mm]
    Number of Receivers $L$      &   4 \\
    \vspace{1mm}
    Carrier frequency $f_c$   &   77 GHz\\
    \vspace{1mm}
    Frequency modulation slope $\mu$          &  $29.982$ MHz/s\\
    \vspace{1mm}
    sweep time $T_p$         &  $60 {\rm \mu s}$\\
    \vspace{1mm}
    Pulse repeat interval $T_{\text{r}}$         &  $160 {\rm \mu s}$\\
    \vspace{1mm}
    Bandwidth $B$          &  1798.92 MHz\\
    \vspace{1mm}
    Sampling frequency $f_s = 1 / T_s$         &  10 MHz\\
    \vspace{1mm}
    Number of pulses in one CPI $M$          & 128\\
    \vspace{1mm}
    Number of fast time samples $N$          &  256 \\
    \hline
    \hline
  \end{tabular}
\end{table}

In the following, three experiments are conducted to evaluated the performance of NOMP-CFAR.

\subsection{Experiment $1$}
Fig. \ref{people2Scene} shows the setup of field experiment $1$. The radial distances and the azimuths of the two static people named people $1$ and people $2$ are about ($4.92$ m, $0^{\circ}$) and ($3.09$ m, $-19.8^{\circ}$). First, we conduct range estimation only via the CFAR, NOMP and NOMP-CFAR approaches. Then, we implement NOMP and NOMP-CFAR for range and azimuth estimation.

The number of the fast domain samples is $N = 256$. The number of guard cells and the number of training cells are $8$ and $60$, respectively. The CA-CFAR is adopted. The required threshold multipliers of the CFAR and NOMP-CFAR are $\alpha_{\rm tr} = 11.06$ and $\alpha = 11.03$, corresponding to $10.44$ dB and $10.43$ dB, respectively. Fig. \ref{people2 range FFT result} shows that the noise variance is about $24$ dB, and we set the noise variance $\sigma^2 =  251$ for NOMP algorithm, and the corresponding threshold can be calculated as $2.55 \times 10^3$, i.e. $34.06$ dB.
The range estimation results are shown in Fig. \ref{people2RangeResult}. Fig. \ref{people2 range FFT result} shows that CFAR detects people $1$ and people $2$, which estimates their radial distance as $4.90$ m and $3.14$ m. From Fig. \ref{people2 range NOMP result}, NOMP detects both people $1$ and people $2$.
It is concluded that people $1$ and people $2$ are extended targets, and each two detected results correspond to people $1$ and people $2$.
The highest amplitudes corresponding to people $1$ and people $2$ are $48.86$ dB and $51.17$ dB, respectively.
Equivalently, the amplitudes of people $1$ and people $2$ are $14.80$ dB and $17.11$ dB above the NOMP thresholds.
Fig. \ref{people2 range NOMPCFAR result} shows that the detection results of NOMP-CFAR are similar to that of NOMP. In details, the highest amplitudes corresponding to people $1$ and people $2$ are $48.83$ dB and $51.16$ dB, respectively. In addition, the thresholds are $37.30$ dB and $37.39$ dB, respectively. Equivalently, the amplitudes of people $1$ and people $2$ are $11.53$ dB and $13.77$ dB above the corresponding thresholds, respectively.

\begin{figure}
  \centering
  \includegraphics[width = 2.7in]{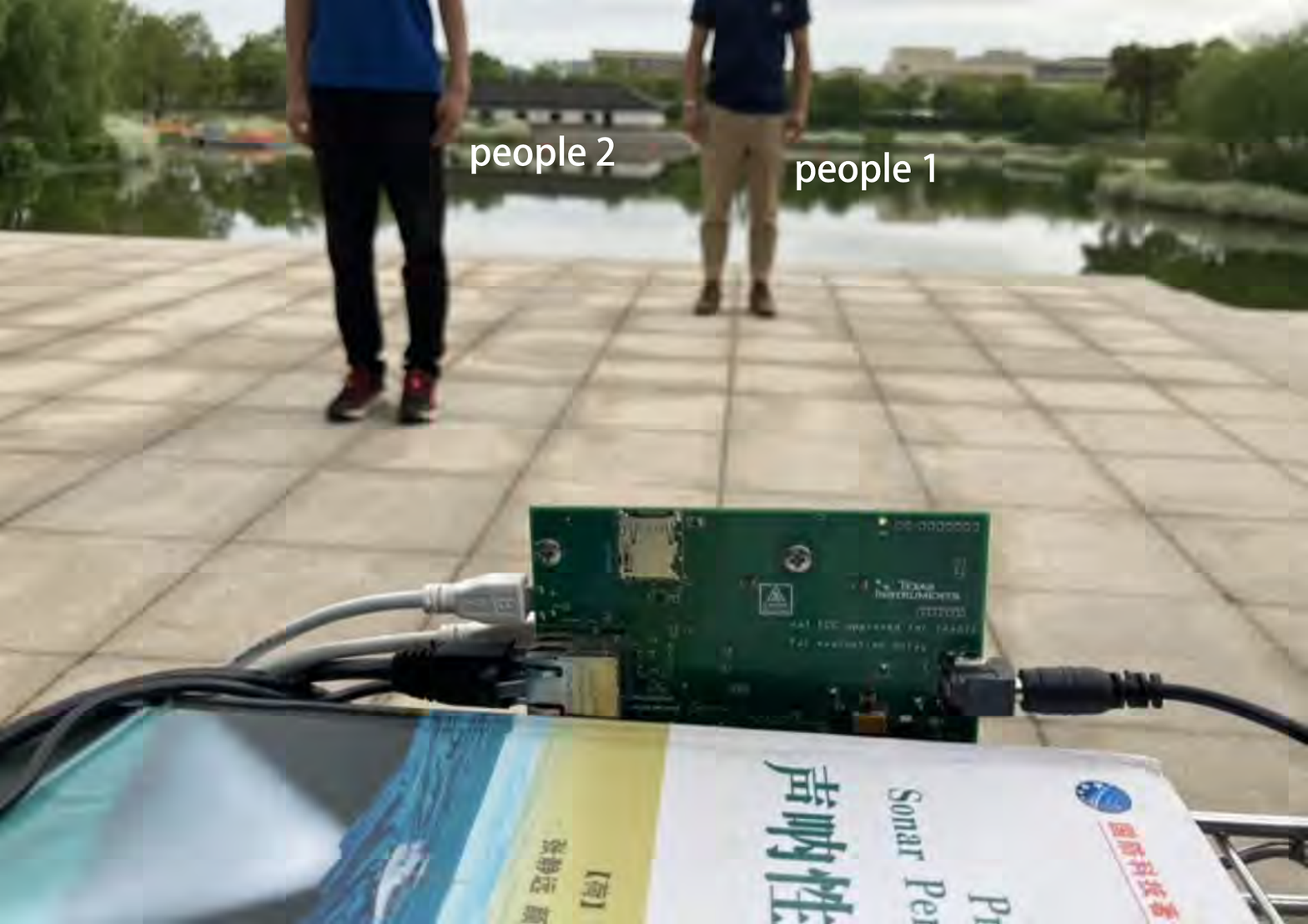}
  \caption{Field experiment $1$ setup.}\label{people2Scene}
\end{figure}

\begin{figure}
  \centering
  \subfigure[]{
  \label{people2 range FFT result}
  \includegraphics[width=50mm]{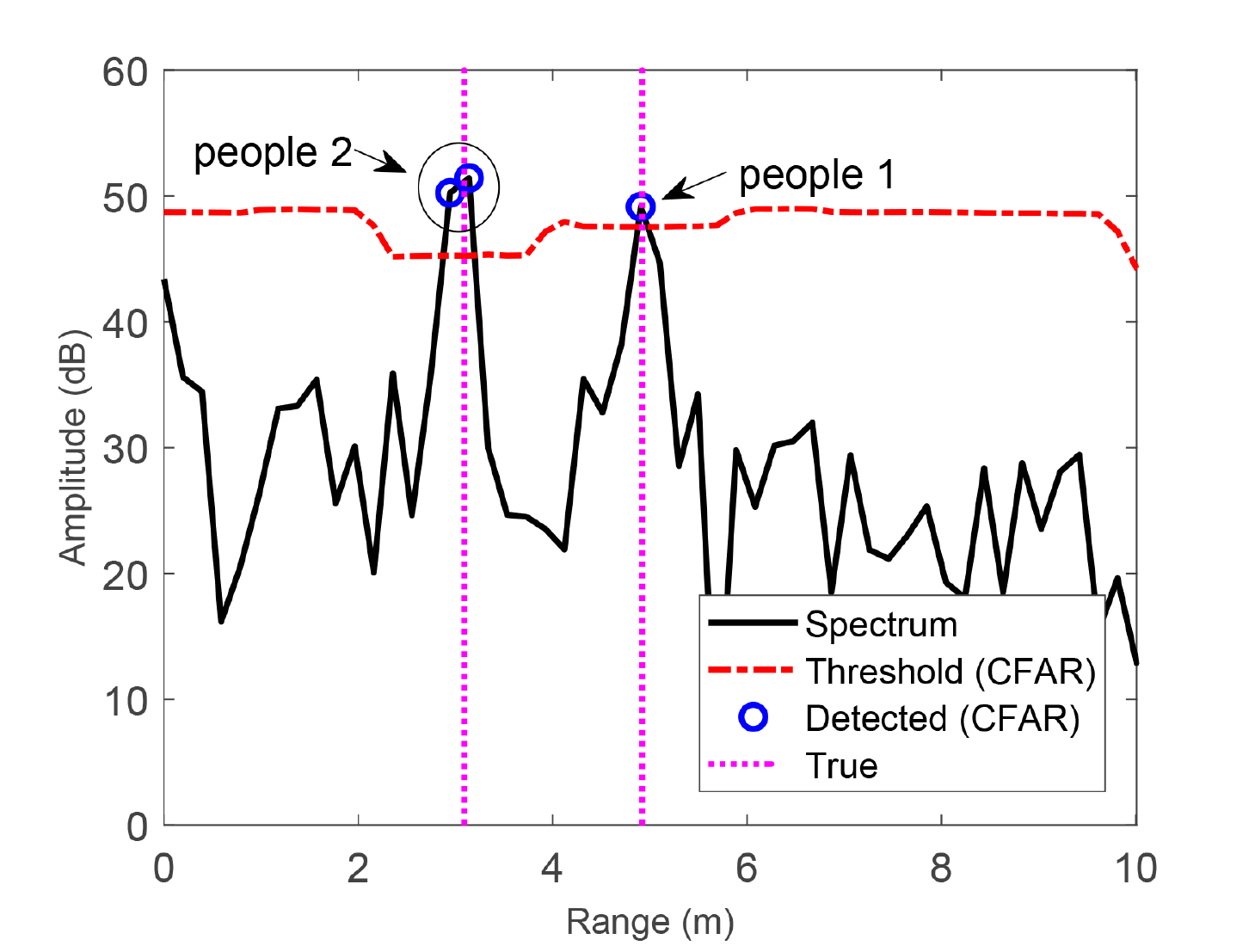}}
  \subfigure[]{
  \label{people2 range NOMP result}
  \includegraphics[width=50mm]{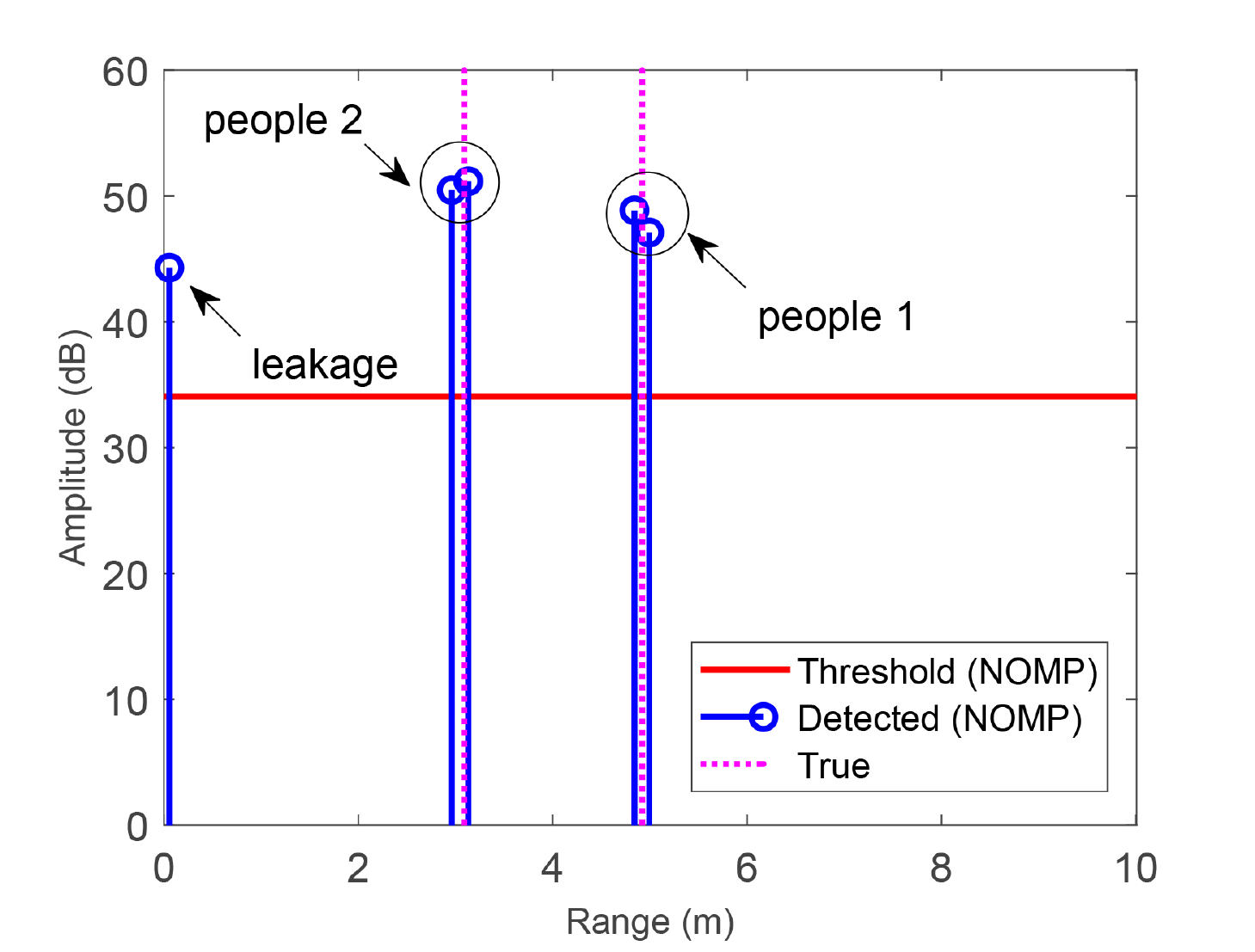}}
  \subfigure[]{
  \label{people2 range NOMPCFAR result}
  \includegraphics[width=50mm]{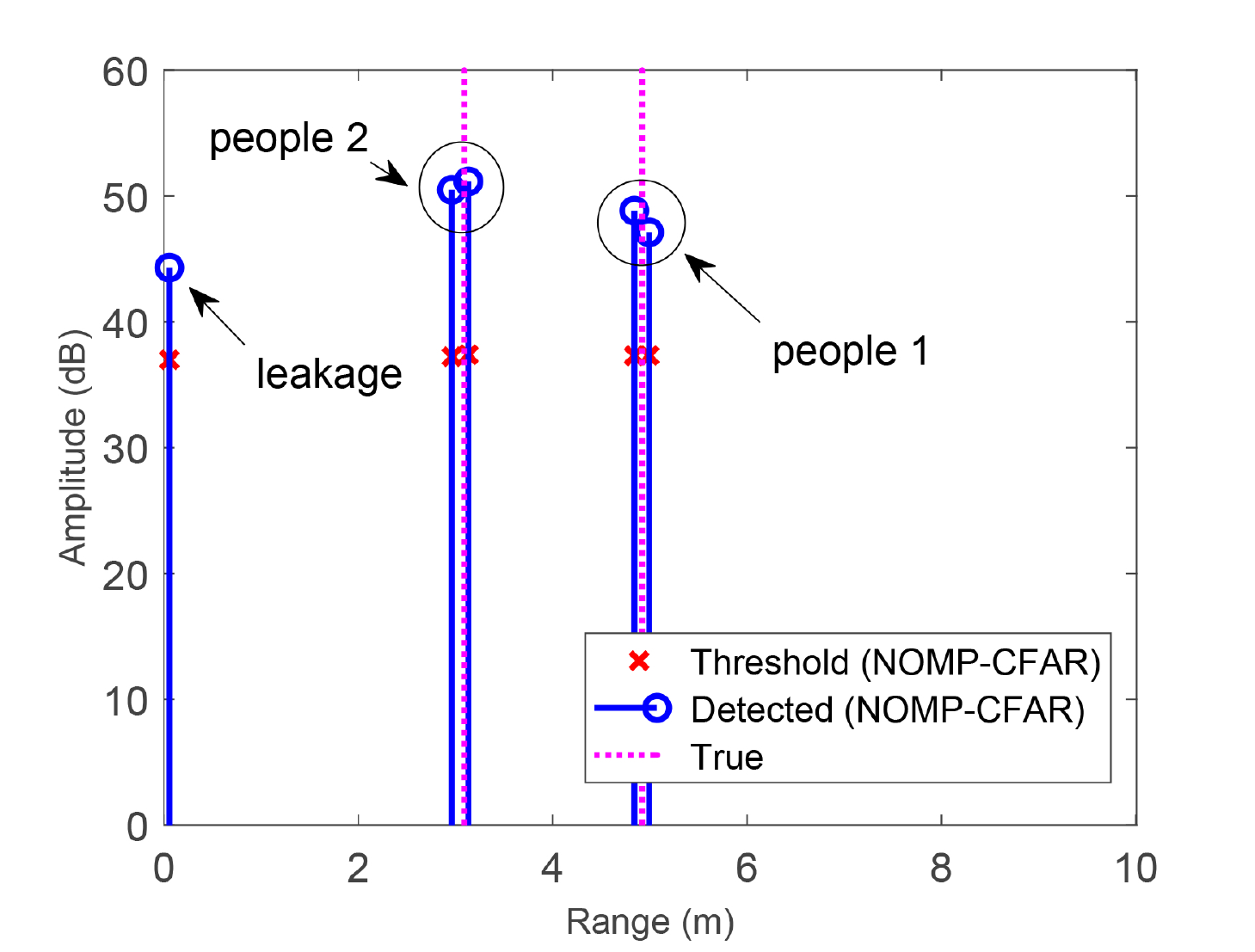}}
  \caption{Range estimation results in experiment $1$. (a): CFAR, (b): NOMP, (c): NOMP-CFAR.}\label{people2RangeResult}
\end{figure}

We then implement the two dimensional NOMP and NOMP-CFAR to perform the range and azimuth estimation, the results are shown in Fig. \ref{people2RangeAzimuthResult}. Fig. \ref{people2RangeAzimuthNOMP} shows that NOMP detects $19$ points and the highest $2$ points are estimated as $(4.84 {\rm m}, -0.12^{\circ})$ and $(3.09 {\rm m}, -22.28^{\circ})$, which are close to the truth of people $1$ and people $2$, respectively.
And their integrated amplitudes are $54.16$ dB, $55.38$ dB, respectively. Fig. \ref{people2RangeAzimuthNOMPCFAR} shows that the number of detected points of NOMP-CFAR is $2$ and the results are $(4.85 {\rm m}, 0.50^{\circ}, 54.21{\rm dB}, 42.30{\rm dB})$, $(3.20 {\rm m}, -21.61^{\circ}, 53.29{\rm dB}, 51.42 {\rm dB})$, where the first, second, third and fourth components are the radial distance, the azimuth, reconstructed amplitude and threshold. The first detected points correspond to people $1$ and the second points correspond to people $2$. The estimations of radial distances and azimuths of people $1$ and people $2$ are close to the truth. This demonstrates that NOMP-CFAR suppresses the false alarm significantly, compared to NOMP.

\begin{figure*}
  \centering
  \subfigure[]{
  \label{people2RangeAzimuthNOMP} 
  \includegraphics[width = 65mm]{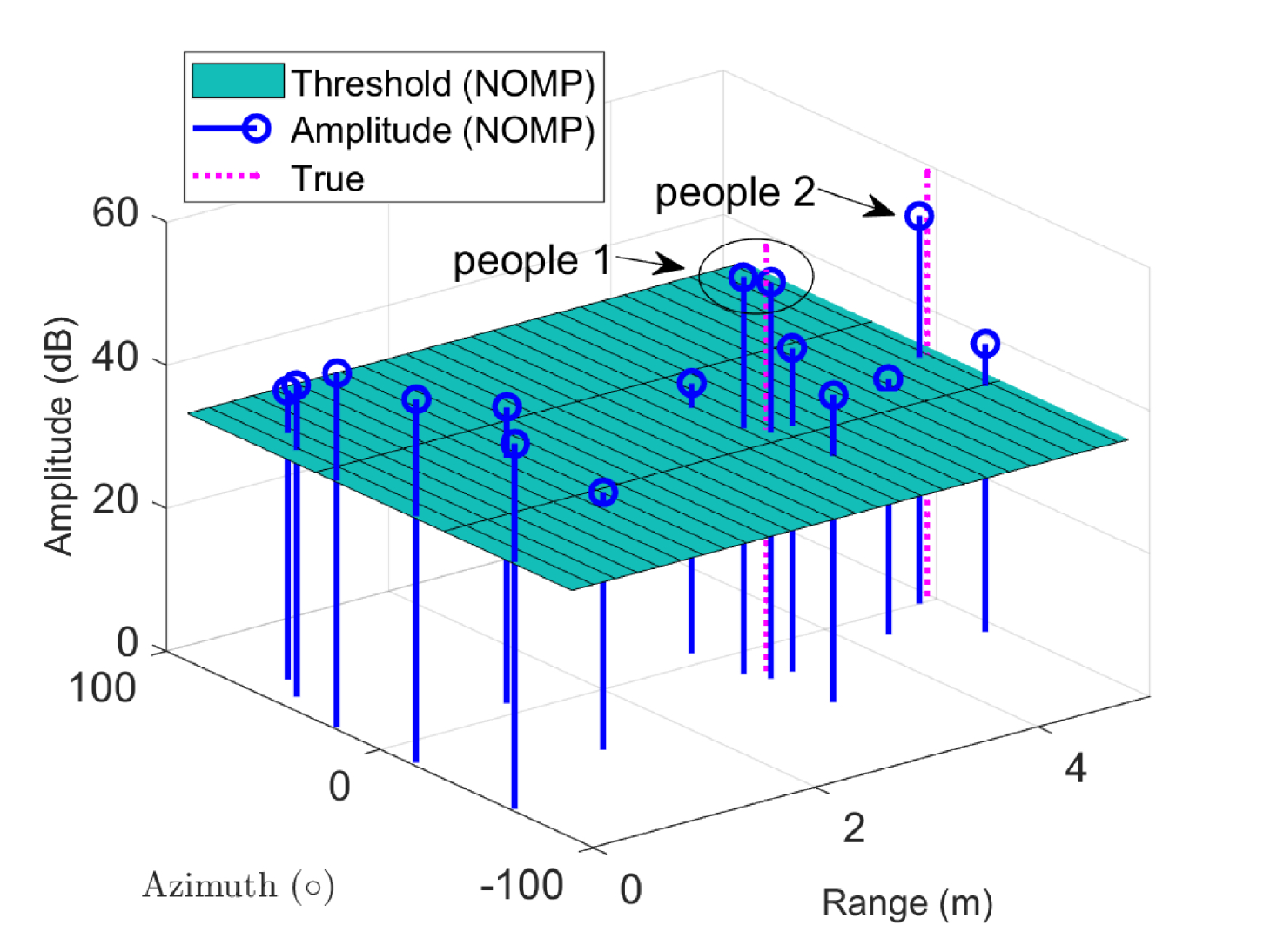}}
  \subfigure[]{
  \label{people2RangeAzimuthNOMPCFAR} 
  \includegraphics[width = 65mm]{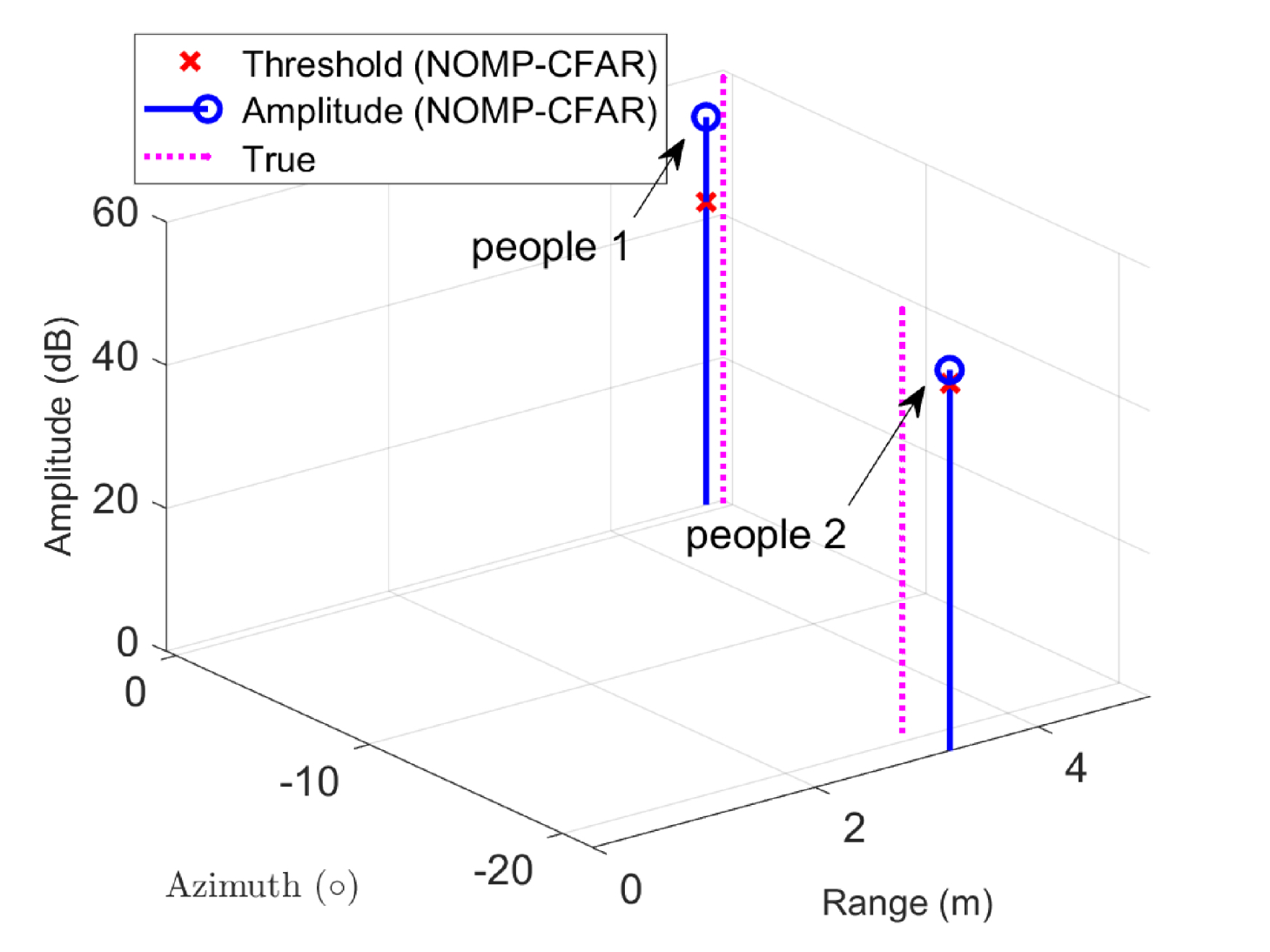}}
 \caption{Range and azimuth estimation in experiment $1$. (a) NOMP, (b) NOMP-CFAR.}
 \label{people2RangeAzimuthResult}
\end{figure*}

\subsection{Experiment $2$}
The setup of field experiment $2$ is shown in Fig. \ref{movingPeopleScene}. The radial distances of the two static people named people $3$ and people $2$ are about $(4.87~{\rm m}, 0^{\circ})$ and $(2.63 ~{\rm m}, 24^{\circ})$. A cyclist moves toward the radar with the radial distance starting from $7$ m to $2$ m and the velocity about $2$ m/s. The numbers of fast time domain samples and slow time domain samples are $N = 128$ and $M = 64$, respectively. Two dimensional CA-CFAR is adopted. The widths of guard band along the fast time dimension and slow time dimension are $3$, the widths of training band along the fast time dimension and slow time dimension are $5$.
Results are shown in Fig. \ref{moving_RangeDopplerResult}. Fig. \ref{movingPeopleCFAR} shows that people $2$ and cyclist are detected, and people $1$ is missed by CFAR. In addition, CFAR also produces $1$ false alarm. From Fig. \ref{movingPeopleNOMP}, the three targets including people $1$, people $2$ and cyclist are detected by NOMP. The total number of points detected by NOMP is $47$. Fig. \ref{movingPeopleNOMPCFAR} shows the detection result of NOMP-CFAR algorithm. The total number of points detected by NOMP-CFAR is $24$. Note that for people $1$, the reconstructed amplitude and the threshold output by NOMP are $61.17$ dB and $42.51$ dB, respectively. And the reconstructed amplitude of people $1$ are $18.66$ dB above the corresponding thresholds, which is of high confidence of the existence of the people $1$.
\begin{figure}
  \centering
  \includegraphics[width=65mm]{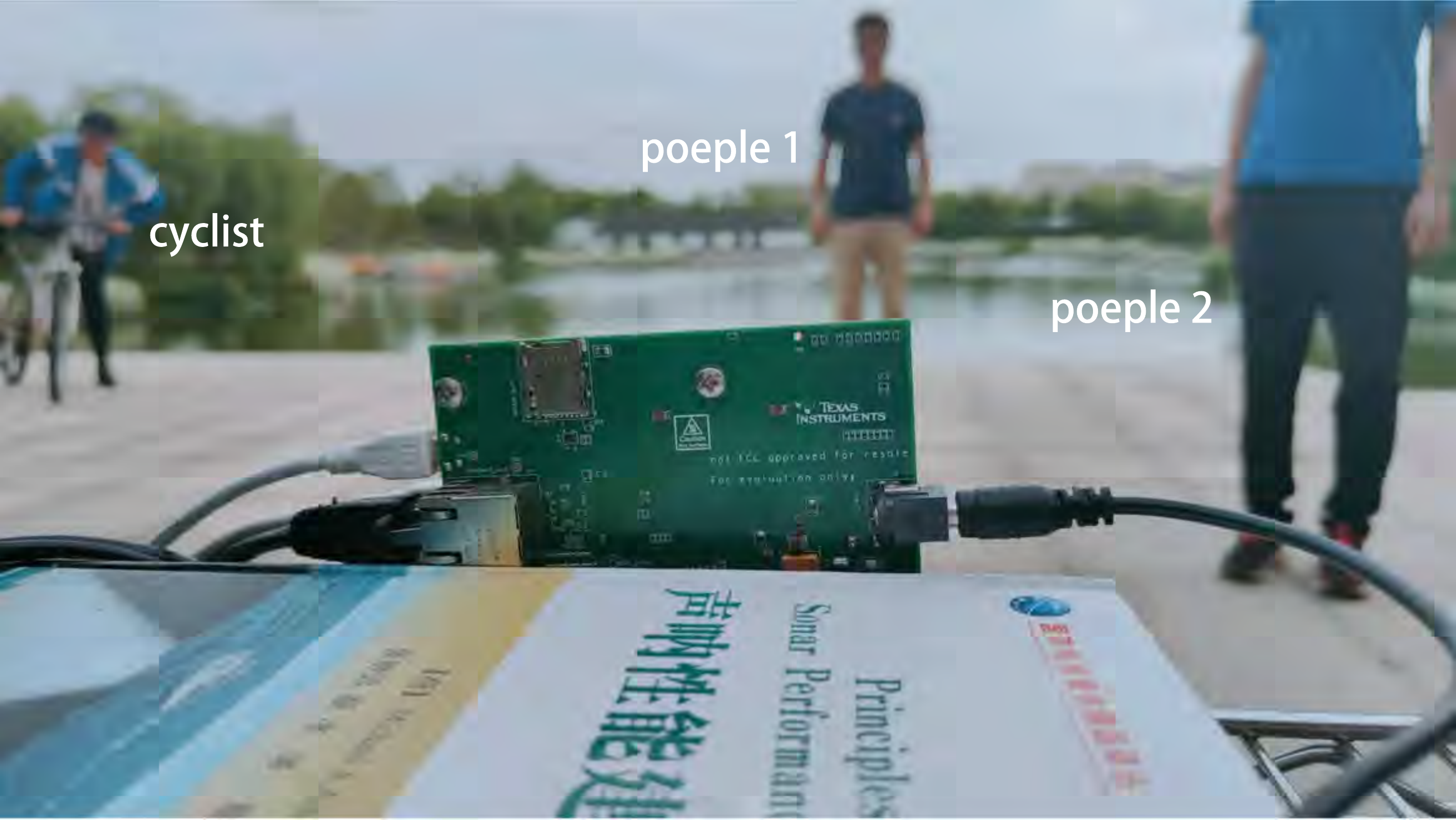}
  \caption{Field experiment 2 setup.}\label{movingPeopleScene}
\end{figure}
\begin{figure*}
\centering
  \subfigure[]{
  \label{movingPeopleCFAR}
  \includegraphics[width = 50mm]{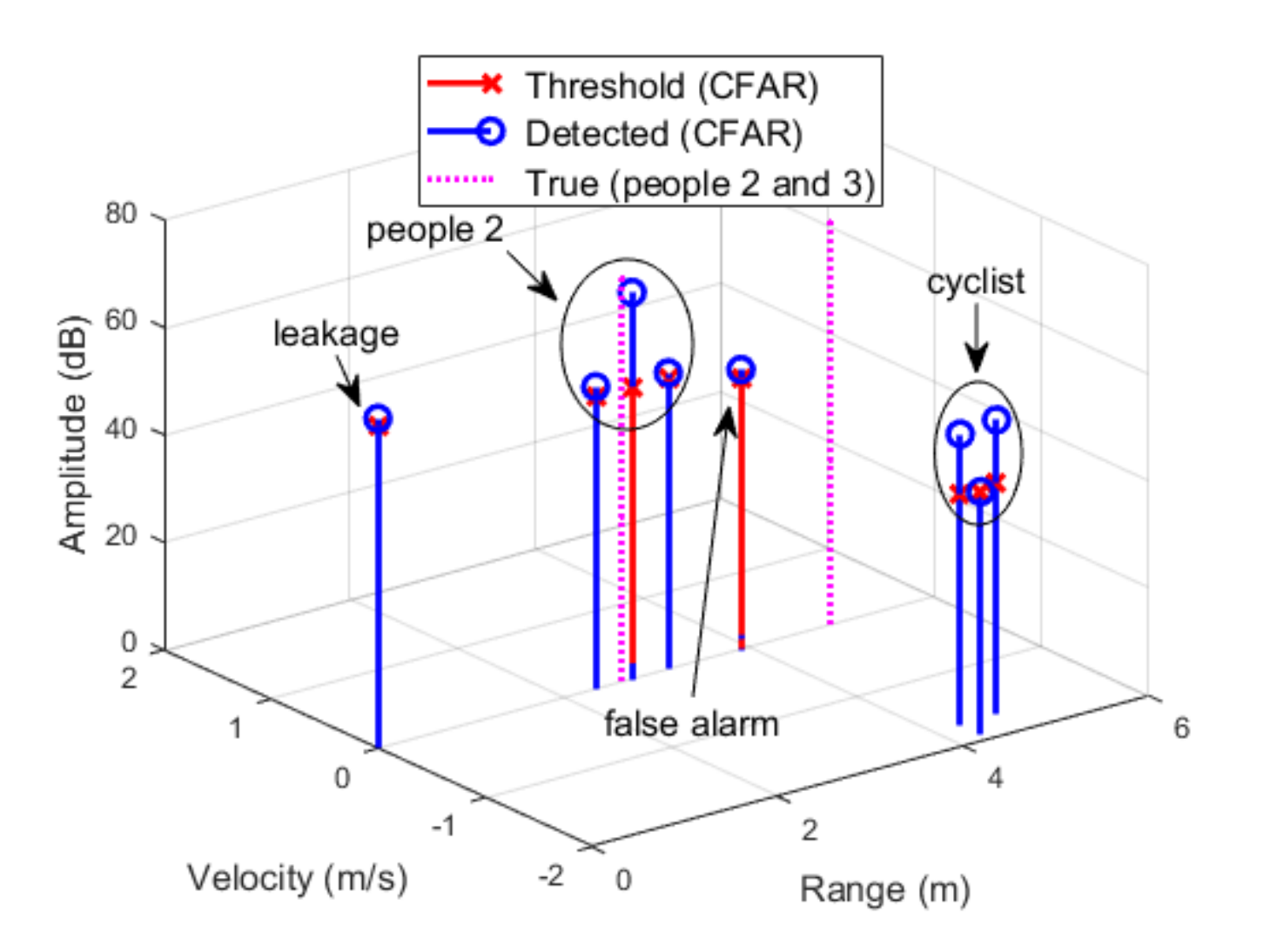}}
  \subfigure[]{
  \label{movingPeopleNOMP}
  \includegraphics[width = 50mm]{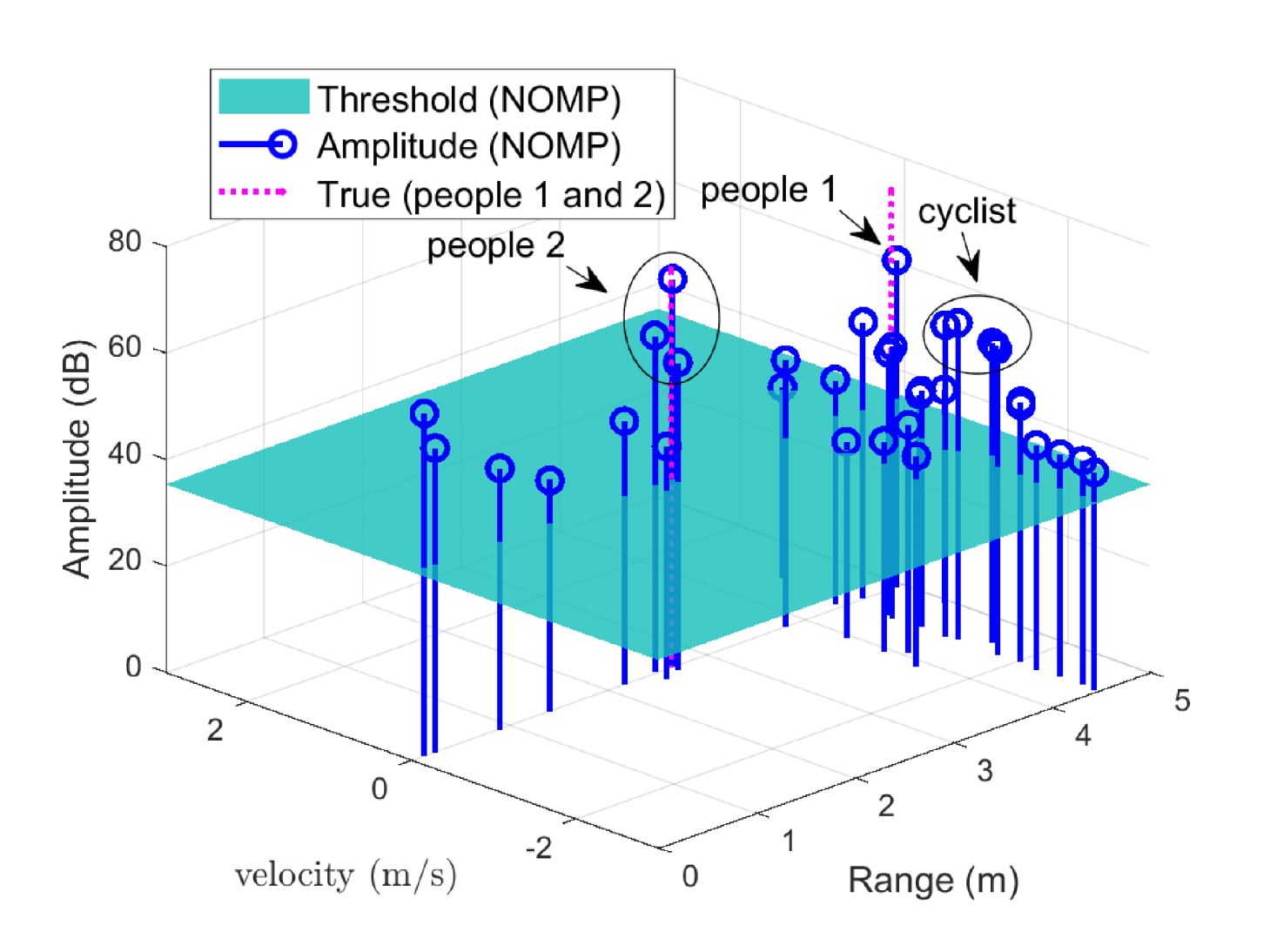}}
  \subfigure[]{
  \label{movingPeopleNOMPCFAR}
  \includegraphics[width = 50mm]{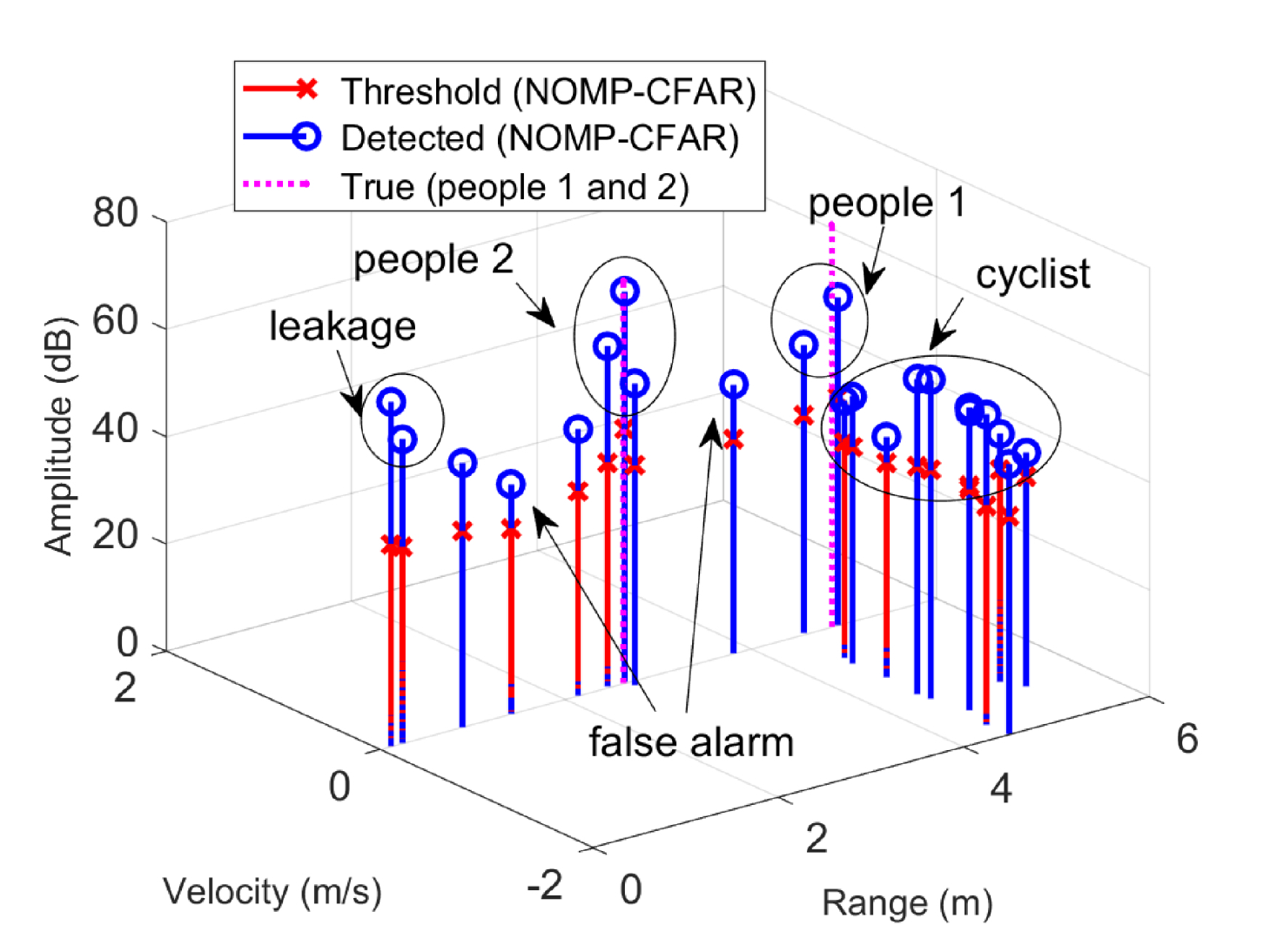}}
 \caption{Range and Doppler estimation results in experiment $2$. (a): CFAR, (b): NOMP, (c): NOMP-CFAR.}
 \label{moving_RangeDopplerResult}
\end{figure*}

\subsection{Experiment $3$}
Fig. \ref{CarScene} shows the setup of field experiment $3$. A people and a cyclist move away from the radar at the radial velocity of about $2.1$ m/s and $1.2$ m/s, and a car moves towards the radar at the radial velocity of about $2.3$ m/s. The algorithm's parameters are the same as that of the experiment $2$.
\begin{figure}
  \centering
  \includegraphics[width=65mm]{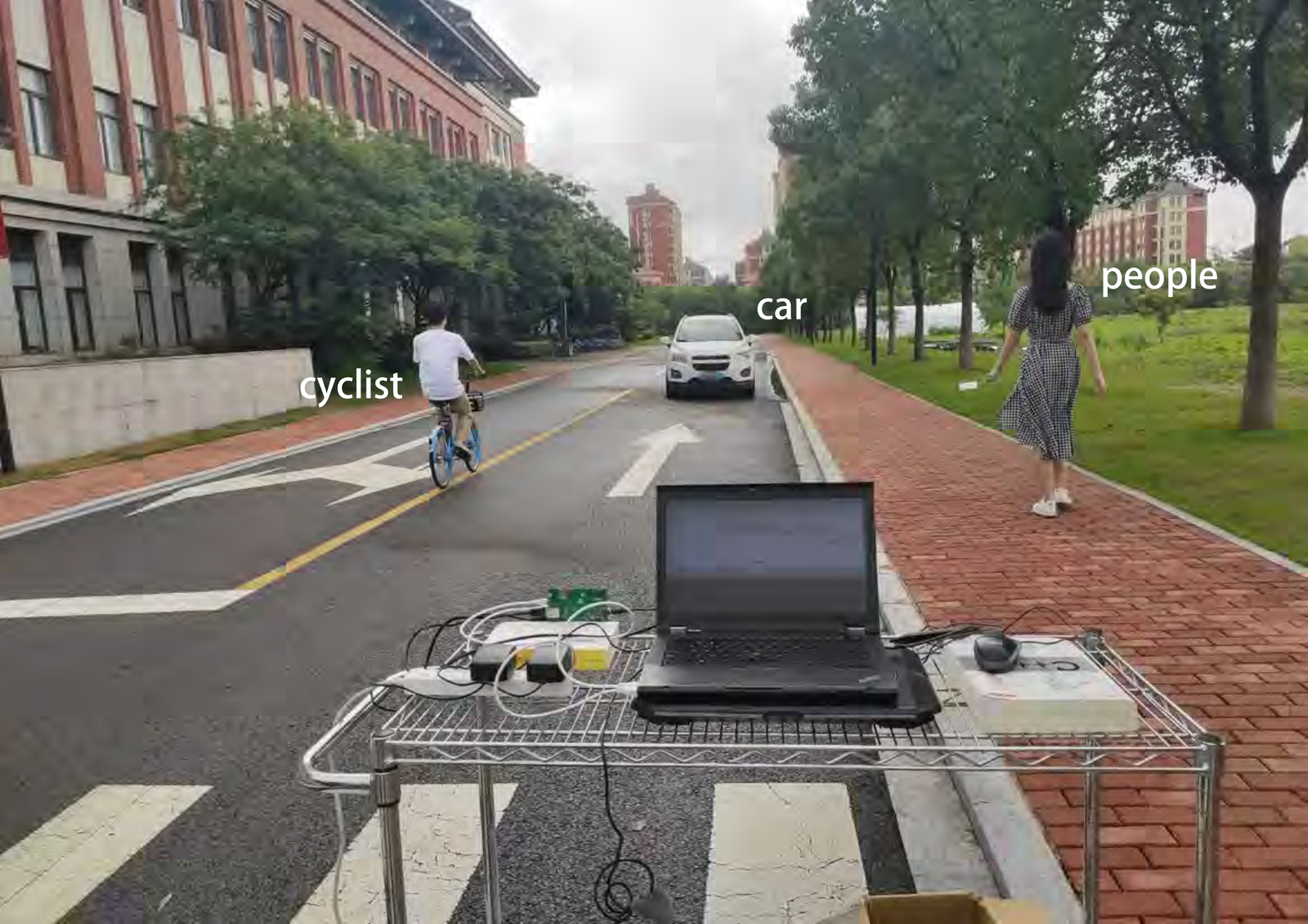}
  \caption{Field experiment 3 setup.}\label{CarScene}
\end{figure}

The range azimuth estimation results are shown in Fig. \ref{exp3car_RangeDOAResult}. The NOMP algorithm detects $28$ targets.
Fig. \ref{exp3car_RangeDOANOMP} shows that the people, cyclist and car are detected by the NOMP algorithm, whose range and azimuth estimation results are $(4.13~{\rm m},	29.93^{\circ})$,
$(5.61~{\rm m},	-29.95^{\circ})$ and
$(19.90~{\rm m},	-0.27^{\circ})$, respectively.
For the NOMP-CFAR algorithm, it detects $10$ targets, including the people, cyclist and car, as shown in Fig. \ref{exp3car_RangeDOANOMPCFAR}. In detail, the range and azimuth estimation results are
$(4.12~{\rm m},	29.73^{\circ})$,
$(5.62~{\rm m}, -29.91^{\circ})$ and
$(19.90~{\rm m},	-0.18^{\circ})$, respectively.
And their corresponding integrated amplitudes and thresholds are $(51.58~{\rm dB},	40.24~{\rm dB})$, $(44.44~{\rm dB},	40.62~{\rm dB})$, and $(45.52~{\rm dB},	39.74~{\rm dB})$, respectively.
\begin{figure*}
\centering
  \subfigure[]{
  \label{exp3car_RangeDOANOMP}
  \includegraphics[width = 50mm]{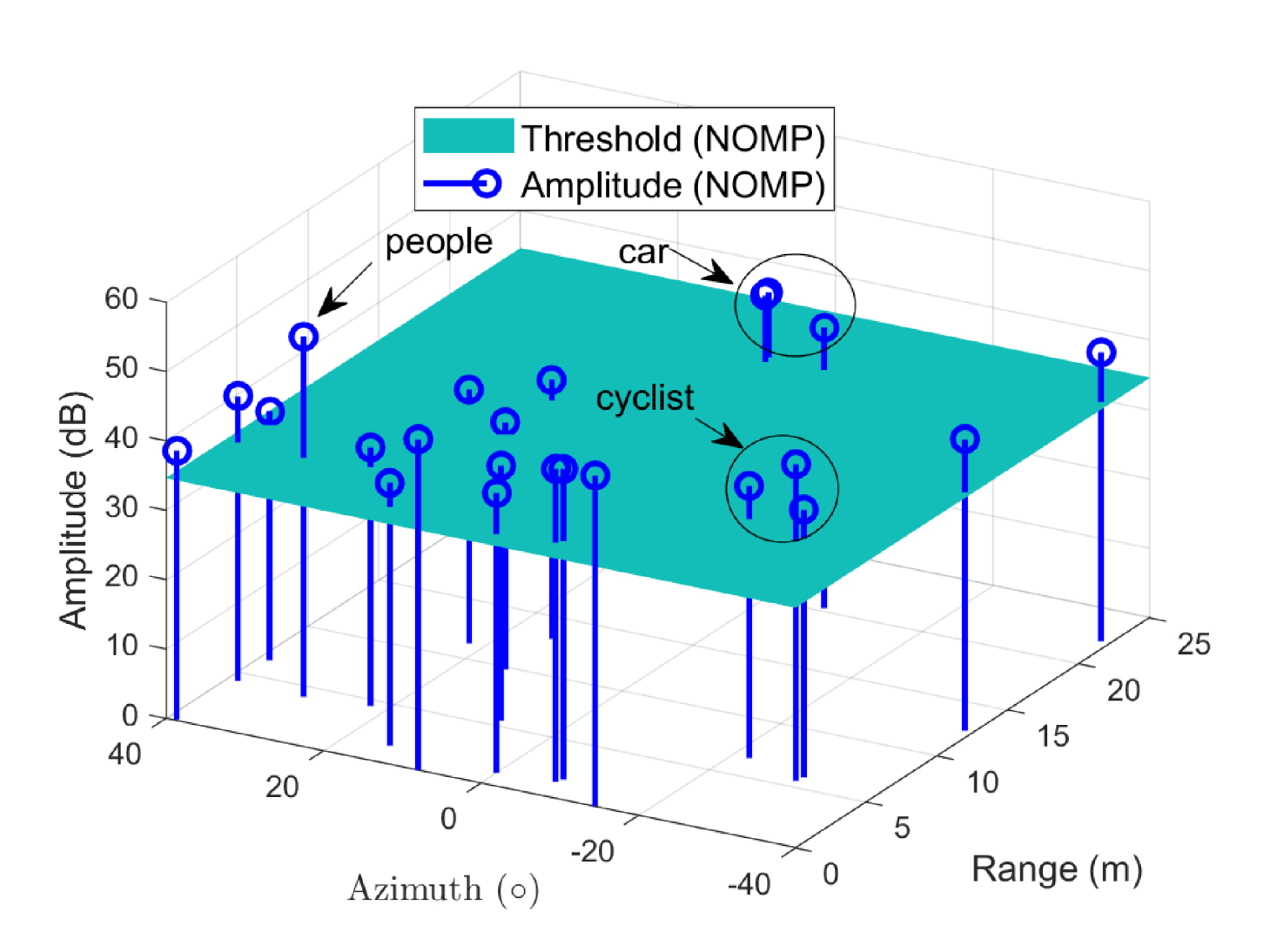}}
  \subfigure[]{
  \label{exp3car_RangeDOANOMPCFAR}
  \includegraphics[width = 50mm]{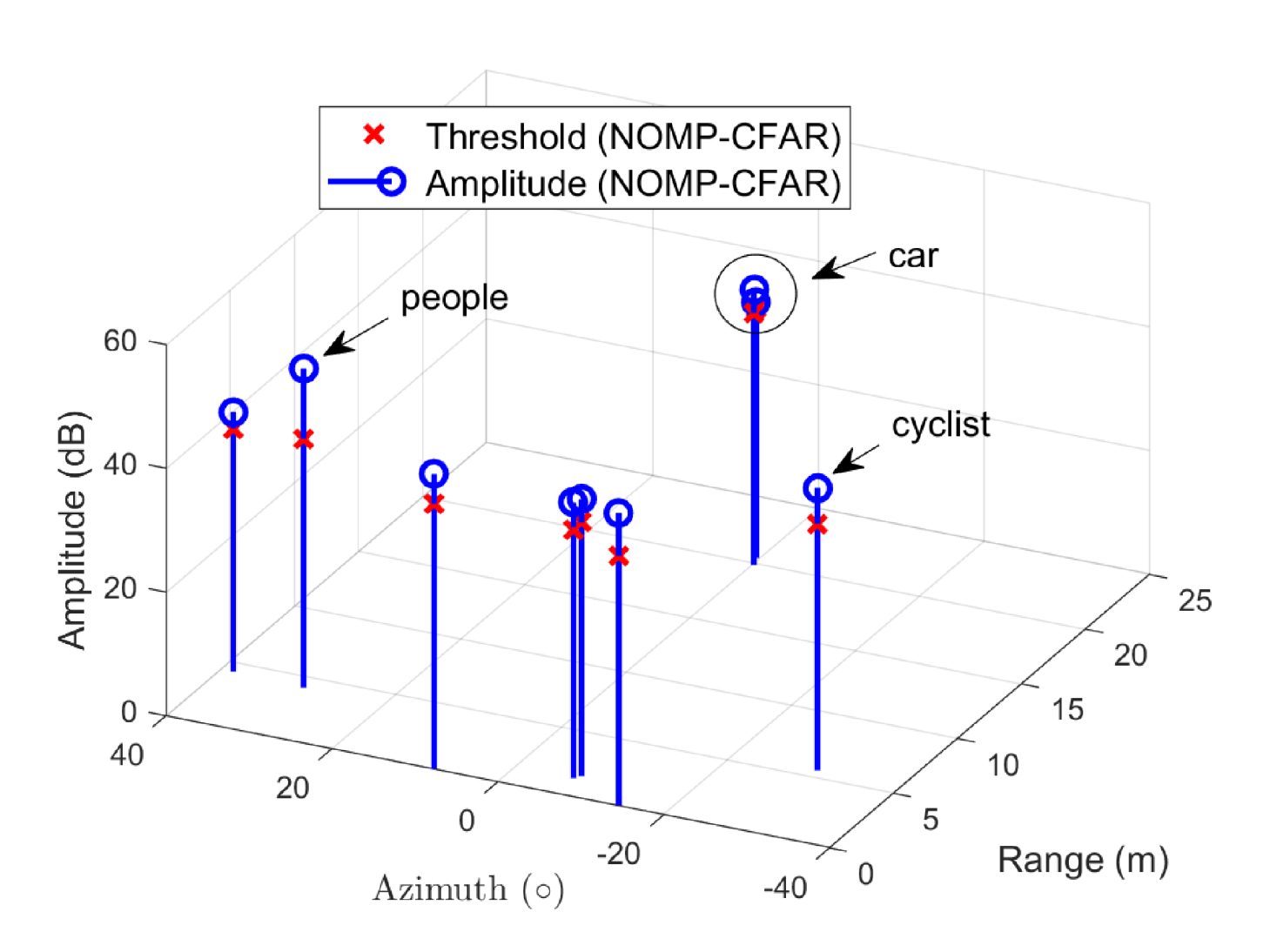}}
 \caption{Range and azimuth estimation results in experiment $3$. (a): NOMP, (b): NOMP-CFAR.}
 \label{exp3car_RangeDOAResult}
\end{figure*}

The range Doppler estimation results are shown in Fig.  \ref{exp3car_RangeDopplerResult}.
The upper bound number of targets for NOMP-CFAR is $K_{\rm max} = 48$;
Fig. \ref{exp3car_RangeDopplerCFAR} shows that the CFAR method detects the people and car but misses the cyclist.
NOMP and NOMP-CFAR detect the people, cyclist and car, and results are shown in Fig. \ref{exp3car_RangeDopplerNOMP} and Fig. \ref{exp3car_RangeDopplerNOMPCFAR}, respectively. Meanwhile, the false alarms generated by NOMP-CFAR is smaller than that of NOMP.

\begin{figure*}
\centering
  \subfigure[]{
  \label{exp3car_RangeDopplerCFAR}
  \includegraphics[width = 50mm]{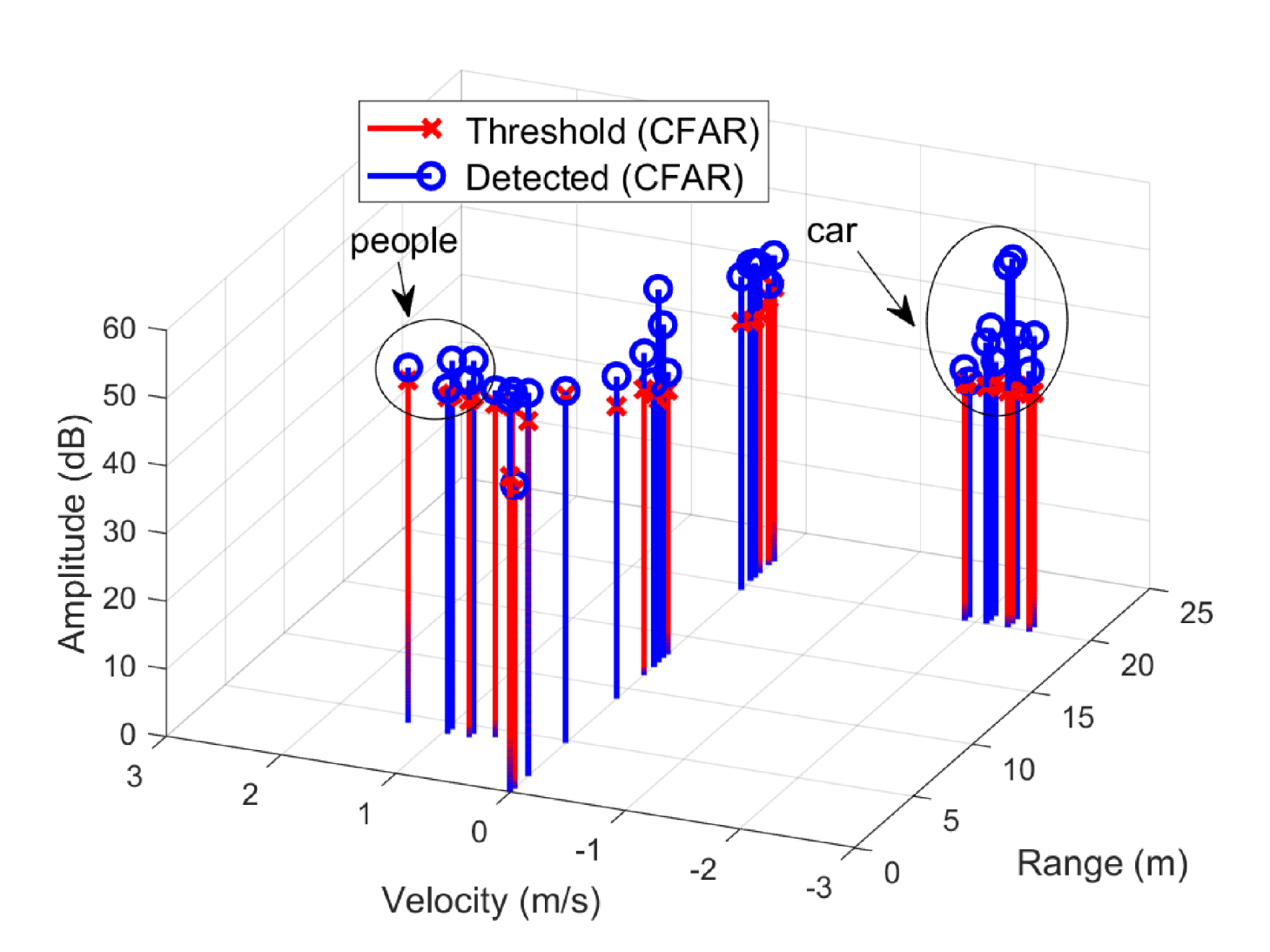}}
  \subfigure[]{
  \label{exp3car_RangeDopplerNOMP}
  \includegraphics[width = 50mm]{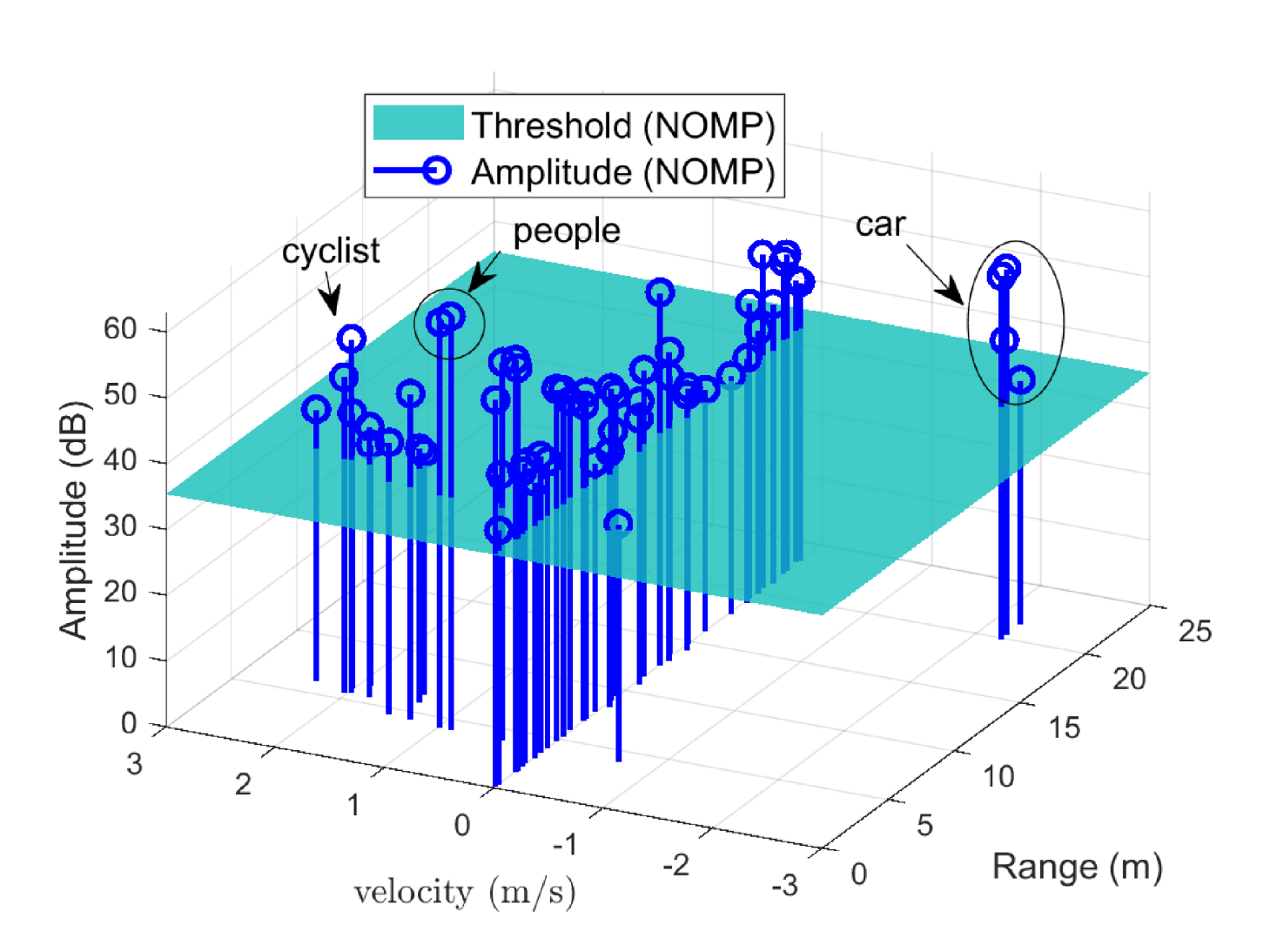}}
  \subfigure[]{
  \label{exp3car_RangeDopplerNOMPCFAR}
  \includegraphics[width = 50mm]{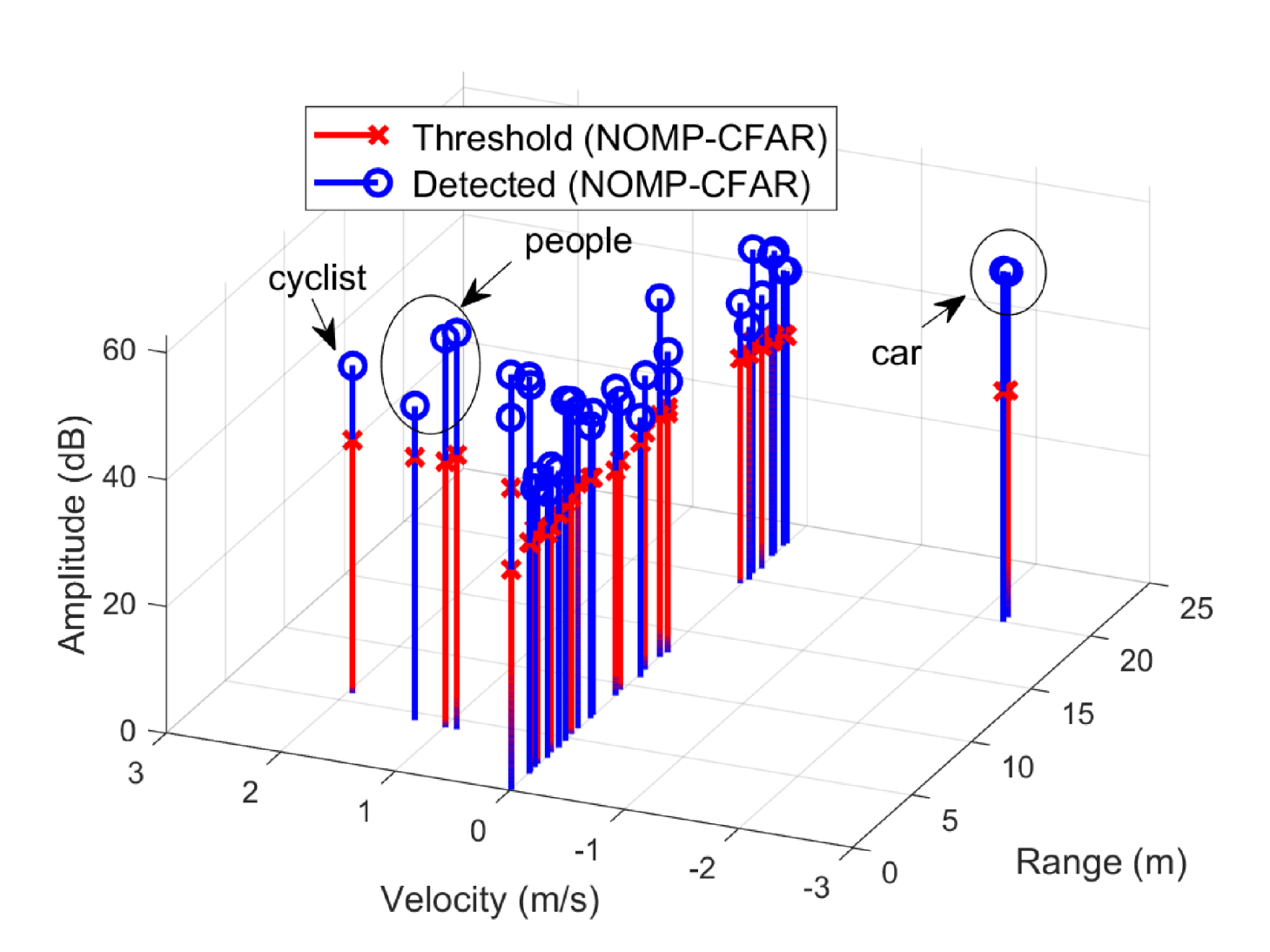}}
 \caption{Range and Doppler estimation results in experiment $3$. (a): CFAR, (b): NOMP, (c): NOMP-CFAR.}
 \label{exp3car_RangeDopplerResult}
\end{figure*}

\section{Conclusion}\label{Conclusion}
We have developed NOMP-CFAR algorithm to achieve the high estimation accuracy and maintain the CFAR behaviour for line spectrum estimation and detection. The algorithm consists of two steps: Initialization and Detection step. For the initialization step, NOMP-CFAR uses the NOMP to provide candidate frequency set in high accuracy, which avoids target masking effects incurred by CFAR. In the Detection step, a soft CFAR detector is introduced to output a quantity which characterizes the confidence of each frequency in the candidate set. A greedy approach is adopted to remove or add one frequency in each iteration, and Newton refinement is then implemented to refine the parameters of the remaining sinusoids. The effectiveness of NOMP-CFAR is verified in substantial numerical experiments and real data.
\section{Appendix}
\subsection{Compute $\bar{\rm P}_{\rm FA}$ for the MMV Scenario}\label{MMVscenario}
The CFAR detector deciding ${\mathcal H}_1$ $T({\mathcal Y})\geq \alpha_{\rm mmv}$ (\ref{GLRT_noiseunknownmmv}) can be equivalently written as
\begin{align}\label{GLRT_noiseunknownmmv}
\frac{\sum\limits_{s=1}^S|\tilde{\mathcal Y}_{{\tilde{\mathbf n}}_{\rm peak}}(s)|^2}{\sigma^2/2}\geq
\frac{\sum\limits_{s=1}^S\sum\limits_{{\tilde{\mathbf n}}\in {\mathcal{T}}_{{\tilde{\mathbf n}}_{\rm peak}}} \left|\tilde{\mathcal Y}_{\tilde{\mathbf n}}(s)\right|^2}{\sigma^2/2} \frac{\alpha_{\rm mmv}}{N_r}\triangleq \bar{T} \frac{\alpha_{\rm mmv}}{N_r}.
\end{align}
Under the null hypothesis, $\frac{\sum\limits_{s=1}^S|\tilde{\mathcal Y}_{{\tilde{\mathbf n}}}(s)|^2}{\sigma^2/2}$ follows a chi-squared distribution with a degrees of freedom $2S$, i.e., $\frac{\sum\limits_{s=1}^S|\tilde{\mathcal Y}_{{\tilde{\mathbf n}}}(s)|^2}{\sigma^2/2}\sim \chi^2(2S)$. Conditioned on $\bar{T}$, the peak of $\frac{\sum\limits_{s=1}^S|\tilde{\mathcal Y}_{{\tilde{\mathbf n}}}(s)|^2}{\sigma^2/2}$ exceeding $\bar{T} \frac{\alpha_{\rm mmv}}{N_r}$ is
\begin{align}
{\rm P}\left(\frac{\sum\limits_{s=1}^S|\tilde{\mathcal Y}_{{\tilde{\mathbf n}}_{\rm peak}}(s)|^2}{\sigma^2/2}\geq\bar{T} \frac{\alpha_{\rm mmv}}{N_r}\right)=1-F_{2S}^N\left(\bar{T} \frac{\alpha_{\rm mmv}}{N_r}\right),
\end{align}
where $N$ denotes the total number of cells, $F_{2S}(x)$ denotes the cumulative distribution function (CDF) of the chi-squared distribution
\begin{align}
F_{2S}(x)=\begin{cases}
&1-{\rm e}^{-\frac{x}{2}}\sum\limits_{s=0}^{S-1}\frac{1}{s!}\left(\frac{x}{2}\right)^{s}, x>0,\\
&0,\quad\quad\quad\quad {\rm otherwise}.
\end{cases}
\end{align}
In addition, $\bar{T}$ follows
\begin{align}\label{chi2def}
\bar{T}\sim \chi^2(2SN_r)=\begin{cases}
&\frac{1}{2(SN_r - 1)!} \mathrm{e}^{-\frac{x}{2}} \left(\frac{x}{2} \right)^{SN_r - 1},x>0,\\
&0,\quad\quad\quad\quad\quad\quad {\rm otherwise}.
\end{cases}.
\end{align}
Averaging ${\rm P}\left(\frac{\sum\limits_{s=1}^S|\tilde{\mathcal Y}_{{\tilde{\mathbf n}}_{\rm peak}}(s)|^2}{\sigma^2/2}\geq\frac{\alpha_{\rm mmv}}{N_r}\bar{T}\right)$ over $\bar{T}$ yields the false alarm
\begin{align}\label{PbarFAmmv}
\bar{\rm P}_{\rm FA}=1-{\rm E}_{\bar{T}\sim \chi^2(2SN_r)}\left[F_{2S}^N\left(\bar{T} \frac{\alpha_{\rm mmv}}{N_r}\right)\right].
\end{align}
Substituting (\ref{chi2def}) in (\ref{PbarFAmmv}) yields (\ref{simple type of Poe}).

\bibliographystyle{IEEEbib}
\bibliography{strings,refs}

\end{document}